\documentclass[11pt,a4paper]{article}
\usepackage{amsfonts}

\usepackage{amsmath}
\usepackage{amssymb}
\usepackage{times}
\usepackage{bm}
\RequirePackage[round,authoryear,longnamesfirst]{natbib}
\usepackage{epstopdf}
\usepackage{color}
\usepackage[plain,noend,boxed]{algorithm2e}
\usepackage{xr}
\usepackage{amsthm}
\usepackage{url}
\usepackage{adjustbox}
\usepackage{hyperref}
\hypersetup{citecolor=true,colorlinks=true,allcolors=blue}

\newtheorem{lemma}{Lemma}[section]
\newtheorem{theorem}{Theorem}[section]

\newtheorem{assumption}{Assumption}

\DeclareMathOperator*{\argmax}{arg\,max}
\DeclareMathOperator*{\argmin}{arg\,min}

\newtheorem{definition}{Definition}[section]

\numberwithin{equation}{section}
\begin{document}

\title{Determining the Number of Communities in Degree-corrected Stochastic
Block Models }
\author{Shujie Ma\thanks{Department of Statistics, University of California, Riverside. \ E-mail~address: shujie.ma@ucr.edu} \and Liangjun Su\thanks{
		Singapore Management University.\ E-mail~address: ljsu@smu.edu.sg.} \and Yichong Zhang 	\thanks{
		Singapore Management University.\ E-mail~address: yczhang@smu.edu.sg. The
		corresponding author.}}

\maketitle

\begin{abstract}
We propose to estimate the number of communities in degree-corrected
stochastic block models based on a pseudo likelihood ratio statistic. To
this end, we introduce a method that combines spectral clustering with
binary segmentation. This approach guarantees an upper bound for the pseudo
likelihood ratio statistic when the model is over-fitted. We also derive its
limiting distribution when the model is under-fitted. Based on these
properties, we establish the consistency of our estimator for the true
number of communities. Developing these theoretical properties require a
mild condition on the average degrees -- growing at a rate no slower than $%
\log (n)$, where $n$ is the number of nodes. Our proposed method is further
illustrated by simulation studies and analysis of real-world networks. The
numerical results show that our approach has satisfactory performance when
the network is semi-dense.

\vspace{2mm}
\end{abstract}



\markboth{S. Ma et~al.}{DCSBM K}


%
%

\section{Introduction}

Advances in modern technology have facilitated the collection of network
data which emerge in many fields including biology, bioinformatics, physics,
economics, sociology and so forth. Therefore, developing effective analytic
tools for network data has become a focal area in statistics research over
the past decade. Network data often have natural communities which are
groups of interacting objects (i.e., nodes); pairs of nodes in the same
group tend to interact more often than pairs belonging to different groups.
For example, in social networks, communities can be groups of people who
belong to the same club, be of the same profession, or attend the same
school; in protein-protein interaction networks, communities are regulatory
modules of interacting proteins. In many cases, however, the underlying
structure of network data is not directly observable. In such cases, we need
to infer the latent community structure of nodes from knowledge of their
interaction patterns.

The stochastic block model (SBM) proposed by \cite{HLL83} is a random graph
model tailored for clustering nodes, and it is commonly used for recovering
the community structure in network data. SBM has one limitation: it assumes
that all nodes in the same community are stochastically equivalent (i.e.,
they have the same expected degrees). To overcome this limitation, \cite%
{KN11} propose the degree-corrected stochastic block model (DCSBM) which
allows for degree heterogeneity within communities. In the literature,
various methods have been proposed for the estimation of SBM and DCSBM. They
include but are not limited to modularity maximization \citep{NG04},
likelihood-based methods \citep{ACBL13,BC09,CWA12,ZLZ12}, the method of
moments \citep{BCL11}, spectral clustering %
\citep{J15,JY16,LR15,QR13,RCY11,SB15,SWZ17}, and spectral embedding %
\citep{LSTAP14,STFP12}. In most, if not all, works, theoretical properties
such as consistency and asymptotic distributions are built based on the
assumption that the true number of communities $K_{0}$ is known.

In practice, prior information of the number of communities is often
unavailable. Accurately estimating $K_{0}$ from the network data is of
crucial importance, as the following community detection procedure relies
upon it. Determining the number of communities can be regarded as a model
selection problem. A natural approach to the problem is to consider the
popular model selection methods such as cross-validation (CV)\ or
likelihood-based methods. However, tailoring those methods for SBMs or
DCSBMs and establishing the theoretical support are challenging, as network
data are complex in nature.

A few methods have been developed to estimate $K_{0}$. Among them, the
eigenvalue-based methods have been widely applied; see \cite{BS16}, \cite%
{B15}, \cite{LL15} and \cite{Lei16} for the hypothesis testing methods on
eigenvalues. These methods can be computationally fast, but they only use
partial information from the data -- the eigenvalues. Empirically, the good
behavior of eigenvalues often requires a very large sample size. In order to
make use of all the information from the data, we need to estimate the graph
model (SBM or DCSBM). To this end, spectral clustering is considered as a
quick and effective way, and it has been proven to have reliable theoretical
basis \citep{J15,JY16,LR15,QR13,RCY11,SB15,SWZ17}. Based on the spectral
clustering method for estimating the graph model, \cite{CL17} and \cite%
{LLZ16} propose network cross-validation (NCV) and edge cross-validation
(ECV), respectively, for selecting the number of communities. In particular,
\cite{CL17} show that the NCV method guarantees against under-selection in
SBMs, but it does not rule out possible over-selection. Although they have a
discussion on the estimation of DCSBMs, they do not study the theoretical
property of the NCV estimator of the number of communities ($K$) in DCSBMs.
\cite{LLZ16} propose an ECV method for choosing between SBMs and DCSBMs
along with selecting $K$ for each model, but the consistency of ECV is not
established. Moreover, both methods can be computationally intensive when
the number of folds is large; they can lead to unstable results when the
number of folds or the number of random sample splittings (or repetitions in
the ECV case) is small. Another appealing method for model selection is the
likelihood-based approach considered in \cite{WB16}. It uses a BIC-type
penalty, so that it avoids iterations or random sample splittings. However,
for either SBMs or DCSBMs, optimizing the likelihood function which involves
summing over all possible community memberships is computationally
intractable for even moderate sample sizes. As a result, \cite{WB16} use a
variational EM algorithm to approximate the likelihood.

In this article, we propose a new method by taking advantage of both
spectral clustering and likelihood principle. The method is devised for
DCSBM, but can be naturally applied to SBM as it is a special case of DCSBM.
To determine the number of communities $K$, we propose a pseudo likelihood
ratio (pseudo-LR) to compare the goodness-of-fit of two DCSBMs estimated by
using $K$ and $K+1$, respectively, as the number of communities. For
estimation, directly using spectral clustering can be an appealing choice as
it is computationally fast. However, when $K>K_{0}$, it remains unclear
about theoretical properties for the resulting estimators of the DCSBM
obtained through the standard spectral clustering approach. This hinders the
use of goodness-of-fit methods for model selection by spectral clustering
for estimation. To overcome the difficulty, we estimate the DCSBM with $K$
communities by spectral clustering; based on this estimate, we propose a
binary segmentation method for estimating the DCSBM with $K+1$ communities.\
This approach guarantees consistency of the estimator for the model with $K+1
$ communities when the estimator for the model with $K$ communities is
consistent. The binary segmentation technique has been used in the seminal
work \cite{V81} for change-point detection and in recent work \cite{SW17}
for latent group recovery. Our idea of adapting this method to estimate
DCSBM has not been considered by others. Based on the proposed estimation
approach, we show that the pseudo-LR has a sound theoretical basis, and the
resulting estimator of the number of communities is consistent.

It is worth noting that for establishing the consistency of estimating $K_{0}
$, we only require the average degree to grow with the number of nodes $n$
at a rate no slower than $\log (n)$, whereas \cite{WB16} need it to be
faster than $n^{1/2}\log (n)$ in DCSBMs. That is, the approach considered in
\cite{WB16} needs a much denser network than our method for good finite
sample performance. As pointed out by \citet[Section 2.5]{WB16}, their
approach needs a very stringent condition on the average degree, because the
slow convergence rate of the estimate of the node degree variation passes on
to the likelihood ratio. On the contrary, it is not carried on to our
pseudo-LR because of the mutual cancellation of the slow-convergence parts.
As a result, this allows us to relax the strong restriction on the average
degree in theory. Both \cite{CL17} and \cite{LLZ16} only require the growth
rate of the average degree to be no slower than $\log (n)$, which is the
same rate as required by our method. However, theoretical properties are not
available for the NCV and ECV estimators of $K$ in DCSBMs. In contrast, we
develop thorough theoretical results including the consistency of our
proposed pseudo-LR method.

The rest of the paper is organized as follows. We describe the estimation
procedure in Section \ref{sec:dc}. We establish the consistency of our
estimators of the number of communities under DCSBMs in Section \ref%
{sec:theory}. Section \ref{sec:sim} compares the performance of our method
with various existing methods in different simulated networks. Section \ref%
{sec:real} illustrates the proposed method using several real data examples.
Section \ref{sec:concl} concludes. The proofs of all results are relegated
to the Supplemental Materials.

Notation. Throughout the paper, we write $[M]_{ij}$ as the $(i,j)$-th entry
of matrix $M$. Without confusion, we sometimes simplify $[M]_{ij}$ as $M_{ij}
$. In addition, we write $[M]_{i}$ as the $i$-th row of $M$. $\Vert M\Vert $
and $\Vert M\Vert _{F}$ denote the spectral norm and Frobenius norm of $M,$
respectively. Note that $\Vert M\Vert =\Vert M\Vert _{F}$ when $M$ is a
vector. We use $\mathbf{1}\left\{ \cdot \right\} $ to denote the indicator
function which takes value 1 when $\cdot $ holds and 0 otherwise. All
vectors without transpose are understood as column vectors. For a vector $%
\boldsymbol{a}=(a_{1},...,a_{n})^{\top }$, let diag$(\boldsymbol{a})$ be the
diagonal matrix whose diagonal is $\boldsymbol{a}$, and let $||\boldsymbol{a}%
||=(\sum\nolimits_{i}a_{i}^{2})^{1/2}$ be its L$_{2}$ norm. Let $\iota _{n}$%
, $\#\mathcal{S}$, and $[n]$ be the $n$-dimensional vector of ones, the
cardinality of set $\mathcal{S}$, and the integer sequence $\{1,2,\cdots ,n\}
$, respectively. $C$, $c$, and $c^{\prime }$ denote arbitrary positive
constants that are independent of $n$, but may not be the same in different
contexts.

\section{Methodology}

\label{sec:dc}

\subsection{Degree-corrected SBM}

Let $A\in \{0,1\}^{n\times n}$ be the adjacency matrix. By convention, we do
not allow self-connection, i.e., $A_{ii}=0$. The network is generated by a
degree-corrected stochastic block model with $K_{0}$ true communities. The
communities, which represent a partition of the $n$ nodes, are assumed to be
fixed beforehand. Denote $Z_{K_{0}}=\{[Z_{K_{0}}]_{ik}\}$ as the $n\times
K_{0}$ binary matrix providing the true cluster memberships of each node,
i.e., $[Z_{K_{0}}]_{ik}=1$ if node $i$ is in $\mathcal{C}_{k,K_{0}}$ and $%
[Z_{K_{0}}]_{ik}=0$ otherwise, where $\mathcal{C}_{1,K_{0}},\ldots ,\mathcal{%
C}_{K_{0},K_{0}}$ are denoted as the communities identified by $Z_{K_{0}}$.
For $k=1,\cdots ,K_{0},$ let $n_{k,K_{0}}=\#\mathcal{C}_{k,K_{0}},$ the
number of nodes in $\mathcal{C}_{k,K_{0}}$. Given the $K_{0}$ communities,
the edges between nodes $i$ and $j$ are chosen independently with
probability depending on the communities that nodes $i$ and $j$ belong to.
In particular, for nodes $i$ and $j$ belonging to clusters $\mathcal{C}%
_{k,K_{0}}$ and $\mathcal{C}_{l,K_{0}}$, respectively, the probability of
edge between $i$ and $j$ is given by
\begin{equation*}
P_{ij}=E(A_{ij})=\theta _{i}\theta _{j}B_{kl,K_{0}},
\end{equation*}%
where the block probability matrix $B_{K_{0}}=\{B_{kl,K_{0}}\}$, $%
k,l=1,\ldots ,K_{0}$, is a symmetric matrix with each entry between $(0,1]$.
The $n\times n$ edge probability matrix $P=\{P_{ij}\}$ represents the
population counterpart of the adjacency matrix $A$. Let $\Theta =\text{diag}%
(\theta _{1},\ldots ,\theta _{n})$. Then we have
\begin{equation*}
P=E(A)=\Theta Z_{K_{0}}B_{K_{0}}Z_{K_{0}}^{T}\Theta ^{T}.
\end{equation*}%
Note that $\Theta $ and $B_{K_{0}}$ are only identifiable up to scale.
Following the lead of \citet[Theorem 3.3]{SWZ17}, we adopt the following
normalization rule:
\begin{equation}
\sum_{i\in \mathcal{C}_{k,K_{0}}}\theta _{i}=n_{k,K_{0}},\quad k=1,\ldots
,K_{0}.  \label{eq:thetanormalization}
\end{equation}%
Apparently, the DCSBM becomes the standard SBM when $\theta _{i}=1$ for each
$i=1,...,n.$

\subsection{Estimation of the number of communities}

\label{sec:algo} Our procedure of estimating $K_{0}$ requires to obtain two
estimated membership matrices {$(\hat{Z}_{K},\hat{Z}_{K+1}^{b})$} based on $K
$ and $K+1$ communities, respectively.\footnote{
The superscript $b$ in $\hat{Z}_{K+1}^{b}$ denotes that it is estimated by a
binary segmentation from $\hat{Z}_{K}$.} To this end, we estimate $\hat{Z}%
_{K}$ and $\hat{Z}_{K+1}^{b}$ via spectral clustering of the first $K$
eigenvectors of the graph Laplacian and a binary segmentation technique,
respectively. Section \ref{sec:membership} provides more details. Denote $%
\hat{P}_{ij}(Z)$ as the estimator of $P_{ij}$ for a given membership matrix $%
Z$. We compute $\hat{P}_{ij}(\hat{Z}_{K+1}^{b})$ and $\hat{P}_{ij}(\hat{Z}%
_{K})$ by the sample-frequency-type estimators and propose a pseudo-LR $%
L_{n}(\hat{Z}_{K+1}^{b},\hat{Z}_{K})$ defined in \eqref{eq:plr} to measure
the deviance of goodness-of-fit of DCSBMs estimated with $K$ and $K+1$
communities, respectively. The estimators of $\hat{P}_{ij}(\hat{Z}_{K+1}^{b})
$ and $\hat{P}_{ij}(\hat{Z}_{K})$ are given in Section \ref{sec:appago} of
the Supplemental Materials. Lastly, we obtain the estimator of the true
number of communities based on the change of the pseudo-LR. Let $K_{\max }$
denote the maximum number of communities such that $K_{\max }\geq K_{0}.$
The pseudo-code is described in Algorithm \ref{algo:1}.


To understand our algorithm of estimating $K_{0},$ we focus on the case
where $K_{0}\geq 2$. If we know that $K_{0}\geq 2$ for sure$,$ we can
redefine $\hat{K}_{1}=\argmin_{2\leq K\leq K_{\max }}R(K).$ By Theorems \ref%
{thm:underDC} and \ref{thm:overDC} in Section \ref{sec:over}, we have
\begin{equation*}
L_{n}(\hat{Z}_{K}^{b},\hat{Z}_{K-1})\asymp n^{2}\text{ for }2\leq K\leq K_{0}%
\text{ and }L_{n}(\hat{Z}_{K_{0}+1}^{b},\hat{Z}_{K_{0}})\leq O_{a.s.}(n\rho
_{n}^{-1}),
\end{equation*}%
where $a_{n}\asymp b_{n}$ means that $P(c\leq a_{n}/b_{n}\leq C)\rightarrow 1
$ as $n\rightarrow \infty $ for some positive constants $c$ and $C,$ $a.s.$
denotes almost surely, and the parameter $\rho _{n}$ characterizes the
sparsity of the network such that $n\rho _{n}/\log (n)$ is sufficiently
large (see Assumption \ref{ass:rate} in Section \ref{sec:est}). This result
directly implies that
\begin{equation*}
R\left( K\right) \asymp 1\text{ for }2\leq K<K_{0}\text{ and }R\left(
K_{0}\right) =o_{p}\left( 1\right) .
\end{equation*}%
The above results indicate that for $K=K_{0}$, $R\left( K\right) $ is very
small and close to zero, but for $K<K_{0}$, $R\left( K\right) $ is
relatively large. It is worth noting that for $K>K_{0}$, it is possible that
$R\left( K\right) \ $is also small. As a result, the minimizer of $R(K)$ is
only guaranteed to satisfy $\hat{K}_{1}\geq K_{0}$ with probability
approaching 1 (w.p.a.1) as $n\rightarrow \infty .$ Such a result is similar
to that in \cite{CL17} who show that NCV do not underestimate the number of
communities w.p.a.1 as $n\rightarrow \infty $. Based on our theory, we
expect to observe a gap of the values of $R\left( K\right) $ at $K=K_{0}$,
so we introduce $\tilde{K}_{2}$ which is the first $K$ such that $R(K)$ is
less than $h_{n}$, where $h_{n}\rightarrow 0$ and $n\rho
_{n}h_{n}\rightarrow \infty $. Then we have $\tilde{K}_{2}=K_{0}$ w.p.a.1 as
$n\rightarrow \infty $. For better numerical performance, we make use of
both $\hat{K}_{1}$ and $\tilde{K}_{2}$ by letting $\hat{K}_{2}=\min (\hat{K}%
_{1},\tilde{K}_{2})$, and thus it satisfies $P(\hat{K}_{2}=K_{0})\rightarrow
1$ as $n\rightarrow \infty $, i.e., $\hat{K}_{2}$ consistently estimates the
number of communities in large samples. In our algorithm, two tuning
parameters $c_{\eta }$ and $h_{n}$ are involved. Among them, $c_{\eta }$ is
only needed to
deal with the case $K=1$ in which the pseudo-LR cannot be defined. If we are sure that $%
K_{0}\geq 2$, i.e., there are more than one communities, we can obtain the estimate $\hat{K}_{1}$ by searching
over $K\in \left[ 2,K_{\max }\right] .$ Alternatively, one can separately
test $K_{0}=1$ using other methods, e.g., the eigenvalue-based test proposed
by \cite{BS16}, and then use our methods to select $K$ for $K\geq 2.$ In
both cases, one can avoid the use of $c_{\eta}$. Theoretically, $c_{\eta }$ only needs to satisfy $%
c_{\eta }\in (0,\infty )$. Practically, We choose a value for $c_{\eta }$
given in Section \ref{sec:results} that works well in our numerical
analysis. For the choice of $h_{n}$, we have a detailed discussion given
after Theorem \ref{thm:overDC} in Section \ref{sec:over}.

\IncMargin{1em}
\begin{algorithm}[H]
	\label{algo:1}
\caption{Estimation of the number of communities}
	\SetKwInOut{Input}{input}\SetKwInOut{Output}{output}
	\Input{adjacency matrix $A$, tuning parameters $c_\eta$ and $h_n$}
	\Output{$\hat{K}_1$ and $\hat{K}_2$}
	\For{$K \leftarrow 1$ \KwTo $K_{\max}$}{
		obtain $\hat{Z}_K$ and $\hat{Z}_{K+1}^b$ via spectral clustering and binary segmentation, respectively\;
		compute $\hat{P}_{ij}(\hat{Z}_{K})$ and $\hat{P}_{ij}(\hat{Z}_{K+1}^{b})$\;
		compute
		\begin{equation}
		L_{n}(\hat{Z}_{K+1}^{b},\hat{Z}_{K})=\frac{1}{2}\sum_{i\neq j}\biggl(\frac{%
			\hat{P}_{ij}(\hat{Z}_{K+1}^{b})}{\hat{P}_{ij}(\hat{Z}_{K})}-1\biggr)^{2}\;
		\label{eq:plr}
		\end{equation}
		compute $R(K)$ as
		\begin{equation}
		\label{eq:RK}
		R(K)=%
		\begin{cases}
		\frac{L_{n}(\hat{Z}_{K+1}^{b},\hat{Z}_{K})}{\eta _{n}} & \quad K=1 \\
		\frac{L_{n}(\hat{Z}_{K+1}^{b},\hat{Z}_{K})}{L_{n}(\hat{Z}_{K}^{b},\hat{Z}%
			_{K-1})} & \quad K\geq 2,%
		\end{cases}%
		\;
		\end{equation}%
        where $\eta _{n}=c_{\eta }n^{2}$.}
	obtain $\hat{K}_1$ and $\hat{K}_2$ as
	\begin{equation*}
	\hat{K}_{1}=\argmin_{1\leq K\leq K_{\max }}R(K),
	\end{equation*}%
	and
	\begin{equation*}
	\hat{K}_{2}=\min (\hat{K}_{1},\tilde{K}_{2}),
	\end{equation*}%
	where $\tilde{K}_{2}=\min \{K\in \{1,\cdots ,K_{\max }\},R(K)\leq h_{n}\}$
	if $\min_{1\leq K\leq K_{\max }}R(K)\leq h_{n}$ and $\tilde{K}_{2}=K_{\max }$
	otherwise.
\end{algorithm}\DecMargin{1em}\medskip

\subsection{Estimation of the memberships\label{sec:membership}}

The proposed pseudo-LR given in (\ref{eq:plr}) depends on {$( \hat{Z}_{K},%
\hat{Z}_{K+1}^{b})$} which are obtained through spectral clustering and
binary segmentation, respectively. In the following, we describe the
algorithm in detail. Let $\hat{d}_{i}=\sum_{j=1}^{n}A_{ij}$ denote the
degree of node $i$, $D=\text{diag}(\hat{d}_{1},\ldots ,\hat{d}_{n})$. We
regularize the degree for each node as $\hat{d}_{i}^{\tau }=\hat{d}_{i}+\tau
$ where $\tau $ is a regularization parameter. Let $D_{\tau }=\text{diag}(%
\hat{d}_{1}+\tau ,\ldots ,\hat{d}_{n}+\tau )$. The regularized sample graph
Laplacian is
\begin{equation*}
L_{\tau }=D_{\tau }^{-1/2}AD_{\tau }^{-1/2}.
\end{equation*}
We regularize the sample degree matrix $D$ to improve the finite sample
performance of spectral clustering. The same regularization strategy is
considered by \cite{RCY11}, \cite{JY16} and \cite{SWZ17}. The corresponding
theoretical property is established in Section \ref{sec:theory}.

Denote the spectral decomposition of $L_{\tau }$ as
\begin{equation*}
L_{\tau }=\widehat{U}_{n}\widehat{\Sigma }_{n}\widehat{U}_{n}^{T},
\end{equation*}%
where $\widehat{\Sigma }_{n}=\text{diag}(\hat{\sigma}_{1n},\ldots ,\hat{%
\sigma}_{nn})$ with $|\hat{\sigma}_{1n}|\geq |\hat{\sigma}_{2n}|\geq \cdots
\geq |\hat{\sigma}_{nn}|\geq 0,$ and $\widehat{U}_{n}$ is the corresponding
eigenvectors such that $\widehat{U}_{n}^{T}\widehat{U}_{n}=I_{n}$. For each $%
K=1,\cdots ,K_{\max }$, let
\begin{equation}
\hat{\nu}_{iK}=\frac{\hat{u}_{i}(K)}{||\hat{u}_{i}(K)||},  \label{eq:nuK}
\end{equation}%
where $\hat{u}_{i}^{T}$ is the $i$-th row of $\widehat{U}_{n}$ and $\hat{u}%
_{i}^{T}(K)$ collects the first $K$ elements of $\hat{u}_{i}^{T}$. We
estimate the pair of community memberships {{$(\hat{Z}_{K},\hat{Z}_{K+1}^{b})
$}} by the following algorithm.

\IncMargin{1em}
\begin{algorithm}[H]
	\label{algo:2}
	\SetKwInOut{Input}{input}\SetKwInOut{Output}{output}
	\Input{$\{\hat{\nu}_{iK},\hat{\nu}_{iK+1}\}_{i=1}^n$ and $K$}
	\Output{$\hat{Z}_K$ and $\hat{Z}_{K+1}^b$}
divide $\{\hat{\nu}_{iK}\}_{i=1}^n$ into $K$ groups by the k-means
algorithm with $K$ centroids. Denote the membership matrix as $\hat{Z}_K$
with the corresponding communities $\{\widehat{\mathcal{C}}_{k,K}\}_{k=1}^K$\;
	\For{$k \leftarrow 1$ \KwTo $K$}{
	divide $\widehat{\mathcal{C}}_{k,K}$ into two subgroups by applying the
	k-means algorithm on $\{\hat{\nu}_{iK+1}\}_{i \in \widehat{\mathcal{C}}_{k,K}}$.
	Denote the two subgroups as $\widehat{\mathcal{C}}_{k,K}(1)$ and $\widehat{\mathcal{C}}_{k,K}(2)$\;
	compute
	\begin{equation}
	\label{eq:QK}
	\widehat{Q}_K(k) = \frac{\widehat{\Phi}(\widehat{\mathcal{C}}_{k,K}) - \widehat{\Phi}(%
		\widehat{\mathcal{C}}_{k,K}(1)) - \widehat{\Phi}(\widehat{\mathcal{C}}_{k,K}(2))}{\# \widehat{\mathcal{C}}%
		_{k,K}},
	\end{equation}
	where for an arbitrary index set $C$, $\widehat{\Phi}(C) = \sum_{i \in \mathcal{C}} ||%
	\hat{\nu}_{iK+1} - \frac{\sum_{i \in \mathcal{C}}\hat{\nu}_{iK+1}}{\# \mathcal{C}}||^2$\;
	}
choose $\hat{k}=\argmax_{1\leq k\leq K}\hat{Q}_{K}(k)$ and
		denote
		$$\{\widehat{\mathcal{C}}_{k,K+1}^{b}\}_{k=1}^{K+1}=\{\{\widehat{\mathcal{C}}_{k,K}\}_{k<\hat{k%
		}},\widehat{\mathcal{C}}_{\hat{k},K}(1),\newline \{\widehat{\mathcal{C}}_{k,K}\}_{k>\hat{k}},
		\widehat{\mathcal{C}}_{\hat{k},K}(2)\}$$
		as the new groups for $K+1$. The
		corresponding membership matrix is denoted as $\hat{Z}_{K+1}^{b}$.
	\caption{Estimation of the number of communities}
\end{algorithm}\DecMargin{1em}\medskip

%
%
%

Algorithm \ref{algo:2} applies the standard spectral clustering approach to
obtain $\hat{Z}_{K}$ and a binary segmentation method to obtain $\hat{Z}%
_{K+1}^{b}$. This procedure is computationally fast. Moreover, the algorithm
leads to $\widehat{\mathcal{C}}_{k,K+1}^{b}=\widehat{\mathcal{C}}_{k,K}$ for
$k\neq \hat{k}$ and $\widehat{\mathcal{C}}_{\hat{k},K+1}^{b}\cup \widehat{%
\mathcal{C}}_{K+1,K+1}^{b}=\widehat{\mathcal{C}}_{\hat{k},K}$, which ensures
that the parameter estimators $\hat{P}_{ij}(\hat{Z}_{K})$ and $\hat{P}_{ij}(%
\hat{Z}_{K+1}^{b})$ in the DCSBM are consistent when $K=K_{0}$.

\section{Theory}

\label{sec:theory}

\subsection{Identification}

The population counterpart of $L_{\tau }$ is
\begin{equation*}
\mathcal{L}_{\tau }=\mathcal{D}_{\tau }^{-1/2}P\mathcal{D}_{\tau }^{-1/2},
\end{equation*}%
where $\mathcal{D}_{\tau }=\mathcal{D}+\tau I_{n}$ and $\mathcal{D}=\text{
diag}(d_{1},\ldots ,d_{n})$ with $d_{i}=\sum_{j=1}^{n}P_{ij}$. Let $\pi
_{kn}=n_{k,K_0}/n$ and $\Pi_n = \text{diag}(\pi_{1n},\cdots,\pi_{K_0n})$.

\begin{assumption}
	\label{ass:id3} Let $H_{K_0}=\rho _{n}^{-1}B_{K_0} = [H_{kl,K_0}]$
	for some $\rho _{n}>0$, $W_{k}=\sum_{l=1}^{K_{0}}H_{kl,K_0}\pi
	_{ln} $, $\mathcal{D}_{H}=\text{diag}(W_{1},\cdots ,W_{K_{0}})$, and $%
	H_{0,K_0}=\mathcal{D}_{H}^{-1/2}H_{K_0}\mathcal{D}_{H}^{-1/2}$.
	Then,
	\begin{enumerate}
		\item[(1)] $H_{K_0}$ is not varying with $n$,
		\item[(2)] as $n\rightarrow
		\infty $, $H_{0,K_0}\rightarrow H_{0,K_0}^{\ast }$ where $%
		H_{0,K_0}^{\ast }$ has full rank $K_{0}$,
		\item[(3)]all elements of $%
		H_{0,K_0}^{\ast }$ are positive,
		\item[(4)] there exist two constants $%
		\underline{\theta }$ and $\overline{\theta }$ such that $0<\underline{\theta
		}\leq \inf_{i}\theta _{i}\leq \sup_{i}\theta _{i}\leq \overline{\theta }.$
	\end{enumerate}
\end{assumption}

Several remarks are in order. First, Assumption \ref{ass:id3} implies that
the average node degree is of order $n\rho _{n}$. The network can be
semi-dense if $\rho _{n}\rightarrow 0$ but $n\rho _{n}\rightarrow \infty $. Second, Assumption \ref{ass:id3}(1) is just for notational simplicity. All
our results still hold if $H_{K_0}$ depends on $n$ and converges to some
limit. Third, Assumption \ref{ass:id3}(2) ensures that the DCSBM has $K_{0}$
communities. To see this, note that Assumption \ref{ass:id3}(2) implies both
$H_{K_{0}}$ and $B_{K_{0}}$ have full rank. Suppose there exist $\{\tilde{%
\theta}_{i}\}_{i=1}^{n}$, $\tilde{\Theta}=\text{diag}(\tilde{\theta}%
_{1},\cdots ,\tilde{\theta}_{n})$, $\tilde{Z}_{K_{0}^{\prime }}$, and $%
\tilde{B}_{K_{0}^{\prime }}$ such that $\tilde{B}_{K_{0}^{\prime }}$ is a
full rank $K_{0}^{\prime }\times K_{0}^{\prime }$ matrix and
\begin{equation*}
\Theta Z_{K_{0}}B_{K_{0}}Z_{K_{0}}^{T}\Theta ^{T}=P=\tilde{\Theta}\tilde{Z}%
_{K_{0}^{\prime }}\tilde{B}_{K_{0}^{\prime }}\tilde{Z}_{K_{0}^{\prime }}^{T}%
\tilde{\Theta}^{T}.
\end{equation*}%
Further suppose that the membership matrix $\tilde{Z}_{K_{0}^{\prime }}$ is
non-degenerate in the sense that each community identified by $\tilde{Z}%
_{K_{0}^{\prime }}$ is nonempty, which implies that $\tilde{Z}%
_{K_{0}^{\prime }}$ has full column rank. Then, the full rank condition of $%
B_{K_{0}}$ and $\tilde{B}_{K_{0}^{\prime }}$ implies that
\begin{align*}
K_{0}=\text{rank}(B_{K_{0}})=& \text{rank}(\Theta
Z_{K_{0}}B_{K_{0}}Z_{K_{0}}^{T}\Theta ^{T}) \\
=& \text{rank}(P) \\
=& \text{rank}(\tilde{\Theta}\tilde{Z}_{K_{0}^{\prime }}\tilde{B}%
_{K_{0}^{\prime }}\tilde{Z}_{K_{0}^{\prime }}^{T}\tilde{\Theta}^{T})=\text{%
rank}(\tilde{B}_{K_{0}^{\prime }})=K_{0}^{\prime }.
\end{align*}%
That is, the number of communities is identified. Fourth, from the perspective of real data applications, the full-rank
condition on $B_{K_{0}}$ is reasonable. In networks, communities are usually
groups of nodes that have a higher probability of being connected to each
other within the same group than to members of other groups. This directly
implies the full rank condition of $B_{K_{0}}$ if $K_{0}=2$. In general, by
the Gershgorin circle theorem, for each row, if the sum of off-diagonal
elements is strictly less than the diagonal element, i.e., for $k=1,\cdots
,K_{0}$
\begin{equation*}
\sum_{l=1,\cdots ,K_{0},\text{ }l\neq k}B_{kl,K_{0}}<B_{kk,K_{0}},
\end{equation*}%
then $B$ has full rank. Such condition is just a sufficient condition for
our full rank requirement. For estimating the SBMs, the semi-definite
programming method can also be used. It needs the strong assortativity
condition \citep{CL15} given as
\begin{equation*}
\min_{k=1,\cdots ,K_{0}}B_{kk,K_{0}}>\max_{k,l=1,\cdots ,K_{0},\text{ }k\neq
l}B_{kl,K_{0}}.
\end{equation*}%
In general, the strong assortativity and Assumption 1(2) do not nest within
each other. For example, the following matrix has full rank but violates the
strong assortativity:
\begin{equation*}
\begin{pmatrix}
0.8 & 0.4 & 0.1 \\
0.4 & 0.5 & 0.05 \\
0.1 & 0.05 & 0.2%
\end{pmatrix}%
.
\end{equation*}%
Which assumption is more plausible depends on the empirical data at hand. In
the three real data examples considered in Section \ref{sec:real} and Section \ref{sec:add_application} of the Supplemental Materials, the full
rank condition holds for all of them, but the strong assortativity does not
hold for the political books network. Fifth, from the theoretical perspective, the full-rank condition (i.e., the $%
K_0$-th largest absolute eigenvalue of the $\mathcal{L}_\tau$ is bounded
away from zero) is a common assumption in the literature. See, for example,
\cite{RCY11}, \cite{LR15}, and \cite{JY16}. It is fundamental for the
spectral clustering method. If it does not hold, i.e., the $K_0$-th
eigenvalue of the population graph Laplacian is exactly zero, then the
corresponding population eigenvector is not uniquely defined. Sixth, Assumption \ref{ass:id3}(3) is sufficient for $\hat{\nu}_{iK}$ in %
\eqref{eq:nuK} to be well-defined, as shown in Lemma \ref{lem:id3} in the
Supplemental Materials. Last, for simplicity, we restrict $\theta _{i}$ to be bounded between zero
and infinity. This assumption can be relaxed at the cost of more complicated
notations.

Next, let $\Theta _{\tau }=\text{diag}(\theta _{1}^{\tau },\ldots ,\theta
_{n}^{\tau })$, where $\theta _{i}^{\tau }=\theta _{i}d_{i}/(d_{i}+\tau )$
for $i=1,\ldots ,n$, $n_{k,K_0}^{\tau }=\sum_{i\in \mathcal{C}%
_{k,K_0}}\theta _{i}^{\tau }$, and $\Pi_n^\tau = \text{diag}%
(n_{1,K_0}^\tau/n,\cdots,n_{K_0,K_0}^\tau/n).$

\begin{assumption}
	\label{ass:nk2} Suppose
	\begin{enumerate}
		\item[(1)] 	there exist $\{\pi_{k\infty}\}_{k=1}^{K_0}$ and $\{\pi_{k\infty}^{\prime
		}\}_{k=1}^{K_0}$ that are bounded between zero and infinity such that
		\begin{equation*}
		\Pi_n \rightarrow \Pi_\infty = \text{diag}(\pi _{1\infty},\ldots ,\pi
		_{K_0\infty} ) \quad \text{and} \quad \Pi_n^\tau \rightarrow \Pi_\infty^{\prime }= \text{diag%
		}(\pi _{1\infty}^{\prime },\ldots ,\pi^{\prime }_{K_0\infty} ),
		\end{equation*}
	\item[(2)] $(\Pi _{\infty }^{\prime })^{1/2}H_{0,K_0}^{\ast }(\Pi
	_{\infty }^{\prime })^{1/2}$ has $K_0$ distinct eigenvalues.
	\end{enumerate}
\end{assumption}

The second convergence in Assumption \ref{ass:nk2}(1) can be easily
satisfied by choosing $\tau $ to be the average degree $(\bar{d})$ in the
network. Let $|\lambda _{1}|\geq \cdots \geq |\lambda _{K_{0}}|$ be the
eigenvalues of $(\Pi _{\infty }^{\prime })^{1/2}H_{0,K_{0}}^{\ast }(\Pi
_{\infty }^{\prime })^{1/2}$ and
\begin{equation*}
\text{eigsp}((\Pi _{\infty }^{\prime })^{1/2}H_{0,K_{0}}^{\ast }(\Pi
_{\infty }^{\prime })^{1/2})=\min_{k=1,\cdots ,K_{0}-1}|\lambda
_{k+1}-\lambda _{k}|
\end{equation*}%
be the gap between adjacent eigenvalues of $(\Pi _{\infty }^{\prime
})^{1/2}H_{0,K_{0}}^{\ast }(\Pi _{\infty }^{\prime })^{1/2}$, as defined in
\cite{J15}. Then, Assumption \ref{ass:nk2}(2) requires that
\begin{equation*}
\text{eigsp}((\Pi _{\infty }^{\prime })^{1/2}H_{0,K_{0}}^{\ast }(\Pi
_{\infty }^{\prime })^{1/2})\geq C>0
\end{equation*}%
for some constant $C$. The same condition is assumed in \cite{J15}.\footnote{%
See \citet[Lemma 2.3]{J15}.} Assumption \ref{ass:nk2}(2) is mild from a
practical point of view. If we denote $H_{0,K_{0}}^{\ast }$ as $%
vec(H_{0,K_{0}}^{\ast })\in \Re ^{K_{0}^{2}}$ such that $H_{0,K_{0}}^{\ast }$
is symmetric and full rank, then Assumption \ref{ass:nk2}(2) is only
violated for a set in $\Re ^{K_{0}^{2}}$ with zero Lebesgue measure.
Theoretically, as $K_{0}$ is not known a priori, we need to apply spectral
clustering to the first $K$ eigenvectors of the graph Laplacian for $%
K=1,\cdots ,K_{0}$. Therefore, at the population level, we require that the
eigenspace generated by the first $K$ eigenvectors is identified for all $%
K=1,\cdots ,K_{0}$, which is equivalent to Assumption \ref{ass:nk2}(2).

\bigskip

Consider the spectral decomposition of $\mathcal{L}_{\tau }$,
\begin{equation*}
\mathcal{L}_{\tau }=U_{1n}\Sigma _{1n}U_{1n}^{T},
\end{equation*}%
where $\Sigma _{1n}=\text{diag}(\sigma _{1n},\ldots ,\sigma _{K_{0}n})$ is a
$K_{0}\times K_{0}$ matrix that contains the eigenvalues of $\mathcal{L}%
_{\tau }$ such that $|\sigma _{1n}|\geq |\sigma _{2n}|\geq \cdots \geq
|\sigma _{K_{0}n}|>0$ and $U_{1n}^{T}U_{1n}=I_{K_{0}}$.

\begin{theorem}
	\label{thm:id3} Suppose Assumptions \ref{ass:id3} and \ref{ass:nk2} hold.
	Let $u_{i}^{T}$ and $u_{i}(K)$ be the $i$-th row of $U_{1n}$ and the top $K$ elements of $u_i$, respectively.
	\begin{enumerate}
	\item[(1)] If $[Z_{K_{0}}]_{i }=[Z_{K_{0}}]_{j }$, then $\Vert \frac{%
		u_{i}}{\Vert u_{i}\Vert }-\frac{u_{j}}{\Vert u_{j}\Vert }\Vert =0;$ if $%
	[Z_{K_{0}}]_{i }\neq \lbrack Z_{K_{0}}]_{j},$ then $\Vert \frac{%
		u_{i}}{\Vert u_{i}\Vert }-\frac{u_{j}}{\Vert u_{j}\Vert }%
	\Vert =\sqrt{2}$.
	\item[(2)] There exist $L_K$ distinct $K \times 1$ vectors, denoted as $(\bar{\nu}_{1,K},\cdots,\bar{\nu}_{L_K,K})$, such that the nodes can be divided into $L_{K}$
	groups, denoted by $\{G_{l,K}\}_{l=1}^{L_{K}}$, $K \leq L_K \leq K_0$, for any $l=1,\cdots
	,L_{K}$,
	\begin{equation*}
	\limsup_{n}\sup_{i,j\in G_{l,K}}\left\Vert\frac{u_{i}(K)}{||u_{i}(K)||}-%
	\bar{\nu}_{l,K}\right\Vert=0,
	\end{equation*}%
	and for any $l\neq l^{\prime }$ and some constant $c>0$ independent of $n$%
	,
	\begin{equation*}
	\liminf_{n}\inf_{i\in G_{l,K},j\in G_{l^{\prime },K}}\left\Vert \frac{u_{i}(K)}{%
		||u_{i}(K)||}-\bar{\nu}_{l,K}\right\Vert \geq c.
	\end{equation*}%
	\end{enumerate}
\end{theorem}


Several remarks are in order. First, Theorem \ref{thm:id3}(1) has already
been established in the literature. See \cite{QR13} and \cite{SWZ17}. It
implies that the eigenvectors of the graph Laplacian contain information
about the group structure. Second, Theorem \ref{thm:id3}(2) implies that the
first $K$ columns of eigenvectors after row normalization still contain
information for at least $K$ communities, when $K \leq K_0$. In particular,
when $K = K_0$, $L_{K_0} = K_0$ and Theorem \ref{thm:id3}(1) implies that
Theorem \ref{thm:id3}(2) holds with the true communities, i.e., $%
\{G_{l,L_{K_0}}\}_{l=1}^{L_{K_0}} = \{\mathcal{C}_{k,K_0}\}_{k=1}^{K_0}$.
Therefore, $\{G_{l,K}\}_{l=1}^{L_K}$ can be viewed as the true communities
identified by the first $K$ columns of eigenvectors. Third, Lemma \ref%
{lem:id3} in the Supplemental Materials implies that $||u_{i}(K)||$ is
bounded away from zero for $K = 1,\cdots,K_0$, which guarantees that $\frac{%
u_{i}(K)}{||u_{i}(K)||}$ is well defined. This result is similar to %
\citet[Lemma 2.5]{J15}.

\subsection{Properties of the estimated memberships}

\label{sec:est} In the following, we aim to show that, under certain
conditions, if $K\leq K_{0}$, then $\hat{Z}_{K}=Z_{K}$ and $\hat{Z}%
_{K}^{b}=Z_{K}^{b}$ almost surely (a.s.) for some deterministic membership
matrices $Z_{K}$ and $Z_{K}^{b}$. We denote the communities identified by $%
Z_{K}$ and $Z_{K}^{b}$ as $\{\mathcal{C}_{k,K}\}_{k=1}^{K}$ and $\{\mathcal{C%
}_{k,K}^{b}\}_{k=1}^{K}$, respectively. Note that $L_{K}$ is not necessarily
equal to $K$. This implies that neither $\{\mathcal{C}_{k,K}\}_{k=1}^{K}$
nor $\{\mathcal{C}_{k,K}^{b}\}_{k=1}^{K}$ is necessarily equal to the true
communities $\{G_{l,K}\}_{l=1}^{L_{K}}$. We can view $Z_{K}$ and $Z_{K+1}^{b}
$ as the pseudo true values of our estimation procedure described in Section %
\ref{sec:algo}. We slightly abuse the notation by calling $Z_{K}$ evaluated
at $K=K_{0}$ as the pseudo true membership matrix when $K=K_{0}$ while $%
Z_{K_{0}} $ as the true membership matrix. Theorem \ref{thm:oracleDC} below
shows that when $K=K_{0}$, the pseudo true values $Z_{K}$ and $Z_{K}^{b}$
are equal to the true membership matrix $Z_{K_{0}}$. Therefore, the notation
is still consistent and we can just write $Z_{K_{0}}$ as the (pseudo) true
membership matrix for $K=K_{0}$.

\begin{definition}
For $i\in G_{l,K}$ and $l=1,...,L_K$, $K = 2,\cdots,K_0 $, let
\begin{equation*}
\nu _{iK}=\bar{\nu}_{l,K}.
\end{equation*}%
Then, $(Z_{K},Z_{K+1}^b)$ is defined by applying Algorithm \ref{algo:2} to $%
\{\nu _{iK}\}_{i=1}^{n}$, $K =1,\cdots,K_0-1$. When $K=1$, we can trivially define $Z_{1}=Z_{1}^{b}=[n]=\left\{
1,2,...,n\right\}$.
\end{definition}

\begin{assumption}
	\label{ass:uniqueDC} Suppose that
\begin{enumerate}
	\item[(1)] 	the above definitions of $Z_K$ and $Z^b_K$ are unique for $%
	K=1,\cdots,K_0$;
	\item[(2)] there exist a positive constant $c$ independent of $n$ and $k^* = 1,\cdots,K$
	such that $Q_K(k^*) - \max_{k \neq k^*}Q_K(k)\geq c$ for $K = 2,\cdots,K_0-1$, where $Q_K(\cdot)$ is similarly defined as $\hat{Q}_K(\cdot)$ in \eqref{eq:QK} with $\hat{\nu}_{iK+1}$ and $\{\widehat{\mathcal{C}}_{k,K}\}$ replaced by $\nu_{iK+1}$ and $\{\mathcal{C}_{k,K}\}$, respectively.
\end{enumerate}	
\end{assumption}

Several remarks are in order. First, the communities identified by $%
Z_{K+1}^{b}$ can be written as
\begin{equation*}
\{\mathcal{C}_{k,K+1}^{b}\}_{k=1}^{K+1}=\{\mathcal{C}_{1,K},\cdots ,\mathcal{%
C}_{k^{\ast }-1,K},\mathcal{C}_{k^{\ast },K}(1),\mathcal{C}_{k^{\ast },K}(2),%
\mathcal{C}_{k^{\ast }+1,K},\cdots ,\mathcal{C}_{K,K}\}.
\end{equation*}%
Second, we provide more details on $Z_{K}$, $Z_{K}^{b}$, and $Q_{K}(\cdot )$
in Section \ref{sec:appago} in the Supplemental Materials. Third, the
uniqueness requirement is mild. If $L_{K}=K$, then obviously $\{\mathcal{C}%
_{k,K}\}_{k=1}^{K}=\{G_{l,K}\}_{l=1}^{L_{K}}$, which implies $Z_{K}$ is
uniquely defined. Fourth, we have $L_{K_{0}}=K_{0}$. Therefore, by
definition, $\{\mathcal{C}_{k,K_{0}}\}_{k=1}^{K_{0}}$ defined by $Z_{K_{0}}$
equal $\{G_{l,K_{0}}\}_{l=1}^{K_{0}}$, which are the true communities.
Fifth, when $L_{K}=K$ and $L_{K+1}=K+1$ for $K\leq K_{0}-1$, by the
pigeonhole principle, there only exists one $k\in \{1,\cdots ,K\}$, denoted
as $k^{\dagger }$ such that $\mathcal{C}_{k^{\dagger },K}=G_{k^{\dagger },K}$
contains two of $\{G_{l,K+1}\}_{l=1}^{K+1}$. Then by Theorem \ref{thm:id3}%
(2), there exists some constant $c>0$ such that $Q_{K}(k^{\dagger })\geq c$
and $Q_{K}(k)\rightarrow 0$ for $k\neq k^{\dagger }$. In this case, $k^{\ast
}=k^{\dagger }$ and Assumption \ref{ass:uniqueDC}(2) holds. Sixth,
Assumption \ref{ass:uniqueDC} is similar to \citet[Assumption 2.1]{WB16}. It
is used as a matter of notational convenience but not of necessity. Under
Assumption \ref{ass:uniqueDC}, we will show that the pseudo-LR after
re-centering is asymptotically normal. If Assumption \ref{ass:uniqueDC}
fails and $(Z_{K},Z_{K}^{b})$ are not unique, it can be anticipated that the
pseudo-LR after re-centering will be asymptotically mixture normal with
weights depending on the probability of choosing one classification among
all possibilities. Last, although Assumption \ref{ass:uniqueDC} is used to
characterize the limiting distribution of the re-centered pseudo-LR, it does
not affect the rate of bias term in the under-fitting case. Because the bias
term will dominate the centered term, we actually only need the rate of bias
to show the validity of our selection procedure. Therefore, even if
Assumption \ref{ass:uniqueDC} fails, it is reasonable to expect that our
procedure can still consistently select the true number of communities as
established in Section \ref{sec:over}.

\begin{assumption}
	\label{ass:rate} Assume $\rho _{n}n/\log (n)\geq C_1 $ for some constant $C_1>0$ sufficiently large and $\tau
	=O(n\rho _{n})$.
\end{assumption}

Recall that the degree of the network is of order $n\rho _{n}$. Assumption %
\ref{ass:rate} requires the degree to diverge at a rate no slower than $\log
(n)$, which is the most relaxed degree growth rate for exact community
recovery when $K$ is known. See \cite{A18} for an excellent survey on the
recent development of estimation of SBMs and DCSBMs.\footnote{%
We thank a referee for this reference.} For determining the number of
communities, \cite{CL17} require the same condition on the degree for SBMs,
but they do not provide any theory for DCSBMs. \cite{WB16} establish the
theories for DCSBMs but require that $n^{1/2}\rho _{n}/\log (n)\rightarrow
\infty $, or equivalently, the degree diverges to infinity at a rate faster
than $n^{1/2}\log (n)$. We require a weaker condition compared to \cite{WB16}%
, mainly due to the fact that we use a pseudo instead of the true likelihood
ratio. In DCSBMs, the rate of convergence for the estimator $\hat{\theta}_{i}
$ of $\theta _{i}$ is much slower than that for the estimator of the block
probability matrix. By using the ratio $\frac{\hat{P}_{ij}(\hat{Z}_{K+1}^{b})%
}{\hat{P}_{ij}(\hat{Z}_{K})}$ in the definition of pseudo-LR, the components
of $\hat{\theta}_{i}$'s that cause the slower convergence rate in both the
numerator and the denominator cancel each other out, so that the convergence
rate of $\frac{\hat{P}_{ij}(\hat{Z}_{K+1}^{b})}{\hat{P}_{ij}(\hat{Z}_{K})}$
is affected. We recommend using regularization to improve the finite sample
performance of spectral clustering. By Assumption \ref{ass:id3}, setting $%
\tau $ as the average degree $\bar{d}$ satisfies Assumption \ref{ass:rate}.
In practice, $\bar{d}$ is unobserved and we replace it by the sample
version, following the lead of \cite{QR13}. In the proof of Theorem \ref%
{thm:underDC} in the Supplemental Materials, we show that the sample average
degree is of the same order of magnitude as its population counterpart
almost surely because
\begin{equation*}
\sup_{i}\left\vert \frac{\hat{d}_{i}}{d_{i}}-1\right\vert \leq C\sqrt{\frac{%
\log (n)}{n\rho _{n}}}
\end{equation*}%
for some fixed constant $C>0$. One can also use the data-driven method
proposed by \cite{JY16} to select the regularizer. Based on the simulation
study in \cite{SWZ17}, the performances of spectral clustering using sample
average degree and data-driven regularizer are similar.

\begin{definition}
	Suppose there are two membership matrices $Z_1$ and $Z_2$ with corresponding
	communities $\{\mathcal{C}_{k}^j\}_{k=1}^{K_j}$, $j=1,2$, respectively. Then we say $%
	Z_1$ is finer than $Z_2$ if for any $k_1 = 1,\cdots,K_1$, there exists $k_2
	= 1,\cdots,K_2$ such that
	\begin{equation*}
	\mathcal{C}_{k_1}^1 \subset \mathcal{C}_{k_2}^2.
	\end{equation*}
	In this case, we write $Z_1 \succeq Z_2.$
\end{definition}

\begin{theorem}
	\label{thm:oracleDC} If Assumptions \ref{ass:id3}--\ref{ass:rate} hold, then
	\begin{enumerate}
		\item[(1)] for $K=1,\cdots ,K_{0}$,
		\begin{equation*}
		\hat{Z}_{K}=Z_{K}\quad a.s.\quad \text{and} \quad Z_{K_{0}}\succeq Z_{K},
		\end{equation*}%
		\item[(2)] for $K=1,\cdots ,K_{0}-1$,
		\begin{equation*}
		\hat{Z}_{K+1}^{b}=Z_{K+1}^{b}\quad a.s.\quad \text{and} \quad Z_{K_{0}}\succeq
		Z_{K+1}^{b},
		\end{equation*}%
		\item[(3)] after relabeling, we have $\widehat{\mathcal{C}}_{k,K+1}^{b}=\mathcal{C}_{k,K}$ for $%
		k=1,\cdots ,K-1$ and $\mathcal{C}_{K,K}=\widehat{\mathcal{C}}_{K,K+1}^{b}\cup \widehat{\mathcal{C}}_{K+1,K+1}^{b}$%
		, for $K=1,\cdots ,K_{0}$, $a.s.$
	\end{enumerate}
\end{theorem}

Theorem \ref{thm:oracleDC}(1) and (2) show that $\hat{Z}_{K}$ and $\hat{Z}%
_{K}^{b}$ equal their pseudo true counterparts almost surely. This is the
oracle property of estimating the community membership when we either under-
or just-fit the model, i.e., $K\leq K_{0}$. On the other hand, it is very
difficult, if not completely impossible, to show the similar oracle property
for the over-fitting case, i.e., $K>K_{0}$. In particular, we are unable to
uniquely define $Z_{K_{0}+1}^{b}$ and show that $\hat{Z}%
_{K_{0}+1}^{b}=Z_{K_{0}+1}^{b}$ \textit{a.s}. As pointed out by \cite{WB16},
even in the population level (i.e., the probability matrix is observed),
\textquotedblleft embedding a $K$-block model in a larger model can be
achieved by appropriately splitting the labels $Z$ and there are an
exponential number of possible splits.\textquotedblright\ However, Theorem %
\ref{thm:oracleDC}(3) with $K=K_{0}$ shows that, for any $k=1,\cdots ,K_{0}+1
$, there exists some $k^{\prime }$ such that $\widehat{\mathcal{C}}%
_{k,K_{0}+1}^{b}\subset \widehat{\mathcal{C}}_{k^{\prime },K_{0}}$, which
should be one of the true communities based on the oracle property. We can
use this feature to handle the over-fitting case.

\subsection{Properties of the {pseudo-LR} and the estimated number of
communities}

\label{sec:over} Without loss of generality, we assume that $\hat{Z}_{K}^{b}$
is obtained by splitting the last group in $\hat{Z}_{K-1}$ into the $(K-1)$%
-th and $K$-th groups in $\hat{Z}_{K}^{b}$. Further denote, for $%
k,l=1,\cdots ,K$ and $k\leq l$,
\begin{equation*}
\Gamma _{kl,K}^{0b}=\sum_{s\in I(\mathcal{C}_{k,K}^{b}),\text{ }t\in I(%
\mathcal{C}_{l,K}^{b})}H_{st,K_{0}}\pi _{s\infty }\pi _{t\infty }\quad \text{%
and}\quad \Gamma _{K}^{0b}=[\Gamma _{kl,K}^{0b}],
\end{equation*}%
%
where $I(\mathcal{C}_{k,K}^{b})$ denotes a subset of $[K_{0}]$ such that if $%
m\in I(\mathcal{C}_{k,K}^{b})$, then $\mathcal{C}_{m,K_{0}}\subset \mathcal{C%
}_{k,K}^{b}$.

\begin{assumption}
	\label{ass:BKDC} For $K =2,\cdots,K_0$,  $\Gamma _{K}^{0b} \notin \mathbb{W}_K$, where $\mathbb{W}_K$ is a class of symmetric $K \times K$ matrices which is specified in the Supplemental Materials.

\end{assumption}

Several remarks are in order. First, the expression of $\mathbb{W}_K$ is
complicated and can be found in the proof of Theorem \ref{thm:underDC} in
the Supplemental Materials. Second, when $K=2$,
\begin{equation*}
\mathbb{W}_2 = \{W \in \Re^{2 \times 2}: W = W^T,~ W_{12}^2 = W_{11}W_{22}\}.
\end{equation*}
In general, we can view $\mathbb{W}_K$ as a set of $K(K+1)/2 \times 1$
vectors. Then, the Lebesgue measure of $\mathbb{W}_K$ is zero, which means
Assumption \ref{ass:BKDC} is mild. Third, if the last two columns of $%
\Gamma_{K}^{0b}$ are exactly the same, then $\Gamma_{K}^{0b} \in \mathbb{W}_K
$. Assumption \ref{ass:BKDC} rules out this case when $K \leq K_0 $.

\begin{theorem}
	\label{thm:underDC} If Assumptions \ref{ass:id3}--\ref{ass:rate} hold, then,
	for $2\leq K\leq K_{0}$, there exists $\tilde{\mathcal{B}}%
	_{K,n} $ such that
	\begin{equation*}
	\tilde{\varpi}_{K,n}^{-1}\left\{n^{-1}\rho _{n}^{1/2}[L_{n}(\hat{Z}_{K},\hat{Z}%
_{K-1})-\tilde{\mathcal{B}}_{K,n}]\right\}\rightsquigarrow N(0,1)
	\end{equation*}%
	where the
	asymptotic bias $\tilde{\mathcal{B}}_{K,n}$ and variance $\tilde{\varpi}%
	_{K,n}^{2}$ are defined in (\ref{eq:Bias}) and (\ref{eq:Var}), respectively,
	in the Supplemental Materials. If, in addition, Assumption \ref{ass:BKDC} holds, then
	there exist two positive constants $(c_{K1},c_{K2})$ potentially dependent
	on $K$ such that
	\begin{equation*}
	c_{K2}n^{2}\geq \tilde{\mathcal{B}}_{K,n}\geq c_{K1}n^{2}.
	\end{equation*}
\end{theorem}

Theorem \ref{thm:underDC} shows that in the under-fitting case, the
asymptotic bias term that is of order $n^{2}$ will dominate the centered
pseudo-LR that is of order $n\rho _{n}^{-1/2}$. However, when we over-fit
the model, i.e., $K>K_{0}$, the asymptotic bias term will be zero. The
sudden change in the orders of magnitude of the {pseudo-LR} $L_{n}(\hat{Z}%
_{K}^{b},\hat{Z}_{K-1})$ provides useful information on the true number of
communities.

Next, we consider the over-fitting case. Let $z_{K_0+1}$ be a generic $%
n\times (K_0+1)$ membership matrix,
\begin{align}  \label{eq:nk'}
n_{kl}(z_{K_0+1})= & \sum_{i=1}^{n}\sum_{j\neq
i}1\{[z_{K_0+1}]_{ik}=1,[z_{K_0+1}]_{jl}=1\}  \notag \\
=&
\begin{cases}
n_{k}(z_{K_0+1})n_{l}(z_{K_0+1}) & \quad \text{if}\quad k\neq l \\
n_{k}(z_{K_0+1})(n_{k}(z_{K_0+1})-1) & \quad \text{if}\quad k=l,%
\end{cases}%
\end{align}%
and $n_{k}(z_{K_0+1}) = \sum_{l=1}^{K_0+1}n_{kl}(z_{K_0+1})$. We emphasize
the dependence of $n_{kl}$ and $n_k$ on the membership matrix $z_{K_0+1}$
because when $K>K_0$, neither $Z_K$ nor $Z^b_K$ is uniquely defined. The
following assumption restricts the possible realizations $\hat{Z}_{K_0+1}^b$
can take.
\begin{assumption}
	\label{ass:neps} There exists some sufficiently small constant $\varepsilon$
	such that
	\begin{equation*}
	\inf_{1\leq k \leq K_0+1}n_k(\hat{Z}^b_{K_0+1})/n \geq \varepsilon.
	\end{equation*}
\end{assumption}

Assumption \ref{ass:neps} always holds in our simulation. By Theorem \ref%
{thm:oracleDC}, $\hat{Z}_{K_{0}}=Z_{K_{0}}$ $a.s.$ Suppose we obtain $\hat{Z}%
_{K_{0}+1}^{b}$ by splitting the last community (i.e., the $\mathcal{C}%
_{K_{0},K_{0}}$) into two groups by binary segmentation. In simulation, we
observe that the two new groups $\widehat{\mathcal{C}}_{K_{0},K_{0}+1}^{b}$
and $\widehat{\mathcal{C}}_{K_{0}+1,K_{0}+1}^{b}$ have close to even sizes.
In addition, we can modify the binary segmentation procedure to ensure that
Assumption \ref{ass:neps} holds automatically. In particular, suppose $%
n_{K_{0}}(\hat{Z}_{K_{0}+1}^{b})\leq n\varepsilon $, then let
\begin{equation*}
\widehat{\mathcal{C}}_{K_{0},K_{0}+1}^{b,new}=\widehat{\mathcal{C}}%
_{K_{0},K_{0}+1}^{b}\cup \breve{\mathcal{C}}_{K_{0}+1,K_{0}+1}^{b}\quad
\text{and}\quad \widehat{\mathcal{C}}_{K_{0}+1,K_{0}+1}^{b,new}=\widehat{%
\mathcal{C}}_{K_{0},K_{0}}\backslash \widehat{\mathcal{C}}%
_{K_{0},K_{0}+1}^{b,new},
\end{equation*}%
where $\breve{\mathcal{C}}_{K_{0}+1,K_{0}+1}^{b}$ is half of $\widehat{%
\mathcal{C}}_{K_{0}+1,K_{0}+1}^{b}$ by random splitting. Then $\widehat{%
\mathcal{C}}_{K_{0},K_{0}+1}^{b,new}$ and $\widehat{\mathcal{C}}%
_{K_{0}+1,K_{0}+1}^{b,new}$ satisfy Assumption \ref{ass:neps}. Although we
do not know $K_{0}$ a priori, we can apply this modification for any $%
K=1,\cdots ,K_{\max }$. When $K<K_{0}$, Theorem \ref{thm:oracleDC}(2) shows
that, for some sufficiently small $\varepsilon $,
\begin{equation*}
n_{k}(\hat{Z}_{K+1}^{b})=n_{k}(Z_{K+1}^{b})\geq \inf_{k}n_{k,K_{0}}\geq
n\varepsilon \quad a.s.
\end{equation*}%
Therefore, the modification will never take action when $K<K_{0}$, which
implies that all our results still hold under this modification.

\begin{theorem}
	\label{thm:overDC} Suppose that Assumptions \ref{ass:id3}--\ref{ass:neps}
	hold. Then
	\begin{equation*}
	0\leq L_{n}(\hat{Z}_{K_{0}+1}^{b},\hat{Z}_{K_{0}})\leq O_{p}(n\rho
	_{n}^{-1}).
	\end{equation*}%
	In addition, if $h_{n}\rightarrow 0$ and $n\rho _{n}h_{n}\rightarrow \infty $%
	, then
	\begin{equation*}
	P(\hat{K}_{1}\geq K_{0})\rightarrow 1 \quad \text{and}  \quad	P(\hat{K}_{2}=K_{0})\rightarrow 1.
	\end{equation*}%
\end{theorem}

Several remarks are in order. First, Theorem \ref{thm:overDC} establishes
the upper bound for the pseudo-LR in the over-fitting case. Like \cite{WB16}%
, we are unable to obtain its exact limiting distribution because we do not
have the oracle property for $\hat{Z}_{K_{0}+1}^{b}$. The more profound
reason for the lack of oracle property is that we have limited knowledge on
the asymptotic behavior of the $(K_{0}+1)$-th column of the eigenvector
matrix $\widehat{U}_{n}$. Fortunately, the upper bound is sufficient for the
consistent estimation of $K_{0}$ with the help of the tuning parameter $h_{n}
$. Second, we show that $\hat{K}%
_{1}$ cannot under-estimate the number of communities in large samples. This
result is similar to that in \cite{CL17} who showed that NCV does not
under-estimate the number of communities in large samples. Third, to obtain
a consistent estimate of $K_{0},$ we can employ the estimator $\hat{K}_{2}$
which requires to specify the tuning parameter $h_{n}$. This parameter plays
the same role as the penalty term in \cite{WB16}'s BIC-type information
criterion. As the average degree $\bar{d}$ is of order $n\rho
_{n}\rightarrow \infty $, $h_{n}=c_{h}\bar{d}^{-1/2}$ satisfies $%
h_{n}\rightarrow 0$ and $n\rho _{n}h_{n}=c_{h}(n\rho _{n})^{1/2}\rightarrow
\infty $. Similarly, the average degree is not feasible and is replaced by
its sample counterpart in practice. This replacement has theoretical
guarantee as discussed after Assumption \ref{ass:rate}. In Section \ref%
{sec:sim}, we investigate the sensitivity of the performance of $\hat{K}_{2}$
with respect to the constant $c_{h}$. Last, as mentioned in the introduction, our pseudo-LR method has computational
advantages over the existing methods. In particular, it is well known that
the likelihood-based method of \cite{WB16} is computationally expensive even
when one uses a variational EM algorithm to approximate the true likelihood.
The NCV method of \cite{CL17} and the ECV method of \cite{LLZ16} can also be
computationally intensive when the number of folds is large.

\section{Numerical Examples on Simulated Networks}

\label{sec:sim}

\subsection{Background and methods\label{sec:background}}

In this section, we conduct simulations to evaluate the performance of our
proposed method. We call our pseudo-LR estimators $\widehat{K}_{1}$ and $%
\widehat{K}_{2}$ as PLR1 and PLR2, respectively. Moreover, we compare our
proposed method with four other approaches, including LRBIC \citep{WB16},
NCV \citep{CL17}, ECV \citep{LLZ16} and BHMC \citep{LL15}. LRBIC considers a
likelihood-based approach for estimating the latent node labels and
selecting models. LRBIC is only designed for the standard SBMs. It requires
one to set the maximum number of communities ($K_{\max }$) and to choose a
tuning parameter to control the order of the BIC-type penalty. NCV applies
cross-validation (CV)\ from spectral clustering, while ECV uses CV with edge
sampling for choosing between SBM and DCSBM and selecting the number of
communities simultaneously. NCV requires one to set $K_{\max }$ and to
choose two tuning parameters, viz, the number of folds for the CV and the
number of repetitions to reduce the randomness of the estimator due to
random sample splitting. ECV requires one to set $K_{\max }$ and to choose
two tuning parameters, viz, the probability for an edge to be drawn and the
number of replications. BHMC is developed by using the network Bethe-Hessian
matrix with moment correction. It requires the selection of a scalar
parameter to define the Bethe Hessian matrix and another one for
fine-tuning. Like our method, BHMC can be generally applied to both SBM and
DCSBM. We use the R package \textquotedblleft randnet\textquotedblright\ to
implement these four methods, and set $K_{\max }=10$ for all methods that
require a maximal value when searching over $K$'s.

\subsection{Data generation mechanisms and settings \label{sec:DGP}}

We consider the following mechanisms to generate the connectivity matrix $%
\boldsymbol{B}=\{B_{k\ell }\}_{1\leq k,\ell \leq K_{0}}$.

Setting 1 (S1). Let $B_{k\ell }=0.5\rho n^{-1/2}\{1+I(k=\ell )\}$ for $1\leq
k,\ell \leq K_{0}$, and for some $\rho >0$.

Setting 2 (S2). We first simulate $\boldsymbol{W}=(W_{1},\ldots
,W_{M_{0}})^{\top }$ from Unif$(0,0.3)^{M_{0}}$, where Unif$(a,b)^{M_{0}}$
denotes an $M_{0}$-dimensional uniform distribution on $[a,b]$ and $%
M_{0}=(K_{0}+1)K_{0}/2$. Let the main diagonal of $\boldsymbol{B}$ be the $%
K_{0}$ largest elements in $\boldsymbol{W}$ and the upper triangular part of
$\boldsymbol{B}$ contain the rest elements in $\boldsymbol{W}$. Let $%
B_{k\ell }=B_{\ell k}$ for all $1\leq k,\ell \leq K_{0}$. We use the
generated $\boldsymbol{B}$ with the smallest singular value no smaller than $%
0.1$.

All simulation results are based on 200 realizations. S1 considers different
sparsity levels for different values of $\rho $, and S2 allows all entries
in $\boldsymbol{B}$ to be different. The membership vector is generated by
sampling each entry independently from $\{1,\ldots ,K_{0}\}$ with
probabilities $\{0.4,0.6\}$, $\{0.3,0.3,0.4\}$ and $\{0.25,0.25,0.25,0.25\}$
for $K_{0}=2,3$ and $4$, respectively. We consider both SBMs and DCSBMs. For
the DCSBMs, we generate the degree parameters $\theta _{i}$ from Unif$%
(0.2,1) $ and further normalize them to satisfy the condition (\ref%
{eq:thetanormalization}).

\subsection{Results \label{sec:results}}

For our method, we let $\tau =\bar{d}$ and $c_{\eta }=0.05$. Note that for
computing the PLR2 estimator $\widehat{K}_{2}$, we need a tuning parameter $%
h_{n}.$ We set $h_{n}=c_{h}\bar{d}^{-1/2}$. We first would like to examine
the performance of the PLR2 estimator when $c_{h}$ takes different values.
Consider $c_{h}=0.5,1.0,1.5,2.0$. Let $\rho =3,4,5$ for design S1. Tables \ref{TAB:Khat2DCSBM} and \ref{TAB:Khat1DCSBM} report the mean of $\widehat{K}_{2}$ and $\widehat{K}_{1}$ by the PLR2 and PLR1 methods, respectively, and the proportion (prop) of correctly estimating $%
K_{0}$ among $200$ simulated datasets when data are generated from the
DCSBMs, for $n=500,1000$ and $K_{0}=1,2,3,4$. For saving space, Tables \ref{TAB:Khat2SBM} and \ref{TAB:Khat1SBM} given in the Supplemental Materials report those statistics when data are generated from the
SBMs.  It is worth noting that when $%
c_{h}=0$,$\ $the two estimates $\widehat{K}_{1}$ and $\widehat{K}_{2}$ are
exactly the same. Comparing Tables \ref{TAB:Khat2SBM} and \ref%
{TAB:Khat2DCSBM} to Tables \ref{TAB:Khat1SBM} and \ref{TAB:Khat1DCSBM}, we
see that for smaller values of $c_{h}$, the behavior of $\widehat{K}_{2}$ is
more similar to that of $\widehat{K}_{1}$. Moreover, Tables \ref%
{TAB:Khat2SBM} and \ref{TAB:Khat2DCSBM} show that the PLR2 estimator has
similar performance at $c_{h}=0.5,1.0,1.5,2.0$ for design S1, and its
performance improves when the value of $\rho $ or the sample size $n$
increases. However, for design S2, PLR2 behaves better at $c_{h}=0.5,1.0$.
Overall, both PLR1 and PLR2 at $c_{h}=0.5,1.0$ have good performance, and
PLR2 with $c_{h}=1.0$ slightly outperforms PLR1 and PLR2 with $c_{h}=0.5$.

\begin{table}[tbph]
\caption{The mean of $\protect\widehat{K}_{2}$ and the proportion (prop) of
correctly estimating $K$ among $200$ simulated datasets when data are
generated from DCSBMs$.$}
\label{TAB:Khat2DCSBM}
\par
\begin{center}
\begin{adjustbox}{width=1.1\textwidth,height = 3.8cm}
			\hskip -1.5cm	
			\begin{tabular}{|l|cc|cccc|cccc|cccc|cccc|}
				\cline{1-19}
				&  &  & \multicolumn{4}{c|}{$K_{0}=1$} & \multicolumn{4}{c|}{$K_{0}=2$} &
				\multicolumn{4}{c|}{$K_{0}=3$} & \multicolumn{4}{c|}{$K_{4}=4$} \\
				\cline{1-19}
				& $\rho $ & $c_{h}$ & \multicolumn{1}{c}{$0.5$} & $1.0$ & $1.5$ & $2.0$ &
				\multicolumn{1}{c}{$0.5$} & $1.0$ & $1.5$ & $2.0$ & \multicolumn{1}{c}{$0.5$}
				& $1.0$ & $1.5$ & $2.0$ & $0.5$ & $1.0$ & $1.5$ & $2.0$ \\ \cline{1-19}
				&  &  & \multicolumn{16}{c|}{$n=500$} \\ \cline{1-19}
				S1 & $3$ & mean & $1.000$ & $1.000$ & $1.000$ & $1.000$ & $2.095$ & $2.000$
				& $2.000$ & $2.000$ & $3.070$ & $3.070$ & $3.000$ & $3.000$ & $3.675$ & $%
				3.675$ & $3.615$ & $3.380$ \\
				&  & prop & $1.000$ & $1.000$ & $1.000$ & $1.000$ & $0.980$ & $1.000$ & $%
				1.000 $ & $1.000$ & $0.980$ & $0.980$ & $1.000$ & $1.000$ & $0.380$ & $0.380$
				& $0.390$ & $0.370$ \\
				& $4$ & mean & $1.000$ & $1.000$ & $1.000$ & $1.000$ & $2.035$ & $2.000$ & $%
				2.000$ & $2.000$ & $3.025$ & $3.000$ & $3.000$ & $3.000$ & $4.175$ & $4.150$
				& $4.100$ & $4.050$ \\
				&  & prop & $1.000$ & $1.000$ & $1.000$ & $1.000$ & $0.990$ & $1.000$ & $%
				1.000 $ & $1.000$ & $0.990$ & $1.000$ & $1.000$ & $1.000$ & $0.915$ & $0.920$
				& $0.935$ & $0.940$ \\
				& $5$ & mean & $1.000$ & $1.000$ & $1.000$ & $1.000$ & $2.000$ & $2.000$ & $%
				2.000$ & $2.000$ & $3.020$ & $3.000$ & $3.000$ & $3.000$ & $4.045$ & $4.015$
				& $4.000$ & $4.000$ \\
				&  & prop & $1.000$ & $1.000$ & $1.000$ & $1.000$ & $1.000$ & $1.000$ & $%
				1.000$ & $1.000$ & $0.995$ & $1.000$ & $1.000$ & $1.000$ & $0.985$ & $0.995$
				& $1.000$ & $1.000$ \\
				S2 &  & mean & $1.000$ & $1.000$ & $1.000$ & $1.000$ & $2.000$ & $2.000$ & $%
				2.000$ & $2.000$ & $3.000$ & $3.000$ & $2.010$ & $2.000$ & $4.000$ & $4.000$
				& $3.835$ & $3.665$ \\
				&  & prop & $1.000$ & $1.000$ & $1.000$ & $1.000$ & $1.000$ & $1.000$ & $%
				1.000$ & $1.000$ & $1.000$ & $1.000$ & $0.001$ & $0.000$ & $1.000$ & $1.000$
				& $0.910$ & $0.825$ \\ \cline{1-19}
				&  &  & \multicolumn{16}{c|}{$n=1000$} \\ \cline{1-19}
				S1 & $3$ & mean & $1.000$ & $1.000$ & $1.000$ & $1.000$ & $2.050$ & $2.000$
				& $2.000$ & $2.000$ & $3.000$ & $3.000 $ & $3.000$ & $3.000$ & $4.060$ & $%
				4.045$ & $4.025$ & $4.020$ \\
				&  & prop & $1.000$ & $1.000$ & $1.000$ & $1.000$ & $0.990$ & $1.000$ & $%
				1.000$ & $1.000$ & $1.000$ & $1.000$ & $1.000$ & $1.000$ & $0.980$ & $0.985$
				& $0.990$ & $0.995$ \\
				& $4$ & mean & $1.000$ & $1.000$ & $1.000$ & $1.000$ & $2.000$ & $2.000$ & $%
				2.000$ & $2.000$ & $3.000$ & $3.000$ & $3.000$ & $3.000$ & $4.020$ & $4.000$
				& $4.000$ & $4.000$ \\
				&  & prop & $1.000$ & $1.000$ & $1.000$ & $1.000$ & $1.000$ & $1.000$ & $%
				1.000$ & $1.000$ & $1.000$ & $1.000$ & $1.000$ & $1.000$ & $0.995$ & $1.000$
				& $1.000$ & $1.000$ \\
				& $5$ & mean & $1.000$ & $1.000$ & $1.000$ & $1.000$ & $2.000$ & $2.000$ & $%
				2.000$ & $2.000$ & $3.000$ & $3.000$ & $3.000$ & $3.000$ & $4.020$ & $4.000$
				& $4.000$ & $4.000$ \\
				&  & prop & $1.000$ & $1.000$ & $1.000$ & $1.000$ & $1.000$ & $1.000$ & $%
				1.000$ & $1.000$ & $1.000$ & $1.000$ & $1.000$ & $1.000$ & $0.990$ & $1.000$
				& $1.000$ & $1.000$ \\
				S2 &  & mean & $1.000$ & $1.000$ & $1.000$ & $1.000$ & $2.000$ & $2.000$ & $%
				2.000$ & $2.000$ & $3.000$ & $3.000$ & $3.000$ & $2.030$ & $4.000$ & $4.000$
				& $4.000$ & $3.210$ \\
				&  & prop & $1.000$ & $1.000$ & $1.000$ & $1.000$ & $1.000$ & $1.000$ & $%
				1.000$ & $1.000$ & $1.000$ & $1.000$ & $1.000$ & $0.030$ & $1.000$ & $1.000$
				& $1.000$ & $0.605$ \\ \cline{1-19}
			\end{tabular}%
		\end{adjustbox}
\end{center}
\end{table}

\begin{table}[tbph]
\caption{The mean of $\protect\widehat{K}_{1}$ and the proportion (prop) of
correctly estimating $K_{0}$ among $200$ simulated datasets when data are
generated from DCSBMs.}
\label{TAB:Khat1DCSBM}
\par
\begin{center}
\begin{adjustbox}{width=\textwidth}
			\begin{tabular}{|l|c|c|cccc|cccc|}
				\cline{1-11}
				&  &  & \multicolumn{4}{c}{$n=500$} & \multicolumn{4}{|c|}{$n=1000$} \\
				\cline{1-11}
				&  $\rho $ &  & $K_{0}=1$ & $K_{0}=2$ & $K_{0}=3$ & $K_{4}=4$ & $K_{0}=1$ & $K_{0}=2$
				& $K_{0}=3$ & $K_{4}=4$ \\ \cline{1-11}
				S1 & $3$ & mean & $1.000$ & $2.095$ & $3.070$ & $3.675$ & $1.000$ & $2.050$
				& $3.000$ & $4.060$ \\
				&  & prop & $1.000$ & $0.980$ & $0.980$ & $0.380$ & $1.000$ & $0.990$ & $%
				1.000$ & $0.980$ \\
				& $4$ & mean & $1.000$ & $2.090$ & $3.025$ & $4.175$ & $1.000$ & $2.000$ & $%
				3.000$ & $4.020$ \\
				&  & prop & $1.000$ & $0.980$ & $0.990$ & $0.915$ & $1.000$ & $1.000$ & $%
				1.000$ & $0.995$ \\
				& $5$ & mean & $1.000$ & $2.035$ & $3.030$ & $4.045$ & $1.000$ & $2.000$ & $%
				3.000$ & $4.045$ \\
				&  & prop & $1.000$ & $0.990$ & $0.995$ & $0.985$ & $1.000$ & $1.000$ & $%
				1.000$ & $0.985$ \\
				S2 &  & mean & $1.000$ & $2.000$ & $3.035$ & $4.005$ & $1.000$ & $2.000$ & $%
				3.000$ & $4.000$ \\
				&  & prop & $1.000$ & $1.000$ & $0.995$ & $0.995$ & $1.000$ & $1.000$ & $%
				1.000$ & $1.000$ \\ \cline{1-11}
			\end{tabular}%
		\end{adjustbox}
\end{center}
\end{table}

Based on the above results, we let $c_{h}=1.0$ for the PLR2 estimator. For
evaluating the performance of the six methods at different sparsity levels,
we let $\rho =0.5,1,2,3,4,5,6$ for design S1, so that the average expected
degree ranges from 7.0 to 83.9, for instance, at $K_{0}=4$ and $n=500$ for
the DCSBMs. Figure \ref{FIG:prop} shows the proportions of correctly
estimating $K_{0}$ among $200$ simulated datasets versus the values of $\rho
$ for the six methods: PLR1 (solid lines), PLR2 (dash-dot lines), LRBIC
(dashed lines), NCV (dotted lines), ECV (thin dash-dot lines) and BHMC (thin
dotted lines), when data are simulated from design S1 with $K_{0}=2,3,4$ and
$n=500$. The results for the SBMs and DCSBMs are shown in the left and right
panels, respectively. We observe that our proposed methods PLR1 and PLR2
have similar performance with PLR2 moderately better when $K_{0}=2$.
Moreover, PLR1 and PLR2 have larger proportions of correctly estimating $%
K_{0}$ than the other four methods at small values of $\rho $.\ This
indicates that PLR1 and PLR2 outperform other methods for semi-dense
designs. The BHMC method performs better than LRBIC, NCV and ECV at $%
K_{0}=2,3$, but its performance becomes inferior to that of the other three
methods when $K_{0}=4$. It is worth noting that for larger $K_{0}$, it
correspondingly requires a larger $\rho $ in order to successfully estimate $%
K_{0}$. When $\rho $ is sufficiently large, eventually all methods can
successfully estimate $K_{0}$. Compared to the other four methods, PLR1 and
PLR2 require less constraints on the sparsity level $\rho $ in order to
correctly estimate $K_{0}$. For example, for the DCSBMs with $K_{0}=4$, the
proportions of correctly estimating $K_{0}$ are 0.38 for PLR1 and PLR2,
whereas the proportions are close to zero for other methods at $\rho =3$.
For the DCSBMs with $K_{0}=2$, the proportions are 0.71 and 0.89 for PLR1
and PLR2, respectively, and they are less than 0.1 for other methods at $%
\rho =0.5$.

For further demonstration, Tables \ref{TAB:S1S2K2}-\ref{TAB:S1S2K4} given in the Supplemental Materials report
the mean of the estimated number of communities and the proportion (prop) of
correctly estimating $K_{0}$ for designs S1 and S2 with $n=500$. For S1, we
observe the same pattern as shown in Figure \ref{FIG:prop}. For S2 in which
all entries of $\boldsymbol{B}$ are different, the six methods have
comparable performance.

\begin{figure}[tbp]
\caption{The proportions of correctly estimating $K_{0}$ versus the values
of $\protect\rho $ for the six methods, when data are simulated from design
S1 with $K_{0}=2,3,4$ and $n=500$. }
\label{FIG:prop}
\centering
$
\begin{array}{cc}
\includegraphics[width=6cm,height=6cm]{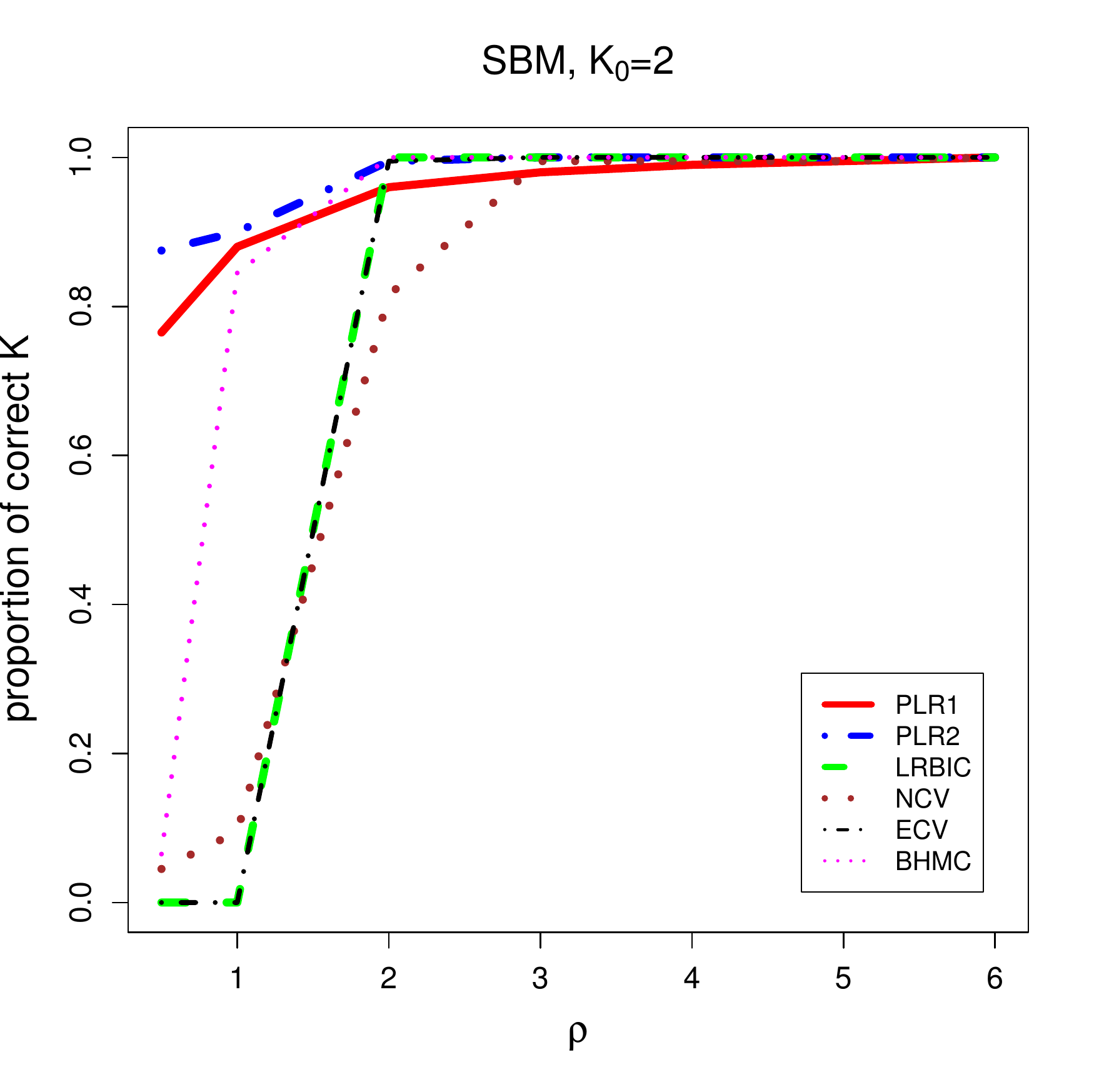} & %
\includegraphics[width=6cm,height=6cm]{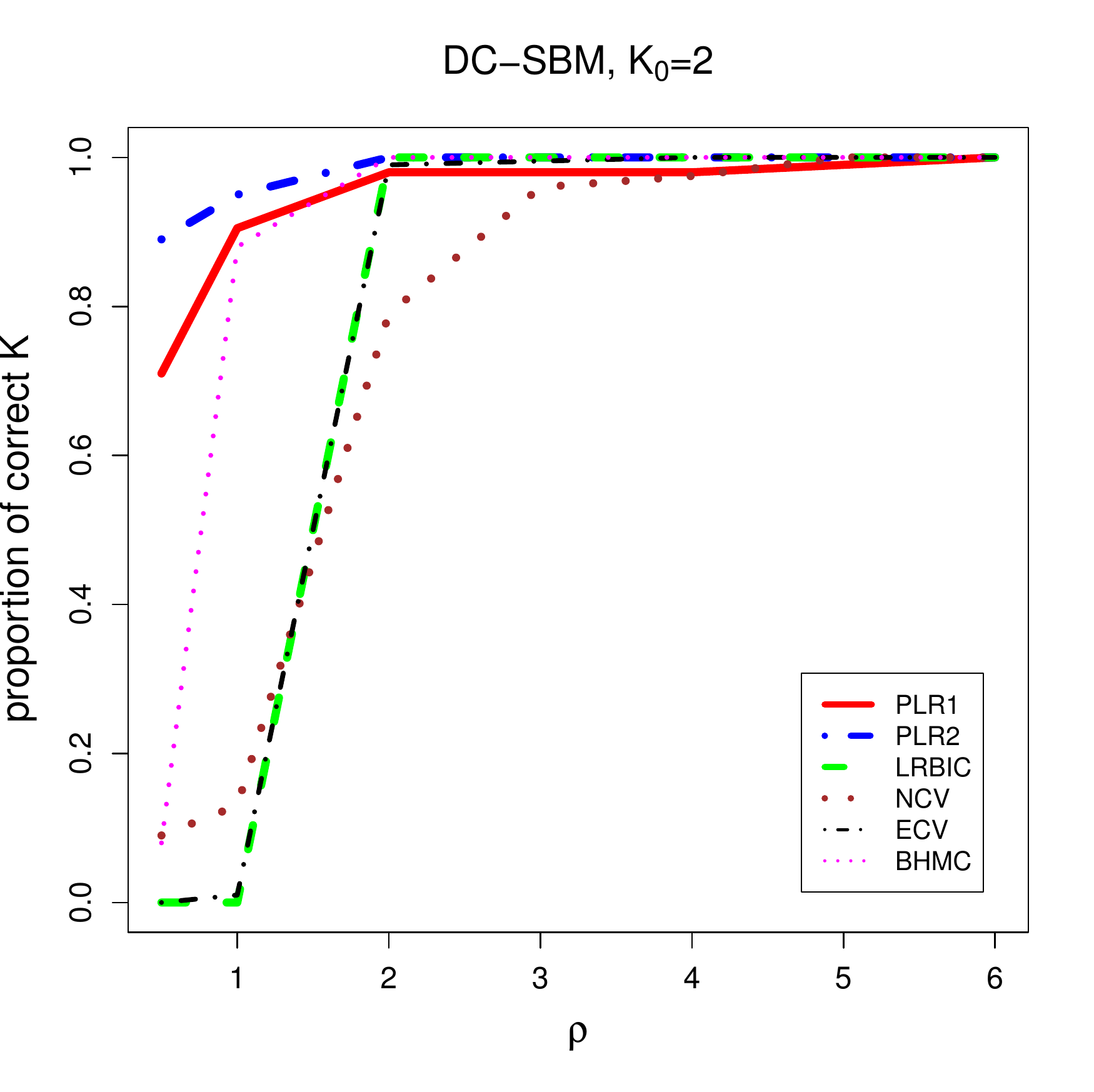} \\
\includegraphics[width=6cm,height=6cm]{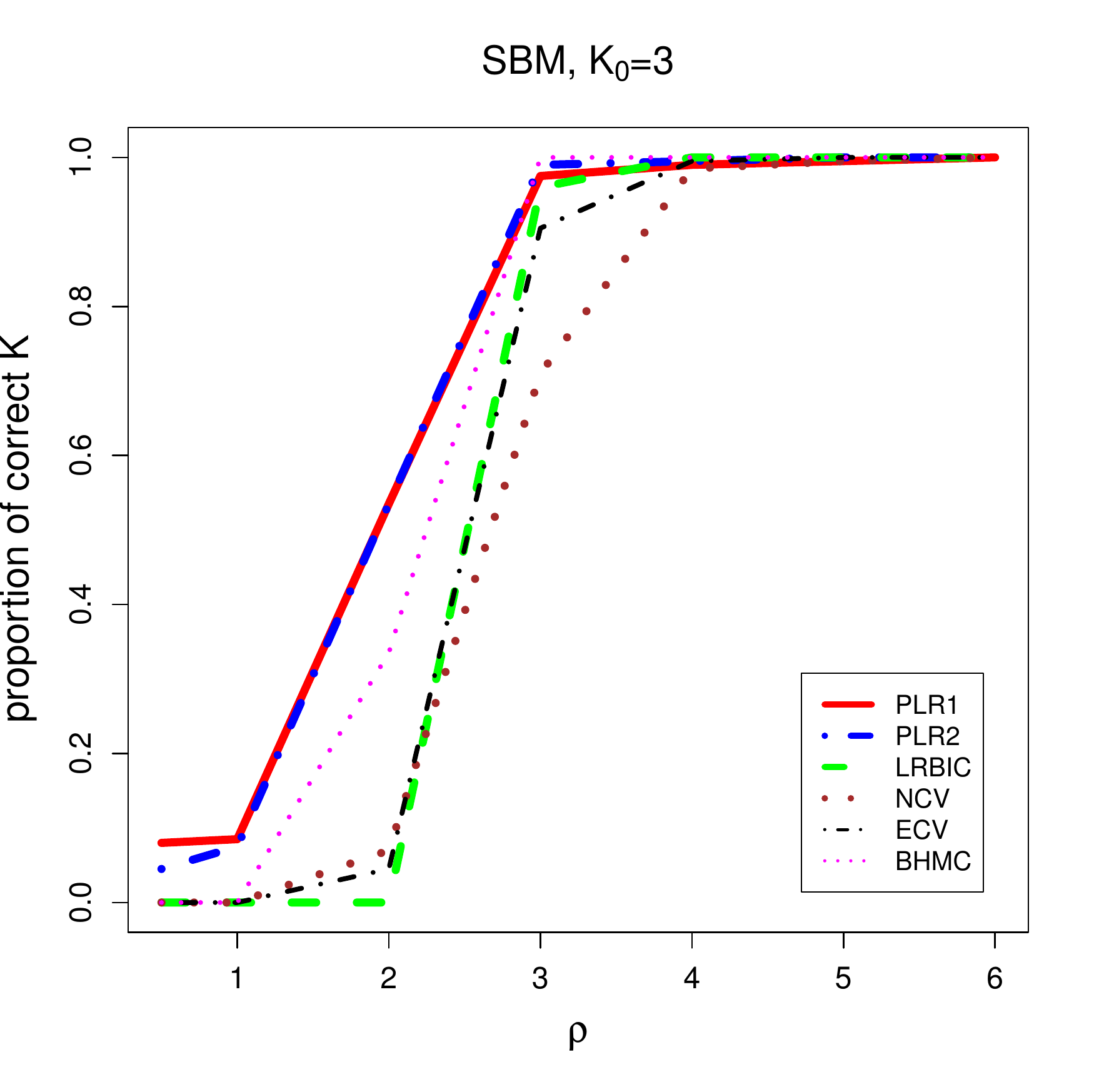} & %
\includegraphics[width=6cm,height=6cm]{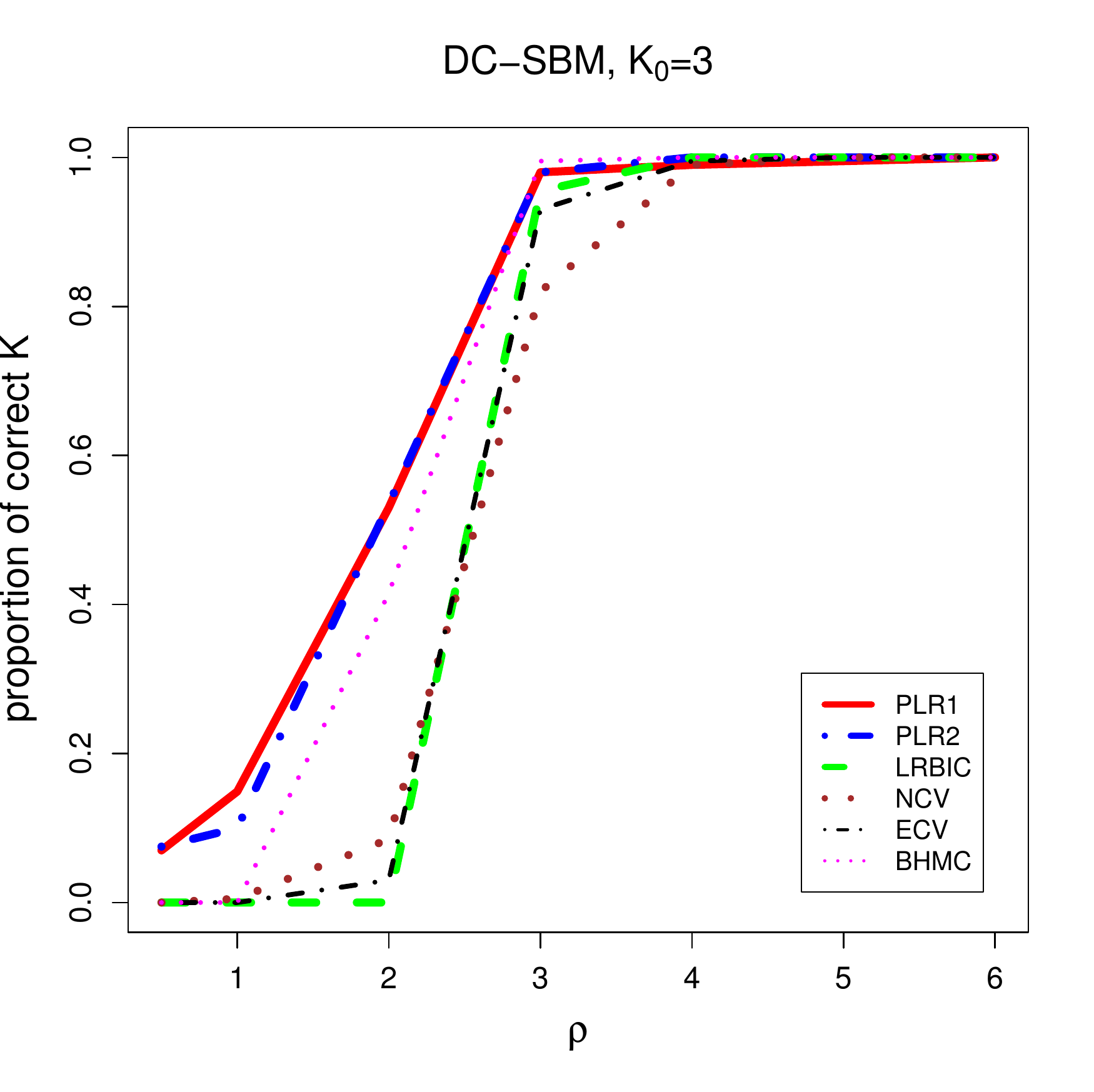} \\
\includegraphics[width=6cm,height=6cm]{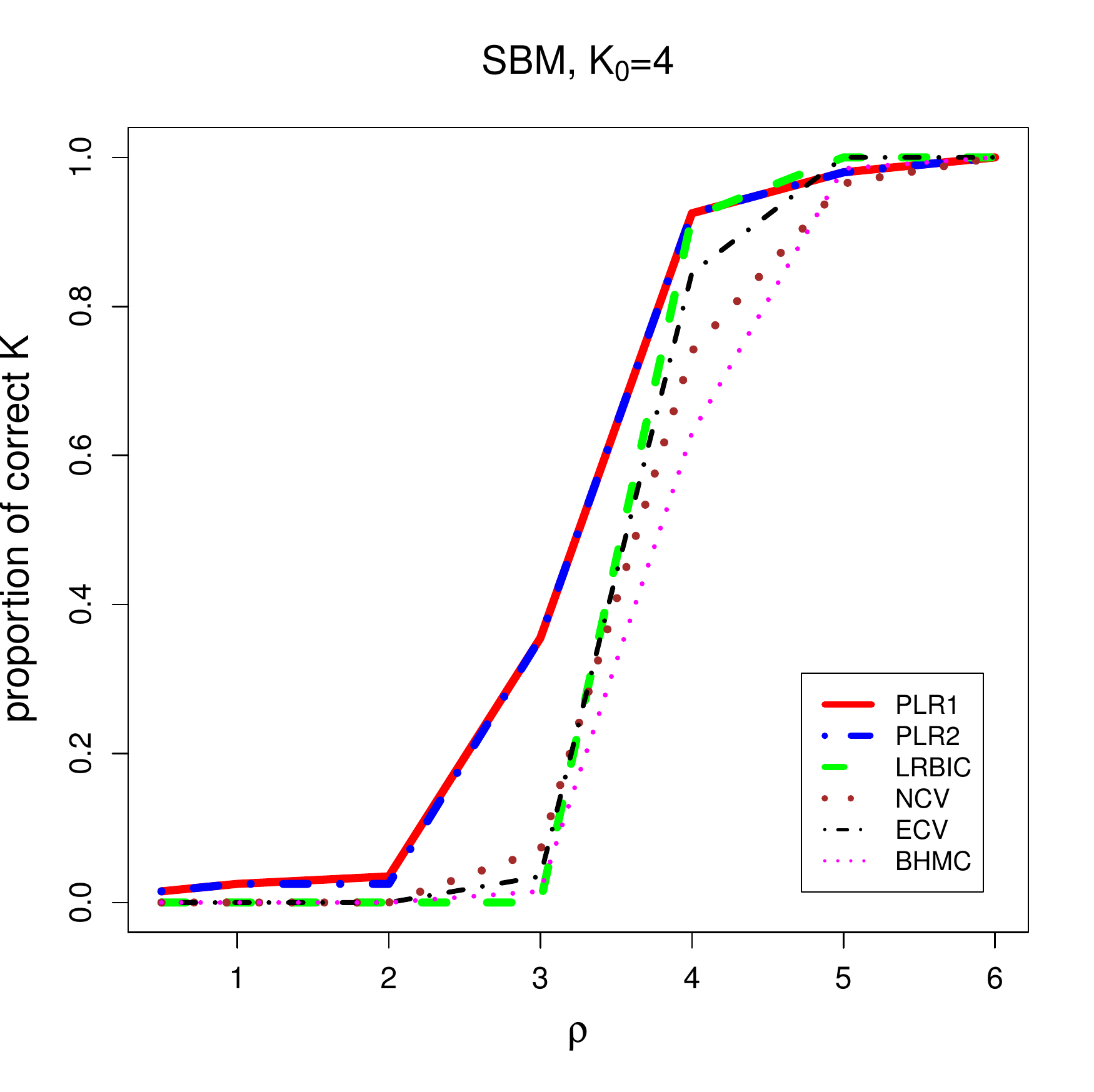} & %
\includegraphics[width=6cm,height=6cm]{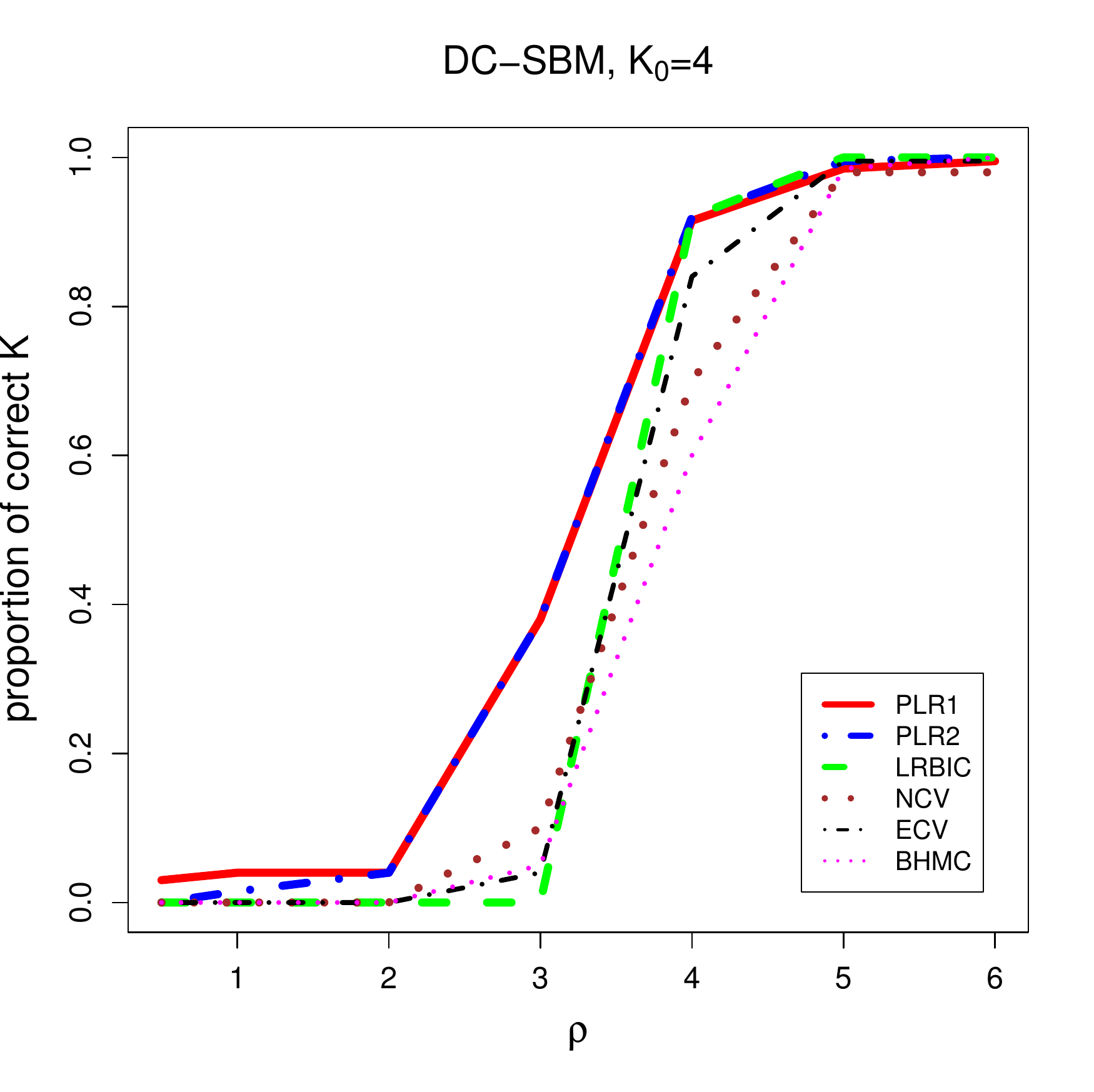}%
\end{array}
$\vskip -1cm
\end{figure}
\section{Real Data Examples\label{sec:real}}

In this section, we evaluate the performance of our method on several
real-world networks.

\subsection{Jazz musicians network}

We apply the methods to analyze the collaboration network of Jazz musicians.
The data are obtained from \textit{The Red Hot Jazz Archive} digital
database (www.redhotjazz.com). In our analysis, we include 198 bands that
performed between 1912 and 1940. We study the community structure of the
band network in which there are 198 nodes representing bands and 2742
unweighted edges indicating at least one common musician between two bands.
The left panel of Figure \ref{Fig:jazz} shows the degree distribution for
the jazz band network. The minimal, average and maximum degrees of this
network are 1.0, 27.7 and 100.0, respectively. Moreover, the distribution of
degrees spreads over the range from 1 to 62 with four degree values outside
this range. This indicates that the node degrees are highly varying for this
network.

Let $K_{\max }=10$ for all methods. We apply our proposed PLR1 and PLR2
methods to estimate the number of communities and obtain that $\widehat{K}%
_{1}=3$ and $\widehat{K}_{2}=3$, so that three communities are identified by
both methods. For further illustration, the right panel of Figure \ref%
{Fig:jazz} depicts the band network with 198 nodes divided into three
communities. The results confirm the community structure mentioned in \cite%
{GD03} that the band network is divided into two large communities based on
geographical locations where the bands recorded, and the largest community
also splits into two communities due to a racial segregation. Moreover, we
obtain the estimated edge probabilities within communities which are $%
\widehat{B}_{kk}=0.349,0.297,0.358$ for $k=1,2,3$, respectively, and edge
probabilities between communities which are $\widehat{B}_{12}=0.029$, $%
\widehat{B}_{13}=0.087$ and $\widehat{B}_{23}=0.007$. Lastly, we obtain the
estimated number of communities as $8$, $3$, $6$ and $7$, respectively, by
the LRBIC, NCV, ECV and BHMC methods.
\begin{figure}[tbp]
\caption{Left panel shows the degree distribution; right panel depicts the
jazz band network with three communities. }
\label{Fig:jazz}
\centering
\vspace{0.5cm} $
\begin{array}{cc}
\textbf{Degree distribution of jazz band network} \vspace{-0.4cm} & \textbf{%
\ Jazz band network} \vspace{-0.4cm} \\
\includegraphics[width = 0.45\linewidth]{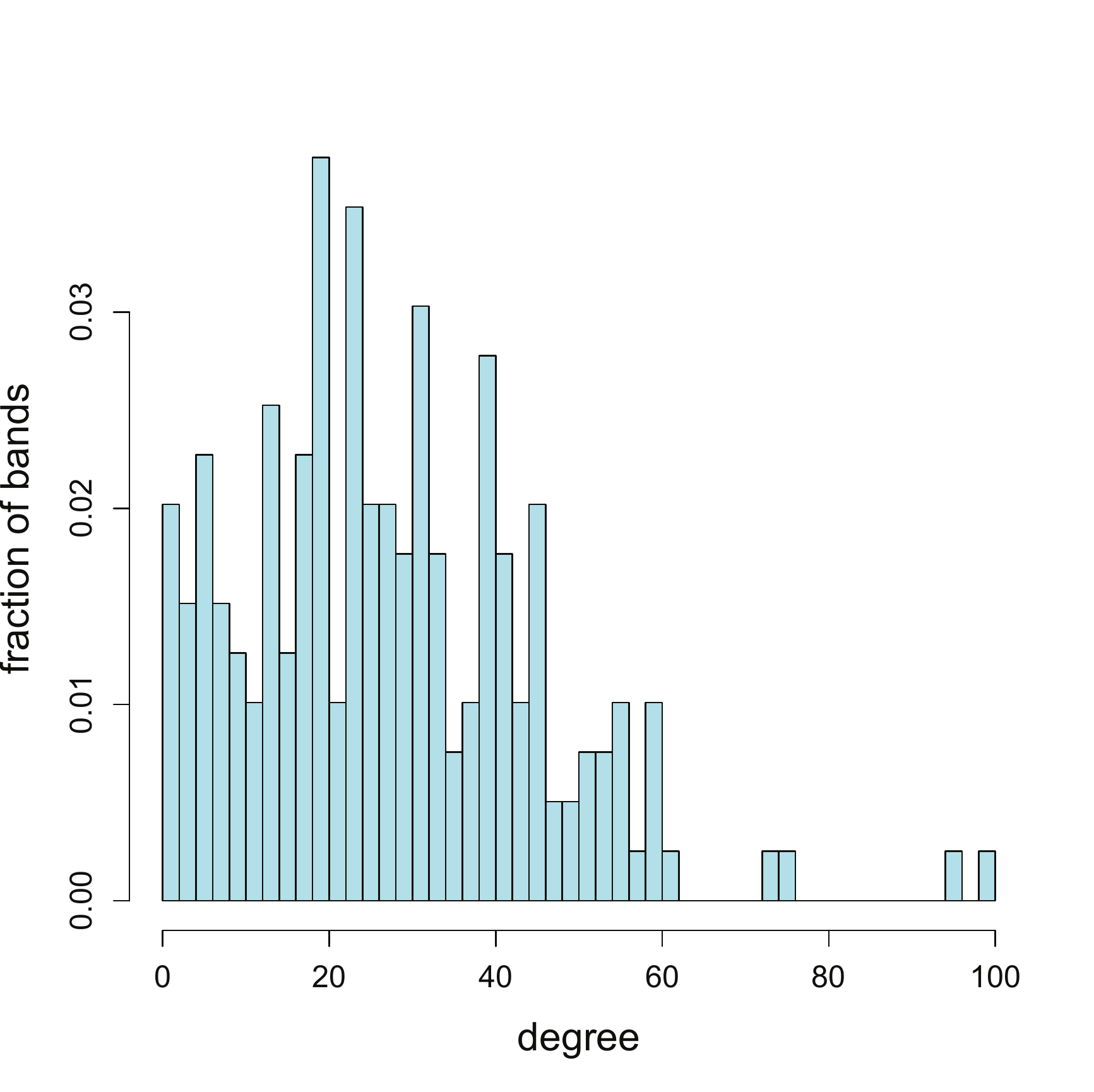} & %
\includegraphics[width = 0.45\linewidth]{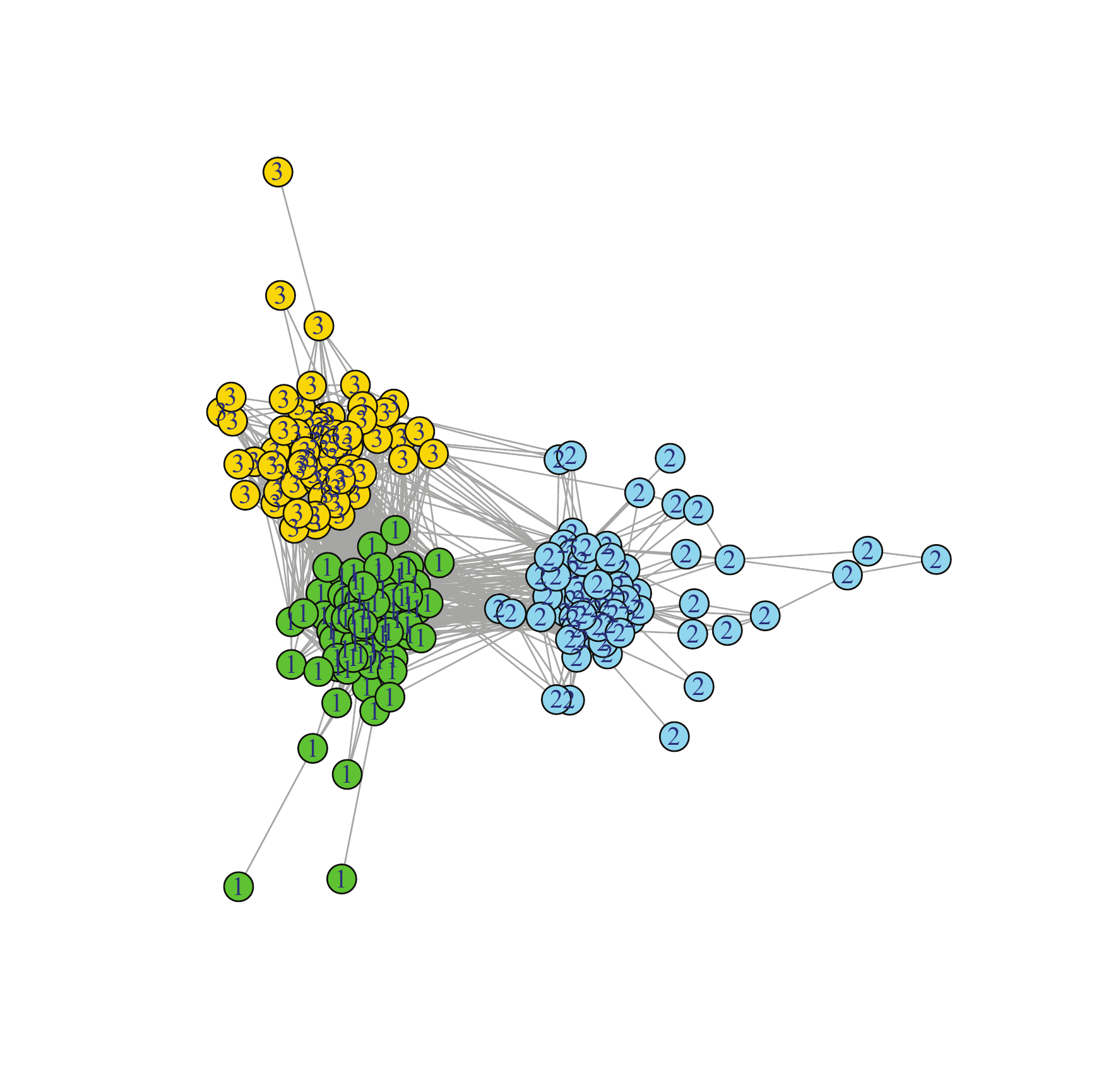}%
\end{array}
$%
\end{figure}

\subsection{Political books network and Facebook friendship network}

We apply our methods to a network of US political books
(available at www.orgnet.com), and to a large social network which contains friendship
data of Facebook users (available at www.snap.stanford.edu). The detailed descriptions of the data applications as well as the numerical results are given in Section \ref{sec:add_application} of the Supplemental Materials. 

\section{Conclusion}

\label{sec:concl}We propose a new pseudo conditional likelihood ratio method
for selecting the number of communities in DCSBMs. The method can be
naturally applied to SBMs. For estimating the model, we consider the
spectral clustering together with a binary segmentation algorithm. This
estimation approach enables us to establish the limiting distribution of the
pseudo likelihood ratio when the model is under-fitted, and derive the upper
bound for it when the model is over-fitted. Based on these properties, we
show the consistency of our estimator for the true number of communities.
Our method is computationally fast as the estimation is based on spectral
clustering, and it also has appealing theoretical properties for the
semi-dense and degree-corrected designs. Moreover, our numerical results
show that the proposed method has good finite sample performance in various
simulation designs and real data applications, and it outperforms several
other popular methods in semi-dense networks.

\section*{Acknowledgement}

The research of Ma is supported in part by the U.S. NSF grant DMS-17-12558.
Su acknowledges the funding support provided by the Lee Kong Chian Fund for
Excellence. Zhang acknowledges the funding support provided by the Singapore
Ministry of Education Tier 2 grant under grant no. MOE2018-T2-2-169 and the
Lee Kong Chian fellowship.

\section*{Supplemental Materials}

\label{SM} Supplemental Materials include more details on the algorithms,
additional simulation and real application results, and the proofs of the main results in the
paper and some technical lemmas.

\newpage
\appendix

\begin{center}
	\huge{
	Supplemental Materials for ``Determining the Number of Communities in
	Degree-corrected Stochastic Block Models"}
\end{center}

\begin{abstract}
This supplement includes five sections. Section A contains more details on
the algorithms. Sections B and C report some additional simulation and real application results.
Section D contains the proofs of the main results in the paper. Section E
provides some technical lemmas and their proofs used in the proofs of the
main results.
	
	\noindent \textbf{Key words and phrases: } Clustering, community detection,
	degree-corrected stochastic block model, k-means, regularization.\vspace{2mm}
\end{abstract}
\setcounter{page}{1} \renewcommand\thesection{\Alph{section}} %

\section{More details on Algorithms \protect\ref{algo:1} and \protect\ref%
	{algo:2}}

\label{sec:appago}

\subsection{Estimators $\hat{P}_{ij}(\hat{Z}_K)$ and $\hat{P}_{ij}(\hat{Z}%
	_K^b)$}

By \cite{WSW2016}, for a given number of communities $K$ and a generic
estimator $\hat{Z}_{K}$ of the community memberships with corresponding
estimated communities $\{\widehat{\mathcal{C}}_{k,K}\}_{k=1}^{K}$, the
maximum likelihood estimators (MLEs) for $\theta _{i}$ and $B_{kl}(\hat{Z}%
_{K})$ in DCSBM are $\hat{\theta}_{i}=\frac{\hat{d}_{i}\hat{n}_{k,K}}{%
	\sum_{i^{\prime }\in \widehat{\mathcal{C}}_{k,K}}\hat{d}_{i^{\prime }}}$ for
$i\in \widehat{\mathcal{C}}_{k,K}$ and $\hat{B}_{kl}(\hat{Z}_{K})=\frac{\hat{%
		O}_{kl,K}}{\hat{n}_{kl,K}}$ for $k,l=1,\cdots ,K$, respectively, where $\hat{%
	n}_{k,K}=\sum_{i=1}^{n}1\{[\hat{Z}_{K}]_{ik}=1\}$,
\begin{equation}
\hat{O}_{kl,K}=\sum_{i=1}^{n}\sum_{j\neq i}1\{[\hat{Z}_{K}]_{ik}=1,[\hat{Z}%
_{K}]_{jl}=1\}A_{ij};  \label{eq:O}
\end{equation}%
\begin{align}
\hat{n}_{kl,K}=& \sum_{i=1}^{n}\sum_{j\neq i}1\{[\hat{Z}_{K}]_{ik}=1,[\hat{Z}%
_{K}]_{jl}=1\}  \notag  \label{eq:nk} \\
=&
\begin{cases}
\hat{n}_{k,K}\hat{n}_{l,K} & \quad \text{if}\quad k\neq l \\
\hat{n}_{k,K}(\hat{n}_{k,K}-1) & \quad \text{if}\quad k=l.%
\end{cases}%
\end{align}%
Therefore, for $i\in \widehat{\mathcal{C}}_{k,K}$ and $j\in \widehat{%
	\mathcal{C}}_{l,K}$, when $k\neq l$,
\begin{align*}
\hat{P}_{ij}(\hat{Z}_{K})=& \hat{\theta}_{i}\hat{\theta}_{j}\hat{B}_{kl}(%
\hat{Z}_{K})=\frac{\hat{O}_{kl,K}\hat{d}_{i}\hat{d}_{j}}{(\sum_{i^{\prime
		}\in \widehat{\mathcal{C}}_{k,K}}\hat{d}_{i^{\prime }})(\sum_{j^{\prime }\in
		\widehat{\mathcal{C}}_{l,K}}\hat{d}_{j^{\prime }})} \\
=& \frac{\hat{O}_{kl,K}\hat{d}_{i}\hat{d}_{j}}{(\sum_{l^{\prime }=1}^{K}\hat{%
		O}_{kl^{\prime },K})(\sum_{l^{\prime }=1}^{K}\hat{O}_{ll^{\prime },K})};
\end{align*}%
when $k=l$ and $i,j\in \widehat{\mathcal{C}}_{k,K}$,
\begin{equation*}
\hat{P}_{ij}(\hat{Z}_{K})=\frac{\hat{O}_{kk,K}\hat{d}_{i}\hat{d}_{j}}{%
	\sum_{i^{\prime },j^{\prime }\in \widehat{\mathcal{C}}_{k,K},i^{\prime }\neq
		j^{\prime }}\hat{d}_{i^{\prime }}\hat{d}_{j^{\prime }}}.
\end{equation*}%
We can compute $\hat{P}_{ij}(\hat{Z}_{K}^{b})$ in the same manner by
replacing $\hat{Z}_{K}$ in the above procedure by $\hat{Z}_{K}^{b}$.

\subsection{More details on the k-means algorithm}

In Algorithm \ref{algo:2}, we propose to estimate $\hat{Z}_K$ and $\hat{Z}%
_{K+1}^b$ by the k-means algorithm. Let $\{\beta_i\}_{i \in \mathcal{C}}$ be
a sequence of $d_\beta \times 1$ vectors. The k-means algorithm with $K$
centroids divides $\{\beta_i\}_{i \in \mathcal{C}}$ into $K$ clusters via
solving the following minimization problem:
\begin{align}  \label{eq:kmeans}
(\alpha_1^*,\cdots,\alpha_K^*) = \argmin_{\alpha_1,\cdots,\alpha_K}\sum_{i
	\in \mathcal{C}} \min_{1 \leq k \leq K}||\beta_i - \alpha_k||^2,
\end{align}
where the $i$-th node is classified into cluster $k$ if $k = \argmin_{1\leq
	l \leq K}||\beta_i - \alpha_{l}^*||$ and if there exists a tie, i.e., $%
\argmin_{1\leq l \leq K}||\beta_i - \alpha_{l}^*||$ is not a singleton, then
we denote $k$ as the smallest minimizer. Then, $\hat{Z}_K$ is obtained by
solving \eqref{eq:kmeans} with $\beta_i = \hat{\nu}_{iK}$, $i=1,\cdots,n$ with $K$ centroids. For $\hat{Z}_{K+1}^b$, the binary segmentation step is implemented via
solving \eqref{eq:kmeans} with 2 centroids and $\beta_i = \hat{\nu}_{iK+1}$,
$i \in \widehat{\mathcal{C}}_{k,K}$, for $k = 1,\cdots,K$.

In Section \ref{sec:est}, we define $(Z_K,Z_K^b)$ by applying Algorithm \ref%
{algo:2} on $\nu_{iK}$. In view of Theorem \ref{thm:id3}(2), $\nu _{iK}$
takes $L_K$ distinct values $(\bar{\nu}_{1K},\cdots,\bar{\nu}_{L_KK})$. Let
\begin{equation*}
\pi_{l,K} = \#\{i:\nu_{iK} = \bar{\nu}_{lK}\}/n \geq \inf_{1\leq k \leq
	K_0}\pi_{kn}
\end{equation*}
and $g_{iK}$ be the membership for node $i$ obtained this way, i.e., $g_{iK}=%
\argmin_{1\leq k\leq K}||\nu _{iK}-\alpha _{k}^{\ast }||$ where
\begin{align}
\{\alpha _{k}^{\ast }\}_{k=1}^{K}=& \argmin_{\alpha _{1},\cdots ,\alpha
	_{K}}n^{-1}\sum_{i=1}^{n}\min_{1\leq k\leq K}||\nu_{iK}-\alpha _{k}||^{2}
\notag \\
=& \argmin_{\alpha _{1},\cdots ,\alpha _{K}}\sum_{l=1}^{L_K}\pi
_{l,K}\min_{1\leq k\leq K}||\bar{\nu} _{lK}-\alpha _{k}||^{2}.  \label{eq:ZK}
\end{align}%
Then $[Z_{K}]_{ik}=1$ if $g_{iK}=k$, $[Z_{K}]_{ik}=0$ otherwise, and $%
\mathcal{C}_{k,K}=\{i:g_{iK}=k\}$. We define $Z_{K+1}^{b}$ for $K=1,\cdots
,K_{0}-1$ as follows.


\begin{enumerate}
	\item Given $\{\mathcal{C}_{k,K}\}_{k=1}^{K}$, let $\widetilde{\mathcal{C}}%
	_{k,K}^{l}=\mathcal{C}_{k,K}\cap  G_{l,K+1}$, for $l=1,\cdots ,L_{K}$,%
	\footnote{%
		As can be shown, $\widetilde{\mathcal{C}}_{k,K}^{l} = G_{l,K+1}$ or $%
		\emptyset.$} where $G_{l,K+1}$ is defined in Theorem \ref{thm:id3}(2). We divide each $\mathcal{C}_{k,K}$ into two subgroups by
	applying the k-means  algorithm to $\{\nu _{iK+1}\}_{i\in \mathcal{C}_{k,K}}$
	with two centroids. Denote  the two subgroups as $\mathcal{C}_{k,K}(1)$ and $%
	\mathcal{C}_{k,K}(2)$. Note that, by the proof of Theorem \ref{thm:id3}(2),
	for $i \in \tilde{ \mathcal{C}}^l_{k,K}$, $\nu _{iK+1}$ take the same value.
	
	\item For each $k=1,\cdots ,K$, compute
	\begin{equation}
	Q_{K}(k)=\frac{\Phi (\mathcal{C}_{k,K})-\Phi (\mathcal{C}_{k,K}(1))-\Phi (%
		\mathcal{C}_{k,K}(2))}{\#\mathcal{C}_{k,K}},  \label{eq:Q}
	\end{equation}%
	where for an arbitrary index set $\mathcal{C}$, $\Phi (\mathcal{C}%
	)=\sum_{i\in \mathcal{C}}||\nu _{iK+1}-\frac{\sum_{i\in \mathcal{C}}\nu
		_{iK+1}}{\#\mathcal{C}}||^{2}.$
	
	\item Choose $k^{\ast }=\argmax_{1\leq k\leq K}Q_{K}(k)$. Denote
	\begin{equation*}
	\{\mathcal{C}_{k,K+1}^{b}\}_{k=1}^{K+1}=\{\{\mathcal{C}_{k,K}\}_{k<k^{\ast
	}},\mathcal{C}_{k^{\ast },K}(1),\newline
	\{\mathcal{C}_{k,K}\}_{k>k^{\ast }},\mathcal{C}_{k^{\ast },K}(2)\}
	\end{equation*}%
	as the new groups in $Z_{K+1}^{b}$.
\end{enumerate}

\bigskip

\section{Additional simulation results}

\label{sec:add_simulations}
Tables \ref{TAB:Khat2SBM} and \ref{TAB:Khat1SBM} given below report the mean of $\widehat{K}_{2}$ and $\widehat{K}_{1}$ by the PLR2 and PLR1 methods, respectively, and the proportion (prop) of correctly estimating $%
K_{0}$ among $200$ simulated datasets when data are generated from the
SBMs described in Section \ref{sec:DGP}, for $n=500,1000$ and $K_{0}=1,2,3,4$.

\begin{table}[tbph]
	\caption{The mean of $\protect\widehat{K}_{2}$ and the proportion (prop) of
		correctly estimating $K_{0}$ among $200$ simulated datasets when data are
		generated from SBMs$.$}
	\label{TAB:Khat2SBM}
	\par
	\begin{center}
		\begin{adjustbox}{width=1.1\textwidth,height = 3.8cm}
			\hskip -1.5cm					
			\begin{tabular}{|l|c|c|cccc|cccc|cccc|cccc|}
				\cline{1-19}
				&  &  & \multicolumn{4}{c|}{$K_{0}=1$} & \multicolumn{4}{c|}{$K_{0}=2$} &
				\multicolumn{4}{c|}{$K_{0}=3$} & \multicolumn{4}{c|}{$K_{4}=4$} \\
				\cline{1-19}
				& $\rho $ & $c_{h}$ & \multicolumn{1}{c}{$0.5$} & $1.0$ & $1.5$ & $2.0$ &
				\multicolumn{1}{c}{$0.5$} & $1.0$ & $1.5$ & $2.0$ & \multicolumn{1}{c}{$0.5$}
				& $1.0$ & $1.5$ & $2.0$ & $0.5$ & $1.0$ & $1.5$ & $2.0$ \\ \cline{1-19}
				&  &  & \multicolumn{16}{c|}{$n=500$} \\ \cline{1-19}
				S1 & $3$ & mean & $1.035$ & $1.000$ & $1.000$ & $1.000$ & $2.025$ & $2.000$
				& $2.000$ & $2.000$ & $3.060$ & $3.060$ & $3.000$ & $3.000$ & $3.465$ & $%
				3.465$ & $3.430$ & $3.355$ \\
				&  & prop & $0.995$ & $1.000$ & $1.000$ & $1.000$ & $0.995$ & $1.000$ & $%
				1.000$ & $1.000$ & $0.990$ & $0.990$ & $1.000$ & $1.000$ & $0.355$ & $0.355$
				& $0.350$ & $0.330$ \\
				& $4$ & mean & $1.000$ & $1.000$ & $1.000$ & $1.000$ & $2.030$ & $2.000$ & $%
				2.000$ & $2.000$ & $3.115$ & $3.015$ & $3.000$ & $3.000$ & $4.085$ & $4.085$
				& $4.085$ & $4.005$ \\
				&  & prop & $1.000$ & $1.000$ & $1.000$ & $1.000$ & $0.995$ & $1.000$ & $%
				1.000$ & $1.000$ & $0.975$ & $0.995$ & $1.000$ & $1.000$ & $0.925$ & $0.925$
				& $0.925$ & $0.925$ \\
				& $5$ & mean & $1.000$ & $1.000$ & $1.000$ & $1.000$ & $2.000$ & $2.000$ & $%
				2.000$ & $2.000$ & $3.000$ & $3.000$ & $3.000$ & $3.000$ & $4.060$ & $4.060$
				& $4.060$ & $4.000$ \\
				&  & prop & $1.000$ & $1.000$ & $1.000$ & $1.000$ & $1.000$ & $1.000$ & $%
				1.000$ & $1.000$ & $1.000$ & $1.000$ & $1.000$ & $1.000$ & $0.980$ & $0.980$
				& $0.980$ & $1.000$ \\
				S2 &  & mean & $1.000$ & $1.000$ & $1.000$ & $1.000$ & $2.000$ & $2.000$ & $%
				2.000$ & $2.000$ & $3.000$ & $3.000$ & $2.035$ & $2.000$ & $4.000$ & $3.995$
				& $3.820$ & $3.620$ \\
				&  & prop & $1.000$ & $1.000$ & $1.000$ & $1.000$ & $1.000$ & $1.000$ & $%
				1.000$ & $1.000$ & $1.000$ & $1.000$ & $0.035$ & $0.000$ & $1.000$ & $0.995$
				& $0.895$ & $0.795$ \\ \cline{1-19}
				&  &  & \multicolumn{16}{c|}{$n=1000$} \\ \cline{1-19}
				S1 & $3$ & mean & $1.000$ & $1.000$ & $1.000$ & $1.000$ & $2.055$ & $2.000$
				& $2.000$ & $2.000$ & $3.040$ & $3.005$ & $3.000$ & $3.000$ & $4.080$ & $%
				4.050$ & $4.020$ & $3.990$ \\
				&  & prop & $1.000$ & $1.000$ & $1.000$ & $1.000$ & $0.990$ & $1.000$ & $%
				1.000$ & $1.000$ & $0.985$ & $0.995$ & $1.000$ & $1.000$ & $0.980$ & $0.990$
				& $0.995$ & $0.995$ \\
				& $4$ & mean & $1.000$ & $1.000$ & $1.000$ & $1.000$ & $2.000$ & $2.000$ & $%
				2.000$ & $2.000$ & $3.015$ & $3.000$ & $3.000$ & $3.000$ & $4.020$ & $4.000$
				& $4.000$ & $4.000$ \\
				&  & prop & $1.000$ & $1.000$ & $1.000$ & $1.000$ & $1.000$ & $1.000$ & $%
				1.000$ & $1.000$ & $0.995$ & $1.000$ & $1.000$ & $1.000$ & $0.995$ & $1.000$
				& $1.000$ & $1.000$ \\
				& $5$ & mean & $1.000$ & $1.000$ & $1.000$ & $1.000$ & $2.000$ & $2.000$ & $%
				2.000$ & $2.000$ & $3.045$ & $3.000$ & $3.000$ & $3.000$ & $4.030$ & $4.020$
				& $4.000$ & $4.000$ \\
				&  & prop & $1.000$ & $1.000$ & $1.000$ & $1.000$ & $1.000$ & $1.000$ & $%
				1.000$ & $1.000$ & $0.990$ & $1.000$ & $1.000$ & $1.000$ & $0.990$ & $0.995$
				& $1.000$ & $1.000$ \\
				S2 &  & mean & $1.000$ & $1.000$ & $1.000$ & $1.000$ & $2.000$ & $2.000$ & $%
				2.000$ & $2.000$ & $3.000$ & $3.000$ & $3.000$ & $2.035$ & $4.000$ & $4.000$
				& $4.000$ & $3.320$ \\
				&  & prop & $1.000$ & $1.000$ & $1.000$ & $1.000$ & $1.000$ & $1.000$ & $%
				1.000$ & $1.000$ & $1.000$ & $1.000$ & $1.000$ & $0.035$ & $1.000$ & $1.000$
				& $1.000$ & $0.660$ \\ \cline{1-19}
			\end{tabular}%
		\end{adjustbox}		
	\end{center}
\end{table}

\begin{table}[tbph]
	\caption{The mean of $\protect\widehat{K}_{1}$ and the proportion (prop) of
		correctly estimating $K_{0}$ among $200$ simulated datasets when data are
		generated from SBMs.}
	\label{TAB:Khat1SBM}
	\par
	\begin{center}
		\begin{adjustbox}{width=\textwidth}
			\begin{tabular}{|l|c|c|cccc|cccc|}
				\cline{1-11}
				&  &  & \multicolumn{4}{c}{$n=500$} & \multicolumn{4}{|c|}{$n=1000$} \\
				\cline{1-11}
				&  $\rho $ &  & $K_{0}=1$ & $K_{0}=2$ & $K_{0}=3$ & $K_{4}=4$ & $K_{0}=1$ & $K_{0}=2$
				& $K_{0}=3$ & $K_{4}=4$ \\ \cline{1-11}
				S1 & $3$ & mean & $1.035$ & $2.095$ & $3.115$ & $3.465$ & $1.000$ & $2.055$
				& $3.040$ & $4.080$ \\
				&  & prop & $0.995$ & $0.980$ & $0.975$ & $0.355$ & $1.000$ & $0.990$ & $%
				0.985$ & $0.980$ \\
				& $4$ & mean & $1.000$ & $2.045$ & $3.060$ & $4.085$ & $1.000$ & $2.000$ & $%
				3.015$ & $4.020$ \\
				&  & prop & $1.000$ & $0.990$ & $0.990$ & $0.925$ & $1.000$ & $1.000$ & $%
				0.995$ & $0.995$ \\
				& $5$ & mean & $1.000$ & $2.020$ & $3.015$ & $4.060$ & $1.000$ & $2.000$ & $%
				3.045$ & $4.030$ \\
				&  & prop & $1.000$ & $0.995$ & $0.995$ & $0.980$ & $1.000$ & $1.000$ & $%
				0.990$ & $0.990$ \\
				S2 &  & mean & $1.000$ & $2.000$ & $3.110$ & $4.000$ & $1.000$ & $2.000$ & $%
				3.000$ & $4.000$ \\
				&  & prop & $1.000$ & $1.000$ & $0.980$ & $1.000$ & $1.000$ & $1.000$ & $%
				1.000$ & $1.000$ \\ \cline{1-11}
			\end{tabular}%
		\end{adjustbox}
	\end{center}
\end{table}

For further comparisons of the six methods, PLR1, PLR2, LRBIC, NCV, ECV and BHMC, mentioned in Section \ref{sec:background}, Tables \ref{TAB:S1S2K2}-\ref{TAB:S1S2K4} report
the mean of the estimated number of communities and the proportion (prop) of
correctly estimating $K_{0}$ for designs S1 and S2 with $n=500$. For S1, we
observe the same pattern as shown in Figure \ref{FIG:prop}. For S2 in which
all entries of $\boldsymbol{B}$ are different, the six methods have
comparable performance.

\begin{table}[tbph]
	\caption{The mean of $\protect\widehat{K}$ by the six methods and the
		proportion (prop) of correctly estimating $K_{0}$ among $200$ simulated
		datasets for $K_{0}=2$ and $n=500$.}
	\label{TAB:S1S2K2}
	\par
	\begin{center}
		\begin{adjustbox}{width=0.85\textwidth}
			\begin{tabular}{|l|c|cccccccc|}
				\cline{1-10}
				&  & \multicolumn{7}{c}{S1} & S2 \\
				& $\rho $ & $0.5$ & $1$ & $2$ & $3$ & $4$ & $5$ & $6$ &  \\ \cline{1-10}
				&  & \multicolumn{8}{c|}{SBM} \\ \cline{1-10}
				PLR1 & mean & 2.865 & 2.380 & 2.235 & 2.095 & 2.045 & 2.020 & 2.000 & 2.000
				\\
				& prop & 0.765 & 0.880 & 0.960 & 0.980 & 0.990 & 0.995 & 1.000 & 1.000 \\
				PLR2 & mean & 2.290 & 2.285 & 2.025 & 2.000 & 2.000 & 2.000 & 2.000 & 2.000
				\\
				& prop & 0.875 & 0.900 & 0.995 & 1.000 & 1.000 & 1.000 & 1.000 & 1.000 \\
				LRBIC & mean & 1.000 & 1.000 & 2.000 & 2.000 & 2.000 & 2.000 & 2.000 & 2.000
				\\
				& prop & 0.000 & 0.000 & 1.000 & 1.000 & 1.000 & 1.000 & 1.000 & 1.000 \\
				NCV & mean & 1.055 & 1.105 & 2.205 & 2.005 & 2.010 & 2.020 & 2.000 & 2.005
				\\
				& prop & 0.045 & 0.095 & 0.815 & 0.995 & 0.990 & 0.995 & 1.000 & 0.995 \\
				ECV & mean & 1.000 & 1.000 & 2.005 & 2.000 & 2.000 & 2.000 & 2.000 & 2.000
				\\
				& prop & 0.000 & 0.000 & 0.995 & 1.000 & 1.000 & 1.000 & 1.000 & 1.000 \\
				BHMC & mean & 1.065 & 1.865 & 2.000 & 2.000 & 2.000 & 2.000 & 2.000 & 2.000
				\\
				& prop & 0.065 & 0.845 & 1.000 & 1.000 & 1.000 & 1.000 & 1.000 & 1.000 \\
				\cline{1-10}
				&  & \multicolumn{8}{c|}{DCSBM} \\ \cline{1-10}
				PLR1 & mean & 3.015 & 2.425 & 2.120 & 2.095 & 2.090 & 2.035 & 2.025 & 2.000
				\\
				& prop & 0.710 & 0.905 & 0.980 & 0.980 & 0.980 & 0.990 & 0.995 & 1.000 \\
				PLR2 & mean & 2.275 & 2.205 & 2.000 & 2.000 & 2.000 & 2.000 & 2.000 & 2.000
				\\
				& prop & 0.890 & 0.950 & 1.000 & 1.000 & 1.000 & 1.000 & 1.000 & 1.000 \\
				LRBIC & mean & 1.000 & 1.000 & 2.000 & 2.000 & 2.000 & 2.000 & 2.000 & 2.000
				\\
				& prop & 0.000 & 0.000 & 1.000 & 1.000 & 1.000 & 1.000 & 1.000 & 1.000 \\
				NCV & mean & 1.150 & 1.170 & 2.040 & 1.970 & 1.995 & 2.000 & 2.000 & 2.005
				\\
				& prop & 0.090 & 0.130 & 0.790 & 0.960 & 0.975 & 1.000 & 1.000 & 0.995 \\
				ECV & mean & 1.000 & 1.010 & 2.000 & 2.005 & 2.000 & 2.000 & 2.000 & 2.000
				\\
				& prop & 0.000 & 0.010 & 0.990 & 0.995 & 1.000 & 1.000 & 1.000 & 1.000 \\
				BHMC & mean & 1.080 & 1.880 & 2.000 & 2.000 & 2.000 & 2.000 & 2.000 & 2.000
				\\
				& prop & 0.080 & 0.880 & 1.000 & 1.000 & 1.000 & 1.000 & 1.000 & 1.000 \\
				\cline{1-10}
			\end{tabular}%
		\end{adjustbox}
	\end{center}
\end{table}
\begin{table}[tbph]
	\caption{The mean of $\protect\widehat{K}$ by the six methods and the
		proportion (prop) of correctly estimating $K_{0}$ among $200$ simulated
		datasets for $K_{0}=3$ and $n=500$.}
	\label{TAB:S1S2K3}
	\par
	\begin{center}
		\begin{adjustbox}{width=0.85\textwidth}
			\begin{tabular}{|l|c|cccccccc|}
				\cline{1-10}
				&  & \multicolumn{7}{c}{S1} & S2 \\
				& $\rho $ & $0.5$ & $1$ & $2$ & $3$ & $4$ & $5$ & $6$ &  \\ \cline{1-10}
				&  & \multicolumn{8}{c|}{SBM} \\ \cline{1-10}
				PLR1 & mean & 3.035 & 2.715 & 2.975 & 3.115 & 3.060 & 3.015 & 3.000 & 3.110
				\\
				& prop & 0.080 & 0.085 & 0.535 & 0.975 & 0.990 & 0.995 & 1.000 & 0.980 \\
				PLR2 & mean & 2.125 & 2.595 & 2.975 & 3.060 & 3.015 & 3.000 & 3.000 & 3.000
				\\
				& prop & 0.045 & 0.075 & 0.535 & 0.990 & 0.995 & 1.000 & 1.000 & 1.000 \\
				LRBIC & mean & 1.000 & 1.000 & 1.005 & 2.960 & 3.000 & 3.000 & 3.000 & 3.000
				\\
				& prop & 0.000 & 0.000 & 0.000 & 0.960 & 1.000 & 1.000 & 1.000 & 1.000 \\
				NCV & mean & 1.045 & 1.050 & 1.495 & 2.830 & 3.015 & 3.015 & 3.000 & 3.030
				\\
				& prop & 0.000 & 0.000 & 0.070 & 0.710 & 0.985 & 0.995 & 1.000 & 0.970 \\
				ECV & mean & 1.000 & 1.000 & 1.400 & 2.905 & 3.005 & 3.000 & 3.000 & 3.005
				\\
				& prop & 0.000 & 0.000 & 0.045 & 0.905 & 0.995 & 1.000 & 1.000 & 0.995 \\
				BHMC & mean & 1.055 & 1.160 & 2.335 & 3.000 & 3.000 & 3.000 & 3.000 & 3.000
				\\
				& prop & 0.000 & 0.000 & 0.335 & 1.000 & 1.000 & 1.000 & 1.000 & 1.000 \\
				\cline{1-10}
				&  & \multicolumn{8}{c|}{DCSBM} \\ \cline{1-10}
				PLR1 & mean & 2.925 & 2.930 & 3.180 & 3.070 & 3.025 & 3.030 & 3.025 & 3.035
				\\
				& prop & 0.070 & 0.149 & 0.530 & 0.980 & 0.990 & 0.995 & 0.995 & 0.995 \\
				PLR2 & mean & 2.125 & 2.830 & 3.150 & 3.070 & 3.000 & 3.000 & 3.000 & 3.000
				\\
				& prop & 0.075 & 0.100 & 0.535 & 0.980 & 1.000 & 1.000 & 1.000 & 1.000 \\
				LRBIC & mean & 1.000 & 1.000 & 1.025 & 2.955 & 3.000 & 3.000 & 3.000 & 3.000
				\\
				& prop & 0.000 & 0.000 & 0.000 & 0.955 & 1.000 & 1.000 & 1.000 & 1.000 \\
				NCV & mean & 1.040 & 1.065 & 1.595 & 2.955 & 3.000 & 3.005 & 3.000 & 3.010
				\\
				& prop & 0.005 & 0.000 & 0.085 & 0.820 & 0.990 & 0.995 & 1.000 & 0.990 \\
				ECV & mean & 1.000 & 1.000 & 1.350 & 2.940 & 3.005 & 3.000 & 3.000 & 3.000
				\\
				& prop & 0.000 & 0.000 & 0.030 & 0.930 & 0.995 & 1.000 & 1.000 & 1.000 \\
				BHMC & mean & 1.055 & 1.145 & 2.415 & 2.995 & 3.000 & 3.000 & 3.000 & 3.000
				\\
				& prop & 0.000 & 0.000 & 0.415 & 0.995 & 1.000 & 1.000 & 1.000 & 1.000 \\
				\cline{1-10}
			\end{tabular}%
		\end{adjustbox}
	\end{center}
\end{table}
\begin{table}[tbph]
	\caption{The mean of $\protect\widehat{K}$ by the six methods and the
		proportion (prop) of correctly estimating $K_{0}$ among $200$ simulated
		datasets for $K_{0}=4$ and $n=500$.}
	\label{TAB:S1S2K4}
	\par
	\begin{center}
		\begin{adjustbox}{width=0.85\textwidth}
			\begin{tabular}{|l|c|cccccccc|}
				\cline{1-10}
				&  & \multicolumn{7}{c}{S1} & S2 \\
				& $\rho $ & $0.5$ & $1$ & $2$ & $3$ & $4$ & $5$ & $6$ &  \\ \cline{1-10}
				&  & \multicolumn{8}{c|}{SBM} \\ \cline{1-10}
				PLR1 & mean & 2.665 & 2.850 & 3.200 & 3.465 & 4.085 & 4.060 & 4.000 & 4.000
				\\
				& prop & 0.015 & 0.025 & 0.035 & 0.355 & 0.925 & 0.980 & 1.000 & 1.000 \\
				PLR2 & mean & 2.300 & 2.850 & 2.665 & 3.465 & 4.085 & 4.060 & 4.000 & 3.995
				\\
				& prop & 0.015 & 0.025 & 0.025 & 0.355 & 0.925 & 0.980 & 1.000 & 0.995 \\
				LRBIC & mean & 1.000 & 1.000 & 1.000 & 1.005 & 3.840 & 4.000 & 4.000 & 4.000
				\\
				& prop & 0.000 & 0.000 & 0.000 & 0.000 & 0.920 & 1.000 & 1.000 & 1.000 \\
				NCV & mean & 1.015 & 1.020 & 1.004 & 1.500 & 4.030 & 4.005 & 4.000 & 4.060
				\\
				& prop & 0.000 & 0.000 & 0.000 & 0.070 & 0.740 & 0.965 & 1.000 & 0.940 \\
				ECV & mean & 1.000 & 1.000 & 1.000 & 1.370 & 3.905 & 4.000 & 4.000 & 4.000
				\\
				& prop & 0.000 & 0.000 & 0.000 & 0.035 & 0.845 & 1.000 & 1.000 & 1.000 \\
				BHMC & mean & 1.035 & 1.020 & 1.200 & 2.330 & 3.610 & 3.985 & 4.000 & 4.000
				\\
				& prop & 0.000 & 0.000 & 0.000 & 0.015 & 0.630 & 0.985 & 1.000 & 1.000 \\
				\cline{1-10}
				&  & \multicolumn{8}{c|}{DCSBM} \\ \cline{1-10}
				PLR1 & mean & 2.750 & 2.780 & 2.765 & 3.675 & 4.175 & 4.045 & 4.010 & 4.005
				\\
				& prop & 0.030 & 0.040 & 0.040 & 0.380 & 0.915 & 0.985 & 0.995 & 0.995 \\
				PLR2 & mean & 2.105 & 2.655 & 2.745 & 3.675 & 4.150 & 4.015 & 4.000 & 4.005
				\\
				& prop & 0.000 & 0.015 & 0.040 & 0.380 & 0.920 & 0.995 & 1.000 & 0.995 \\
				LRBIC & mean & 1.000 & 1.000 & 1.000 & 1.005 & 3.845 & 4.000 & 4.000 & 4.000
				\\
				& prop & 0.000 & 0.000 & 0.000 & 0.000 & 0.920 & 1.000 & 1.000 & 1.000 \\
				NCV & mean & 1.050 & 1.003 & 1.045 & 1.805 & 4.005 & 4.015 & 4.020 & 4.060
				\\
				& prop & 0.000 & 0.000 & 0.000 & 0.100 & 0.700 & 0.980 & 0.980 & 0.940 \\
				ECV & mean & 1.000 & 1.000 & 1.000 & 1.435 & 3.895 & 4.000 & 4.005 & 4.005
				\\
				& prop & 0.000 & 0.000 & 0.000 & 0.040 & 0.840 & 1.000 & 0.995 & 0.995 \\
				BHMC & mean & 1.075 & 1.015 & 1.285 & 2.360 & 3.575 & 3.985 & 4.000 & 4.000
				\\
				& prop & 0.000 & 0.000 & 0.000 & 0.050 & 0.600 & 0.985 & 1.000 & 1.000 \\
				\cline{1-10}
			\end{tabular}
		\end{adjustbox}
	\end{center}
\end{table}

As suggested by one referee, we can replace the pseudo likelihood function
by the k-means loss function to compare the estimated $K$ communities with
the estimated $K+1$ communities obtained from our spectral clustering with
binary segmentation method. To this end, we let $Q_{n}(\hat{Z}_{K+1}^{b},%
\hat{Z}_{K})$ be the difference of the k-means loss functions for the
estimated $K$ and $K+1$ communities obtained from the first $K+1$ normalized
eigenvectors of the regularized graph Laplacian. Then the estimated number
of communities minimizes $\frac{Q_{n}(\hat{Z}_{K+1}^{b},\hat{Z}_{K})/(K+1)}{%
	Q_{n}(\hat{Z} _{K}^{b},\hat{Z}_{K-1})/K}$, and we call this estimator
\textquotedblleft KML\textquotedblright. Note that $Q_{n}(\hat{Z}_{K+1}^{b},%
\hat{Z}_{K})$ involves the eigenvectors with dimension $n \times (K+1)$.
Thus we need to normalize it via dividing it by $K+1$. In addition, we apply
the gap statistic proposed in \cite{T01} for estimating the number of
communities by using the R package \textquotedblleft
cluster\textquotedblright . The gap statistic was proposed for clustering $p$%
-dimensional independent vectors into $K$ groups for $K=1,\cdots ,K_{\max }$%
, where $p$ is fixed and do not change with $K$. We let $p=K_{\max }$ in our
setting, so that we apply this method to the first $K_{\max }$ normalized
eigenvectors of the regularized graph Laplacian. Moreover, \cite{YSC18}
proposed a semi-definite programming method (SPUR) for determining the
number of communities in SBMs. We compare our proposed estimator PLR1 with
these three estimators, KML, GAP and SPUR. Since the proposed estimator PLR2
performs slightly better than PLR1, we only compare PLR1 with other three
estimators.

Table \ref{TAB:PLR1} reports the mean of the estimated number of communities
by the four methods, PLR1, KML, GAP and SPUR, and the proportion (prop) of
correctly estimating $K_{0}$ among $200$ simulated datasets when data are
generated from the SBMs and designs S1 and S2 given in Section \ref{sec:DGP}
with $n=500$. In Table \ref{TAB:PLR1DC}, we report those statistics for the
three methods, PLR1, KML, and GAP, when the data are generated from the
DCSBMs given in Section \ref{sec:DGP}, as the SPUR method was proposed only
for the SBMs. Tables \ref{TAB:PLR1} and \ref{TAB:PLR1DC} show that our
proposed PLR1 has the best performance for all cases. Specifically, the gap
statistic method applies the k-means to $p$-dimensional vectors, where $p$
is fixed and is not allowed to change with $K$. Hence, it is not directly
applicable to network data clustering. As a result, it performs worse than
other methods. The KML method performs better than the GAP and SPUR for most
cases of design S1, but it is inferior to the proposed PLR1 method,
especially for large $K$'s. This is due to the fact that for determining the
number of communities, the KML method only uses the information from the
eigenvectors, whereas the proposed PLR1 method uses the likelihood which
involves all information from the parameter estimates. Moreover, the
proposed PLR methods are built on the spectral clustering with binary
segmentation algorithm for estimation, and thus they are computationally
fast. They have the advantage over the semi-definite programming method,
SPUR, in terms of computational speed. Computational efficiency needs to be
taken into account for model selection in large network data.

\begin{table}[tbph]
	\caption{The mean of $\protect\widehat{K}$ by the four methods, PLR1, KML,
		GAP and SPUR, and the proportion (prop) of correctly estimating $K_{0}$
		among $200$ simulated datasets when data are generated from SBMs with $n=500$%
		.}
	\label{TAB:PLR1}
	\par
	\begin{center}
		\begin{adjustbox}{width=\textwidth}
			\begin{tabular}{|l|c|c|cccc|cccc|cccc|}
				\cline{1-15}
				&  &  & \multicolumn{4}{c}{$K_{0}=2$} & \multicolumn{4}{|c|}{$K_{0}=3$}& \multicolumn{4}{|c|}{$K_{0}=4$} \\
				\cline{1-15}
				&  $\rho $ &  & PLR1 & KML & GAP & SPUR & PLR1 & KML
				& GAP & SPUR& PLR1 & KML & GAP & SPUR \\ \cline{1-15}
				S1 & $3$ & mean & $2.095$ & $2.110$ & $7.715$ & $1.815$ & $3.115$ & $ 2.955$
				& $8.615$ & $2.540$& $3.465$ & $3.155$ & $9.290$ & $ 3.005$ \\
				&  & prop & $0.980$ & $0.975$ & $0.115$ & $0.815$ & $0.975$ & $0.895$ & $0.060$ & $0.540$ & $0.355$ & $0.140$ & $0.000$ & $0.115$\\
				& $4$ & mean & $2.045$ & $2.085$ & $6.265$ & $1.860$ & $3.060$ & $2.965$ & $%
				6.830$ & $2.655$ & $4.085$ & $3.655$ & $8.115$ & $3.515$\\
				&  & prop & $0.990$ & $0.980$ & $0.265$ & $0.860$ & $0.990$ & $0.975$ & $%
				0.350$ & $0.655$ & $0.925$ & $0.725$ & $0.115$ &  $0.545$\\
				& $5$ & mean & $2.020$ & $2.040$ & $5.080$ & $1.880$ & $3.015$ & $3.020$ & $%
				5.265$ & $2.755$ & $4.060$ & $3.840$ & $6.320$ & $3.735$\\
				&  & prop & $0.995$ & $0.990$ & $0.400$ & $0.880$ & $0.995$ & $0.990$ & $%
				0.610$ & $0.785$ & $0.980$ & $0.900$ & $0.535$ & $0.785$ \\
				S2 &  & mean & $2.000$ & $2.320$ & $9.470$ & $2.000$ & $3.110$ & $3.200$ & $%
				9.265$ & $2.935$ & $4.000$ & $4.000$ & $9.335$ & $3.905$ \\
				&  & prop & $1.000$ & $0.915$ & $0.000$ & $1.000$ & $0.980$ & $0.970$ & $%
				0.000$ & $0.945$ & $1.000$ & $1.000$ & $0.010$ & $0.925$ \\ \cline{1-15}
			\end{tabular}%
		\end{adjustbox}
	\end{center}
\end{table}

\begin{table}[tbph]
	\caption{The mean of $\protect\widehat{K}$ by the three methods, PLR1, KML
		and GAP, and the proportion (prop) of correctly estimating $K_{0}$ among $200
		$ simulated datasets when data are generated from DCSBMs with $n=500$.}
	\label{TAB:PLR1DC}
	\par
	\begin{center}
		\begin{adjustbox}{width=\textwidth}
			\begin{tabular}{|l|c|c|ccc|ccc|ccc|}
				\cline{1-12}
				&  &  & \multicolumn{3}{c}{$K_{0}=2$} & \multicolumn{3}{|c|}{$K_{0}=3$}& \multicolumn{3}{|c|}{$K_{0}=4$} \\
				\cline{1-12}
				&  $\rho $ &  & PLR1 & KML & GAP & PLR1 & KML
				& GAP &  PLR1 & KML & GAP  \\ \cline{1-12}
				S1 & $3$ & mean & $2.095$ & $2.110$ & $8.210$ & $3.070$ & $ 2.895$
				& $8.855$ & $3.675$ & $3.115$ & $9.300$ \\
				&  & prop & $0.980$ & $0.975$ & $0.055$ & $0.980$ & $ 0.875$
				& $0.045$ &$0.380$ & $0.135$ & $0.000$  \\
				& $4$ & mean & $2.090$ & $2.095$ & $6.730$ & $3.025$ & $2.955$ & $7.015 $
				& $4.175$ & $3.525$& $8.585$ \\
				&  & prop & $0.980$ & $0.980$ & $0.315$ & $0.990$ & $0.970$ & $0.175 $
				& $0.915$ & $0.725$& $0.095$  \\
				& $5$ & mean & $2.035$ & $2.040$ & $5.455$ & $3.030$ & $3.050$ & $6.410 $
				& $4.045$ & $3.840$& $6.990$ \\
				&  & prop & $0.990$ & $0.990$ & $0.490$ & $0.995$ & $0.985$ & $0.420 $
				& $0.985$ & $0.900$& $0.410$  \\
				S2 &  & mean & $2.000$ & $2.585$ & $9.375$ & $3.035$ & $3.055$ & $ 9.440$
				& $4.005$ & $4.010$& $9.455$  \\
				&  & prop & $1.000$ & $0.850$ & $0.000$ & $0.995$ & $0.990$ & $ 0.000$
				& $0.995$ & $0.990$& $0.010$  \\ \cline{1-12}
			\end{tabular}%
		\end{adjustbox}
	\end{center}
\end{table}

Lastly, for the DCSBMs, we generate the degree parameters $\theta _{i}$ from
the Pareto distribution with the scale parameter 1 and the shape parameter
5, and further normalize them to satisfy the condition (\ref%
{eq:thetanormalization}). Tables \ref{TAB:Khat1DCSBMPareto} and \ref%
{TAB:Khat2DCSBMPareto} report the mean of $\widehat{K}_{1}$ and $\widehat{K}%
_{2}$ with $c_{h}=1.0$, respectively, and the proportion (prop) of correctly
estimating $K_{0}$ among $200$ simulated datasets. We see that both PLR1 and
PLR2 perform well, and the results in Tables \ref{TAB:Khat1DCSBMPareto} and %
\ref{TAB:Khat2DCSBMPareto} are comparable to those for $\widehat{K}_{1}$ and
$\widehat{K}_{2}$ with $c_{h}=1.0$ shown in Tables \ref{TAB:Khat2DCSBM} and %
\ref{TAB:Khat1DCSBM} when $\theta _{i}$ are generated from the uniform
distribution.

\begin{table}[tbph]
	\caption{The mean of $\protect\widehat{K}_{1}$ and the proportion (prop) of
		correctly estimating $K_{0}$ among $200$ simulated datasets when data are
		simulated from DCSBMs with the degree parameters $\protect\theta _{i}$
		generated from the Pareto distribution.}
	\label{TAB:Khat1DCSBMPareto}
	\par
	\begin{center}
		\begin{adjustbox}{width=\textwidth}
			\begin{tabular}{|l|c|c|cccc|cccc|}
				\cline{1-11}
				&  &  & \multicolumn{4}{c}{$n=500$} & \multicolumn{4}{|c|}{$n=1000$} \\
				\cline{1-11}
				&  $\rho $ &  & $K_{0}=1$ & $K_{0}=2$ & $K_{0}=3$ & $K_{4}=4$ & $K_{0}=1$ & $K_{0}=2$
				& $K_{0}=3$ & $K_{4}=4$ \\ \cline{1-11}
				S1 & $3$ & mean & $1.085$ & $2.095$ & $3.135$ & $3.510$ & $1.000$ & $2.090$ & $3.035$ & $4.045$ \\
				&  & prop & $0.965$ & $0.985$ & $0.950$ & $0.360$ & $1.000$ & $0.985$ & $0.990$ & $0.990$ \\
				& $4$ & mean & $1.010$ & $2.080$ & $3.040$ & $4.140$ & $1.000$ & $2.050$ & $3.000$ & $4.040$ \\
				&  & prop & $0.995$ & $0.985$ & $0.990$ & $0.910$ & $1.000$ & $0.990$ & $1.000$ & $0.990$ \\
				& $5$ & mean & $1.000$ & $2.000$ & $3.000$ & $4.045$ & $1.000$ & $2.000$ & $3.000$ & $4.035$ \\
				&  & prop & $1.000$ & $1.000$ & $1.000$ & $0.985$ & $1.000$ & $1.000$ & $1.000$ & $0.990$ \\
				S2 &  & mean & $1.000$ & $2.000$ & $3.000$ & $4.000$ & $1.000$ & $2.000$ & $3.000$ & $4.000$ \\
				&  & prop & $1.000$ & $1.000$ & $1.000$ & $1.000$ & $1.000$ & $1.000$ & $1.000$ & $1.000$ \\ \cline{1-11}
			\end{tabular}%
		\end{adjustbox}
	\end{center}
\end{table}

\begin{table}[tbph]
	\caption{The mean of $\protect\widehat{K}_{2}$ and the proportion (prop) of
		correctly estimating $K_{0}$ among $200$ simulated datasets when data are
		simulated from DCSBMs with the degree parameters $\protect\theta _{i}$
		generated from the Pareto distribution.}
	\label{TAB:Khat2DCSBMPareto}
	\par
	\begin{center}
		\begin{adjustbox}{width=\textwidth}
			\begin{tabular}{|l|c|c|cccc|cccc|}
				\cline{1-11}
				&  &  & \multicolumn{4}{c}{$n=500$} & \multicolumn{4}{|c|}{$n=1000$} \\
				\cline{1-11}
				&  $\rho $ &  & $K_{0}=1$ & $K_{0}=2$ & $K_{0}=3$ & $K_{4}=4$ & $K_{0}=1$ & $K_{0}=2$
				& $K_{0}=3$ & $K_{4}=4$ \\ \cline{1-11}
				S1 & $3$ & mean & $1.085$ & $2.000$ & $3.080$ & $3.510$ & $1.000$ & $2.000$ & $3.015$ & $4.045$ \\
				&  & prop & $0.965$ & $1.000$ & $0.965$ & $0.360$ & $1.000$ & $1.000$ & $0.995$ & $0.990$ \\
				& $4$ & mean & $1.010$ & $2.000$ & $3.000$ & $4.140$ & $1.000$ & $2.000$ & $3.000$ & $4.040$ \\
				&  & prop & $0.995$ & $1.000$ & $1.000$ & $0.910$ & $1.000$ & $1.000$ & $1.000$ & $0.990$ \\
				& $5$ & mean & $1.000$ & $2.000$ & $3.000$ & $4.020$ & $1.000$ & $2.000$ & $3.000$ & $4.000$ \\
				&  & prop & $1.000$ & $1.000$ & $1.000$ & $0.990$ & $1.000$ & $1.000$ & $1.000$ & $1.000$ \\
				S2 &  & mean & $1.000$ & $2.000$ & $3.000$ & $4.000$ & $1.000$ & $2.000$ & $3.000$ & $4.000$ \\
				&  & prop & $1.000$ & $1.000$ & $1.000$ & $1.000$ & $1.000$ & $1.000$ & $1.000$ & $1.000$ \\ \cline{1-11}
			\end{tabular}%
		\end{adjustbox}
	\end{center}
\end{table}

\section{Additional real data applications}

\label{sec:add_application}

\subsection{Political books network}

We investigate the community structure of a network of US political books (available at www.orgnet.com) by different methods.
In this network, there
are 105 nodes representing books about US politics published around the 2004
presidential election and sold by the online bookseller Amazon.com, and
there are 441 edges representing frequent co-purchasing of books by the same
buyers. The left graph of Figure \ref{Fig:book} shows the degree
distribution for the political books network with the average degree being
8.4. We see that the degree has a right skewed distribution with most values
ranging from 2 to 9. Let $K_{\max }=10$. We identify $\widehat{K}_{1}=%
\widehat{K}_{2}=3$ communities by both PLR1 and PLR2. This result is
consistent with the ground-truth community structure that these books are
actually divided into three categories \textquotedblleft
liberal\textquotedblright , \textquotedblleft neutral\textquotedblright\ and
\textquotedblleft conservative\textquotedblright\ according to their
political views \citep{Newman06}. For further demonstration, we plot the
political books network with three communities in the right panel of Figure %
\ref{Fig:book}. Groups 1, 2 and 3 represent the estimated communities of
liberal, conservative and neutral books. We also obtain the estimated edge
probabilities within communities which are $\widehat{B}%
_{kk}=0.219,0.224,0.164$ for $k=1,2,3$, and the edge probabilities between
communities which are $\widehat{B}_{12}=0.001$, $\widehat{B}_{13}=0.019$ and
$\widehat{B}_{23}=0.224$. We see that groups 1 and 2 from two different
political affiliations are very weakly connected. We apply the LRBIC, NCV,
ECV and BHMC methods, and obtain the estimated number of communities as 3,
6, 8 and 4, respectively, by these four methods.
\begin{figure}[tbp]
	\caption{Left panel shows the degree distribution; right panel depicts the
		political books network with three communities. }
	\label{Fig:book}
	\centering
	\vspace{0.5cm} $
	\begin{array}{cc}
	\textbf{Degree distribution of political books network} \vspace{-0.4cm} &
	\textbf{Political books network} \vspace{-0.4cm} \\
	\includegraphics[width = 0.45\linewidth]{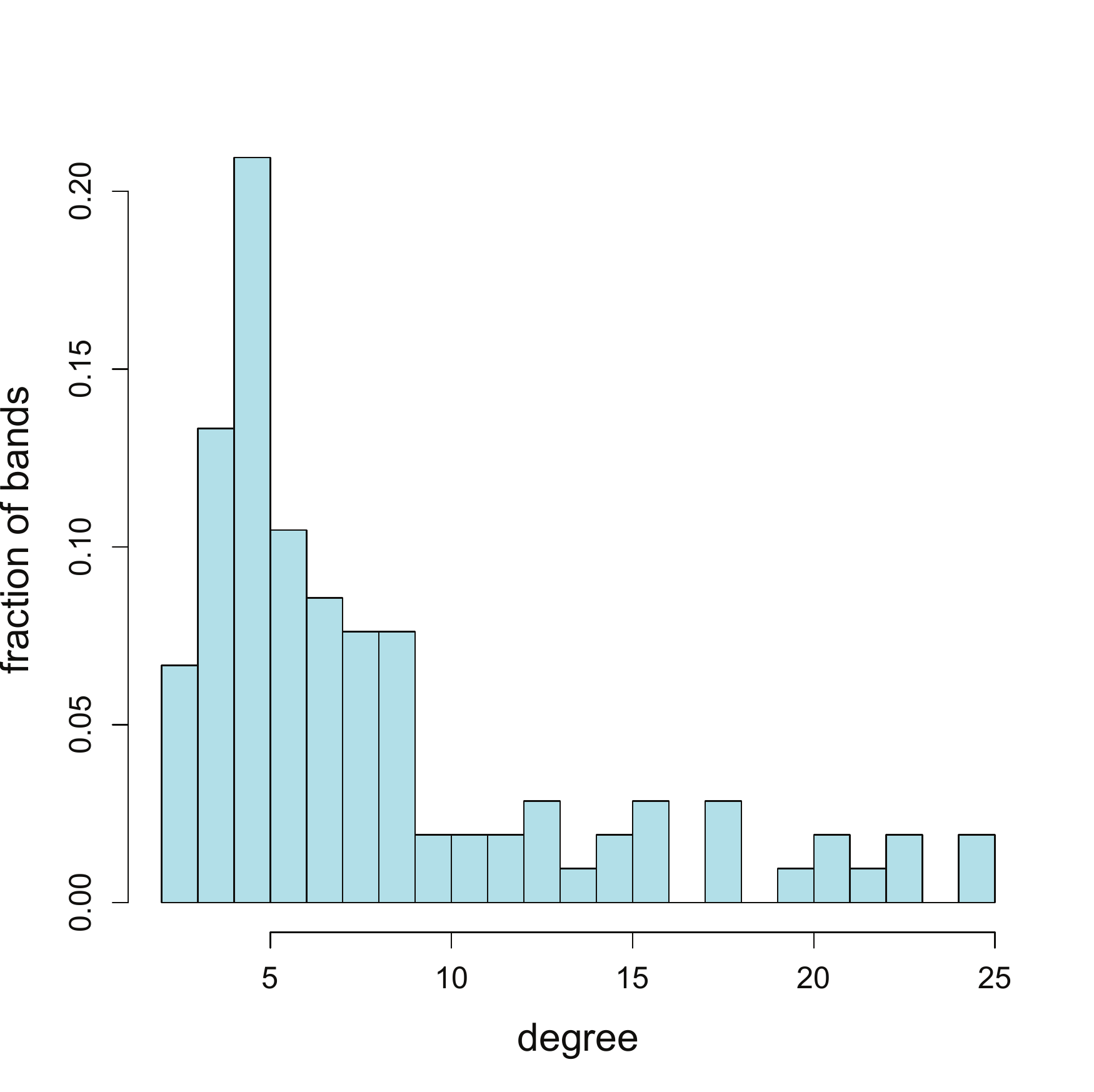} & %
	\includegraphics[width = 0.45\linewidth]{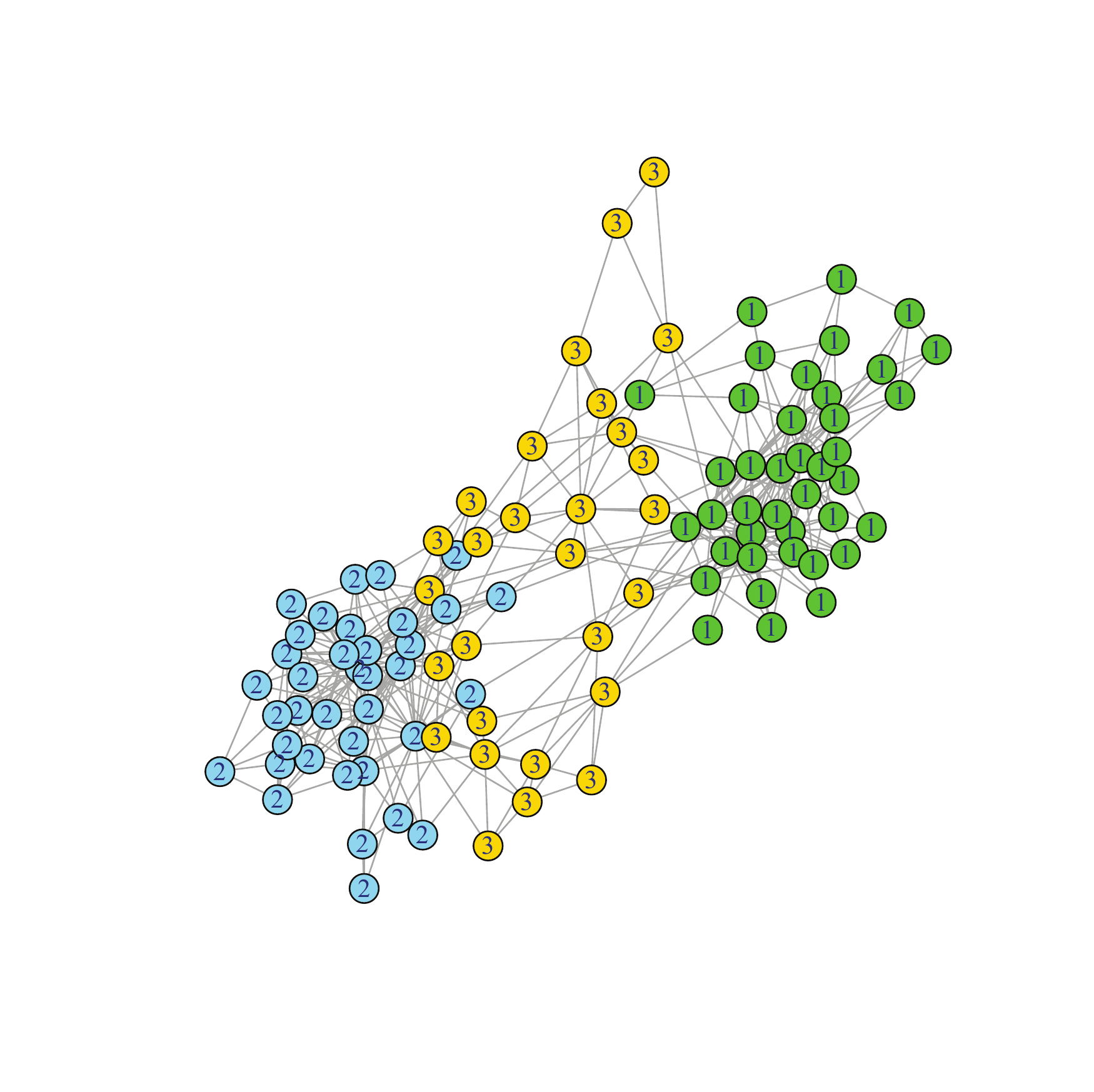}%
	\end{array}
	$%
\end{figure}

\subsection{Facebook friendship network}

We apply our methods to a large social network which contains friendship
data of Facebook users (available at www.snap.stanford.edu). A node
represents a user and an edge represents a friendship between two users. The
data have 4039 nodes and 88218 edges. We use the nodes with the degree
between 10 and 300. As a result, there are 2901 nodes and 80259 edges in our
analysis. The left graph of Figure \ref{Fig:facebook} shows the degree
distribution for the Facebook friendship network with the average degree
being 55.33. The degree distribution is again right skewed. Let $K_{\max
}=20 $. By using the proposed PLR1 and PLR2 methods, we identify $\widehat{K}%
_{1}=\widehat{K}_{2}=11$ communities. The right panel of Figure \ref%
{Fig:facebook} shows the estimated community structure of the Facebook
friendship network with eleven identified communities. We can observe
sub-communities of friends who are tightly connected through mutual
friendships. Lastly, the LRBIC, NCV, ECV and BHMC methods found $19$, $19$, $%
20$ and $14$ communities, respectively.

\begin{figure}[tbp]
	\caption{Left panel shows the degree distribution; right panel depicts the
		facebook friendship network with eleven communities. }
	\label{Fig:facebook}
	\centering
	\vspace{0.5cm} $
	\begin{array}{cc}
	\textbf{Degree distribution of facebook network} \vspace{-0.4cm} & \textbf{\
		Facebook network} \vspace{-0.4cm} \\
	\includegraphics[width = 0.45\linewidth]{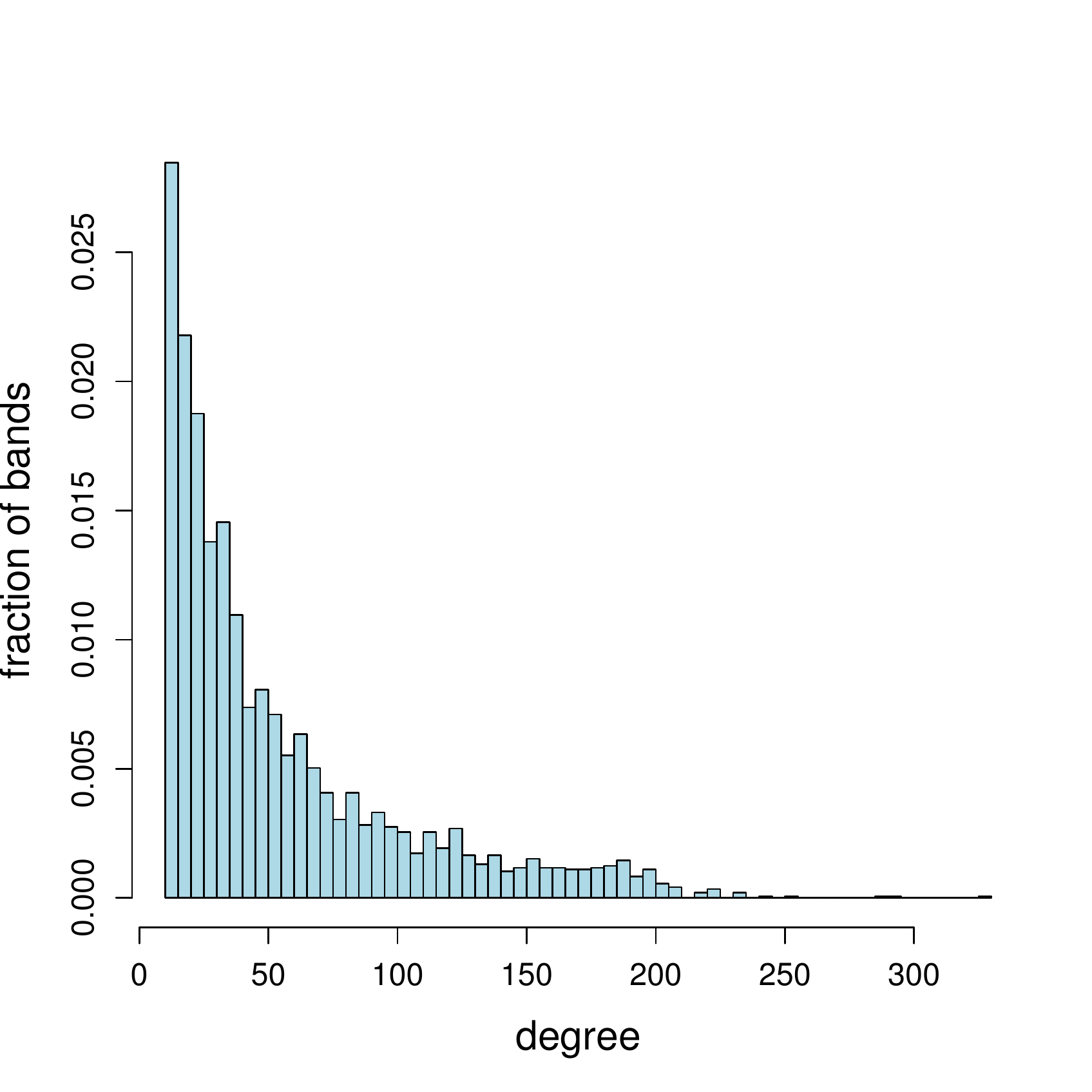} & %
	\includegraphics[width = 0.45\linewidth]{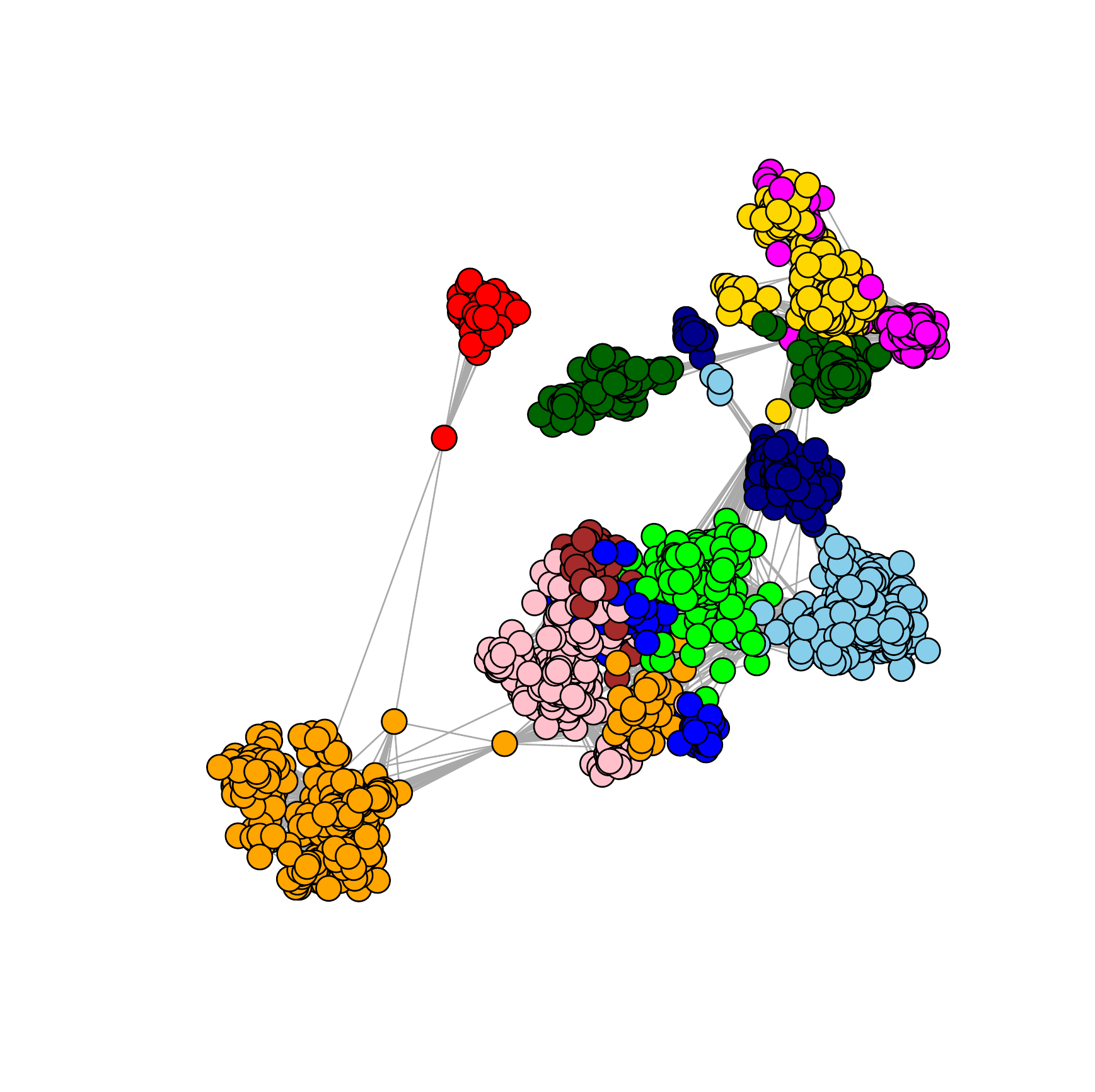}%
	\end{array}
	$%
\end{figure}

\section{Proofs of results in Section \protect\ref{sec:theory}}

\subsection{Proof of Theorem \protect\ref{thm:id3}}

The first result is proved in \citet[Theorem 3.3]{SWZ17}. For part (2), by
Lemma \ref{lem:id3}(1), if $i\in \mathcal{C}_{k,K_{0}}$, then
\begin{equation*}
u_{i}^{T}(K)=(\theta _{i}^{\tau })^{1/2}(n_{k,K_{0}}^{\tau
})^{-1/2}S_{n}^{\tau }(K).
\end{equation*}%
Because $S_{n}^{\tau }(K)$ is a $K_{0}\times K$ matrix, it is easy to see
that $L_{K}\leq K_{0}$. By the proof of \citet[Theorem 3.3]{SWZ17}, $%
S_{n}^{\tau }$ is the $K_{0}\times K_{0}$ eigenvector matrix of $(\Pi
_{n}^{\tau })^{1/2}H_{0,K_{0}}(\Pi _{n}^{\tau })^{1/2}$ with the
corresponding eigenvalues ordered from the biggest to the smallest in
absolute values. By Assumptions \ref{ass:id3} and \ref{ass:nk2}, we have
\begin{equation*}
(\Pi _{n}^{\tau })^{1/2}H_{0,K_{0}}(\Pi _{n}^{\tau })^{1/2}\rightarrow \Pi
_{\infty }^{\prime 1/2}H_{0,K_{0}}^{\ast }\Pi _{\infty }^{\prime
	1/2}:=S_{\infty }\Sigma _{\infty }S_{\infty }.
\end{equation*}%
By Davis-Kahan Theorem in \cite{YWS15} and Assumption \ref{ass:nk2}(2),
there exists a $K\times K$ orthogonal matrix $O_{s}$ such that $S_{n}^{\tau
}(K)O_{s}\rightarrow S_{\infty }[K]$ where $S_{\infty }$ is the eigenvector
matrix of $\Pi _{\infty }^{\prime 1/2}H_{0,K_{0}}^{\ast }\Pi _{\infty
}^{\prime 1/2}$ and is of full rank. Therefore, if $i\in \mathcal{C}%
_{k,K_{0}}$ and $j\in \mathcal{C}_{l,K_{0}}$,
\begin{align}
\left\Vert \frac{u_{i}^{T}(K)}{||u_{i}^{T}(K)||}-\frac{u_{j}^{T}(K)}{%
	||u_{j}^{T}(K)||}\right\Vert =& \left\Vert \left( \frac{[S_{n}^{\tau
	}]_{k}(K)}{||[S_{n}^{\tau }]_{k}(K)||}-\frac{[S_{n}^{\tau }]_{l}(K)}{%
	||[S_{n}^{\tau }]_{l}(K)||}\right) O_{s}\right\Vert   \notag
\label{eq:sinfty} \\
\rightarrow & \left\Vert \frac{[S_{\infty }]_{k}(K)}{||[S_{\infty }]_{k}(K)||%
}-\frac{[S_{\infty }]_{l}(K)}{||[S_{\infty }]_{l}(K)||}\right\Vert .
\end{align}%
Because $S_{\infty }$ is of full rank, the first $K$ columns of $S_{\infty }$
should have rank $K$. This implies the $K$-dimensional row vectors $\{\frac{%
	[S_{\infty }]_{k}(K)}{||[S_{\infty }]_{k}(K)||}\}_{k=1}^{K_{0}}$ take at
least $K$ distinct values, which are denoted as $\bar{\nu}_{1,K},\cdots ,%
\bar{\nu}_{L_{K},K}$. Therefore, $L_{K}\geq K$ . Next, we call nodes $i$ and
$j$ are equivalent if both $\frac{u_{i}^{T}(K)}{||u_{i}^{T}(K)||}$ and $%
\frac{u_{j}^{T}(K)}{||u_{j}^{T}(K)||}$ converges to one of $(\bar{\nu}_{l,K})
$, $l=1,\cdots ,L_{K}$. Then $G_{l,K}$ can be constructed as the equivalence
class of the above equivalence relation. Let
\begin{equation*}
I=\biggl\{(k,l):\left\Vert \frac{[S_{\infty }]_{k}(K)}{||[S_{\infty
	}]_{k}(K)||}-\frac{[S_{\infty }]_{l}(K)}{||[S_{\infty }]_{l}(K)||}%
\right\Vert \neq 0,k=1,\cdots ,K_{0},l=1,\cdots ,K_{0}\biggr\}.
\end{equation*}%
In view of the fact that the cardinality of $I$ is finite, we have
\begin{equation*}
c^{\ast }=\min_{(k,l)\in I}\left\Vert \frac{[S_{\infty }]_{k}(K)}{%
	||[S_{\infty }]_{k}(K)||}-\frac{[S_{\infty }]_{l}(K)}{||[S_{\infty
	}]_{l}(K)||}=\min_{\ell \neq \ell ^{\prime }}||\bar{\nu}_{\ell ,K}-\bar{\nu}%
_{\ell ^{\prime },K}||\right\Vert >0.
\end{equation*}

Then, by \eqref{eq:sinfty}, if nodes $i \notin G_{l,K}$,
\begin{equation*}
\liminf_{n}\left\Vert\frac{u_{i}^{T}(K)}{||u_{i}^{T}(K)||}-\bar{\nu}%
_{l,K}\right\Vert\geq c^{\ast }>0.
\end{equation*}
This implies that $\{G_{l,K}\}_{l=1}^{L_{K}}$ constructed as the equivalence
class satisfy the two requirements in Theorem \ref{thm:id3}(2) with $
c=c^{\ast }$.

\subsection{Proof of Theorem \protect\ref{thm:oracleDC}}

\textbf{First, we prove Theorem \ref{thm:oracleDC}(1). }Let $\hat{g}_{iK}$
be the membership estimated by the k-means algorithm with $K$ centroids,
i.e.,
\begin{equation*}
\hat{g}_{iK}=\argmin_{1\leq k\leq K}||\hat{\nu}_{iK}-\hat{\alpha}_{k}||\quad
\text{and}\quad \{\hat{\alpha}_{k}\}_{k=1}^{K}=\argmin_{\alpha _{1},\cdots
	,\alpha _{K}}\frac{1}{n}\sum_{i=1}^{n}\min_{1\leq k\leq K}||\hat{\nu}%
_{iK}-\alpha _{k}||^{2}.
\end{equation*}%
Because the $L_{2}$-norm is invariant under rotation,
\begin{equation}
\hat{g}_{iK}=\argmin_{1\leq k\leq K}||\hat{\nu}_{iK}\hat{O}_{Kn}O_{s}-\hat{%
	\alpha}_{k}||\quad \text{and}\quad \{\hat{\alpha}_{k}\}_{k=1}^{K}=\argmin%
_{\alpha _{1},\cdots ,\alpha _{K}}\frac{1}{n}\sum_{i=1}^{n}\min_{1\leq k\leq
	K}||\hat{\nu}_{iK}\hat{O}_{Kn}O_{s}-\alpha _{k}||^{2}.  \label{eq:Zhat}
\end{equation}%
where $\hat{O}_{Kn}$ is a $K\times K$ orthonormal matrix such that $\hat{O}%
_{Kn}=\bar{U}\bar{V}^{T},$ $\bar{U}\bar{\Sigma}\bar{V}^{T}$ is the singular
value decomposition of $\widehat{U}_{n}(K)^{T}U_{n}(K),$ $U_{n}$ is the
population analogue of $\widehat{U}_{n}:$ $\mathcal{L}_{\tau }=U_{n}\Sigma
_{n}U_{n}^{T}$, and $O_{s}$ is another $K\times K$ orthonormal matrix
defined in the proof of Theorem \ref{thm:id3}(2). Here, $\Sigma _{n}=\text{%
	diag}(\sigma _{1n},\ldots ,\sigma _{K_{0}n},0,...,0)$ is a $n\times n$
matrix and we suppress the dependence of $\bar{U},\bar{\Sigma},$ and $\bar{V}
$ on $K.$ We aim to show
\begin{equation}
\sup_{i}1\{\hat{g}_{iK}\neq g_{iK}\}=0\quad a.s.  \label{eq:g}
\end{equation}%
Suppose that
\begin{equation}
\sup_{1\leq i\leq n}||\hat{\nu}_{iK}^{T}\hat{O}_{Kn}O_{s}-\nu
_{iK}^{T}||\leq c_{1}\quad a.s.,  \label{eq:dc}
\end{equation}%
for some sufficiently small $c_{1}>0$, which we will prove later. In
addition, by \eqref{eq:ZK},
\begin{equation*}
\{\alpha _{k}^{\ast }\}_{k=1}^{K}=\argmin_{\alpha _{1},\cdots ,\alpha
	_{K}}\sum_{l=1}^{K_{0}}\pi _{ln}\min_{1\leq k\leq K}||\bar{\nu}_{lK}-\alpha
_{k}||^{2}.
\end{equation*}%
Then for any $k=1,\cdots ,K$, we have
\begin{equation*}
\alpha _{k}^{\ast }=\sum_{l\leq K_{0}:\mathcal{C}_{l,K_{0}}\subset \mathcal{C%
	}_{k,K}}\psi _{n,k,l}\bar{\nu}_{lK},
\end{equation*}%
or in matrix form,
\begin{equation*}
(\alpha _{1}^{\ast },\cdots ,\alpha _{K}^{\ast })=(\bar{\nu}_{1K},\cdots ,%
\bar{\nu}_{L_{K},K})\Psi _{n}^{\prime },
\end{equation*}%
where $\psi _{n,k,l}=\pi _{ln}/(\sum_{l\leq K_{0}:\mathcal{C}%
	_{l,K_{0}}\subset \mathcal{C}_{k,K}}\pi _{ln})$ for $k=1,\cdots ,K$ and $%
l=1,\cdots ,L_{K},$ and $\Psi _{n}=[\psi _{n,k,l}]$. Note that $L_{K}\geq K$%
. By Assumption \ref{ass:nk2}, $\Psi _{n}\rightarrow \Psi _{\infty }$, where
$[\Psi _{\infty }]_{k,l}=\pi _{l\infty }/\sum_{l\leq K_{0}:\mathcal{C}%
	_{l,K_{0}}\subset \mathcal{C}_{k,K}}\pi _{l\infty }>0$. Because $Z_{K}$ is
unique by Assumption \ref{ass:uniqueDC}(1) and $\pi _{l\infty }$ is positive
for $l=1,\cdots ,K_{0}$, we have that each column of $\Psi _{\infty }$ has
one and only one nonzero entry. In addition, there exist at least $L_{K}\geq
K$ distinct vectors in $\{\bar{\nu}_{lK}\}_{l=1}^{K_{0}}$. Therefore, by
relabeling both $\{\alpha _{k}^{\ast }\}_{k=1}^{K}$ and $\{\bar{\nu}%
_{lK}\}_{l=1}^{K_{0}}$, we can make
\begin{equation*}
\Psi _{\infty }^{\prime }=(\Psi _{1,\infty },\Psi _{2,\infty }),
\end{equation*}%
where $\Psi _{1,\infty }$ is a $K\times K$ diagonal matrix with strictly
positive diagonal elements. Therefore, $\Psi _{\infty }$ has rank $K$. By
Theorem \ref{thm:id3}(3), $(\bar{\nu}_{1K},\cdots ,\bar{\nu}_{L_{K},K})$
also has rank $K$. This implies, the limit of the $K\times K$ matrix $%
(\alpha _{1}^{\ast },\cdots ,\alpha _{K}^{\ast })$ is of full rank.
Therefore, there exists a constant $\underline{c}>0$ such that
\begin{equation}
\liminf_{n}\min_{k\neq k^{\prime }}|\alpha _{k}^{\ast }-\alpha _{k^{\prime
}}^{\ast }|>\underline{c}.  \label{eq:c1n}
\end{equation}%
Then \eqref{eq:g} follows \eqref{eq:dc} and Lemma \ref{lem:sbsa}(3) with $%
\hat{\beta}_{in}=\hat{\nu}_{iK}\hat{O}_{Kn}O_{s}$ and $\beta _{in}=\nu _{iK}$%
.

Now we turn to prove \eqref{eq:dc}. Since $(\Pi _{n}^{\tau
})^{1/2}H_{0,K_{0}}(\Pi _{n}^{\tau })^{1/2}\rightarrow (\Pi _{\infty
}^{\prime })^{1/2}H_{0,K_{0}}^{\ast }(\Pi _{\infty }^{\prime })^{1/2}$ and
Assumption \ref{ass:nk2}(2), we have $\inf_{n}|\sigma _{K+1n}-\sigma
_{Kn}|\geq C>0$ for any $K\leq K_{0}-1$. Second, Assumption \ref{ass:rate}
implies \citet[Assumption
11]{SWZ17}. Last, let $d_{i}^{\tau }=d_{i}+\tau .$ Since $\tau \leq Mn\rho
_{n}$ for some $M>0$ and $d_{i}\asymp n\rho _{n}$, we have,
\begin{equation*}
d_{i}^{\tau }/d_{i}\asymp 1.
\end{equation*}%
Therefore, there exist constants $C>c>0$ such that
\begin{equation*}
C\geq \sup_{k,n}n_{k}^{\tau }d_{i}^{\tau }/(nd_{i})\geq
\inf_{k,n}n_{k}^{\tau }d_{i}^{\tau }/(nd_{i})\geq c.
\end{equation*}%
This verifies \citet[Assumption 10]{SWZ17}. Hence, by
\citet[Theorem
3.4]{SWZ17},
\begin{equation}
\sup_{i}(n_{g_{iK_{0}}}^{\tau })^{1/2}\theta _{i}^{-1/2}||\hat{u}_{i}(K)^{T}%
\hat{O}_{Kn}-u_{i}^{T}(K)||\leq C^{\ast }\log ^{1/2}(n)(n\rho _{n}+\tau
)^{-1/2}\leq C^{\ast }C_{1}^{-1/2}\quad a.s.,
\end{equation}%
where $C^{\ast }$ is a constant independent of $n$ and $g_{iK_{0}}$ denotes
the membership index of node $i,,$ viz, $g_{iK_{0}}=k$ if $\left[ Z_{K_{0}}%
\right] _{ik}=1.$

In addition, Lemma \ref{lem:id3}(2) shows that, if $i\in \mathcal{C}%
_{k,K_{0}}$ for any $k=1,\cdots ,K_{0}$, then
\begin{equation*}
\liminf_{n}(n_{k}^{\tau })^{1/2}\theta
_{i}^{-1/2}||u_{i}(K)||=\liminf_{n}||[S_{n}]_{k}(K)||\geq c.
\end{equation*}%
Therefore,
\begin{align}
& \sup_{i}||\hat{\nu}_{iK}^{T}\hat{O}_{Kn}O_{s}-\nu _{iK}^{T}||  \notag
\label{eq:G1} \\
\leq & \sup_{i}\left\Vert \hat{\nu}_{iK}^{T}\hat{O}_{Kn}-\frac{u_{i}^{T}(K)}{%
	||u_{i}(K)||}\right\Vert +\sup_{i}\left\Vert \frac{u_{i}^{T}(K)O_{s}}{%
	||u_{i}(K)||}-\nu _{iK}^{T}\right\Vert   \notag \\
\leq & \sup_{1\leq i\leq n}\frac{||\hat{O}_{Kn}^{T}\hat{u}_{i}(K)-u_{i}(K)||%
}{||\hat{u}_{i}(K)||}+o(1)  \notag \\
\leq & \frac{C^{\ast }C_{1}^{-1/2}}{c-C^{\ast }C_{1}^{-1/2}}+o(1)\leq
c_{1},\quad a.s.,
\end{align}%
where the second inequality holds because of the definition of $\nu _{iK}$
and Theorem \ref{thm:id3}. By Assumption \ref{ass:rate}, $C_{1}$ is
sufficiently large, which implies that $c_{1}^{\prime }$ can be sufficiently
small. This concludes the proof of \eqref{eq:g}.

We also  note that, by definition, for any $K=1,\cdots ,K_{0}$ and $%
k=1,\cdots ,K_{0}$ , there exists $l=1,\cdots ,L_{K}$ such that $\mathcal{C}%
_{k,K_{0}}\subset G_{l,K}$.  In addition, by \eqref{eq:ZK}, Assumption \ref%
{ass:uniqueDC}(1), and Lemma \ref{lem:sbsa}(1), for any $ l=1,\cdots ,L_{K}$%
, there exists $k^{\prime }=1,\cdots ,K$ such that $ G_{l,K}\subset \mathcal{%
	C}_{k^{\prime },K}$. Therefore,
\begin{equation*}
\mathcal{C}_{k,K_{0}}\subset G_{l,K}\subset \mathcal{C}_{k^{\prime },K}\quad
\text{and}\quad Z_{K_{0}}\succeq Z_{K}.
\end{equation*}

\textbf{Second, we prove Theorem \ref{thm:oracleDC}(2). }We know from
Theorem \ref{thm:oracleDC}(1) that $\hat{Z}_{K-1}=Z_{K-1}$ $a.s.$, i.e., $%
\widehat{\mathcal{C}}_{k,K-1}=\mathcal{C}_{k,K-1}$ for $k=1,\cdots ,K-1$. We
aim to show that $\hat{Z}_{K}^{b}=Z_{K}^{b}$ \textit{$a.s.$} for $K=2,\cdots
,K_{0}$. Recall $\widetilde{\mathcal{C}}_{k,K-1}^{l}=\mathcal{C}_{k,K-1}\cap
G_{l,K}$. We divide $[K-1]$ into two subsets $\mathcal{K}_{1}$ and $\mathcal{%
	K}_{2}$ such that $k\in \mathcal{K}_{1}$ if there exists at least two
indexes $l_{1}$ and $l_{2}$ such that both $\widetilde{\mathcal{C}}%
_{k,K-1}^{l_{1}}$ and $\widetilde{\mathcal{C}}_{k,K-1}^{l_{2}}$ are nonempty
sets and $\mathcal{K}_{2}=[K-1]\backslash \mathcal{K}_{1}$. Note that $%
L_{K}\geq K>K-1$. Therefore, by the pigeonhole principle, $\mathcal{K}_{1}$
is nonempty. We divide the proof into three steps. For a generic $k\in
\mathcal{K}_{1}$, denote $\widehat{\mathcal{C}}_{k,K-1}(1)$ and $\widehat{%
	\mathcal{C}}_{k,K-1}(2)$ as two subsets of $\mathcal{C}_{k,K-1}$ which are
obtained by applying k-means algorithm on $\{\hat{\nu}_{in}(K)\}_{i\in
	\mathcal{C}_{k,K-1}}$ with two centroids. Similarly, let $\mathcal{C}%
_{k,K-1}(1)$ and $\mathcal{C}_{k,K-1}(2)$ as two subsets of $\mathcal{C}%
_{k,K-1}$ which are obtained by applying k-means algorithm on $\{\nu
_{iK}\}_{i\in \mathcal{C}_{k,K-1}}$ with two centroids. In the first step,
we aim to show $\hat{k}=k^{\ast }\in \mathcal{K}_{1}$ \textit{$a.s.$}, where
$\hat{k}$ is defined in Algorithm \ref{algo:2} in Section \ref{sec:algo}. In
the second step, we aim to show that $\widehat{\mathcal{C}}_{k^{\ast
	},K-1}(1)=\mathcal{C}_{k^{\ast },K-1}(1)$ and $\widehat{\mathcal{C}}%
_{k^{\ast },K-1}(2)=\mathcal{C}_{k^{\ast },K-1}(2)$ \textit{a.s.} These two
results imply that
\begin{equation*}
\mathcal{C}_{k^{\ast },K-1}(1)=\widehat{\mathcal{C}}_{\hat{k},K-1}(1)\quad
\text{and}\quad \mathcal{C}_{k^{\ast },K-1}(2)=\widehat{\mathcal{C}}_{\hat{k}%
	,K-1}(2),
\end{equation*}%
which completes the proof of $\hat{Z}_{K}^{b}=Z_{K}^{b}$ for $k=1,\cdots
,K_{0}$. Last, in the third step, we show that $Z_{K_{0}}\succeq Z_{K+1}^{b}.
$

\textbf{Step 1. We show that }$\hat{k}=k^{\ast }\in \mathcal{K}_{1}$ \textit{%
	a.s.} For a generic $k\in \mathcal{K}_{1}$, because the $L_{2}$-norm is
invariant under rotation, we can regard the procedure as applying k-means
algorithm to $\hat{\beta}_{in}=O_{s}^{T}\hat{O}_{Kn}^{T}\hat{\nu}_{iK}$ for $%
i\in \mathcal{C}_{k,K-1}$. Further denote $\beta _{in}=\nu _{iK}.$ Then, $%
\beta _{in}=\beta _{jn}$ if $i,j\in \widetilde{\mathcal{C}}_{k,K-1}^{l}$ for
some $l$, and
\begin{align*}
& \sup_{i\in \mathcal{C}_{k,K-1}}||\hat{\beta}_{in}-\beta _{in}|| \\
\leq & \sup_{i\in \mathcal{C}_{k,K-1}}\left\Vert \hat{\nu}_{iK}^{T}\hat{O}%
_{Kn}O_{s}-\frac{u_{i}^{T}(K)}{||u_{i}(K)||}\right\Vert +\sup_{i\in \mathcal{%
		C}_{k,K-1}}\left\Vert \frac{u_{i}^{T}(K)O_{s}}{||u_{i}(K)||}-\nu
_{iK}^{T}\right\Vert  \\
\leq & \frac{C^{\ast }C_{1}^{-1/2}}{c-C^{\ast }C_{1}^{-1/2}}+o(1)\leq
c_{1}\quad a.s.,
\end{align*}%
where the first inequality holds by the triangle inequality, the second
inequality holds because of Theorem \ref{thm:id3}(2) and the fact that the
constant $c_{1}$ is sufficiently small. In addition, by the definition of $%
\{G_{l,K}\}_{l=1}^{L_{K}}$ in Theorem \ref{thm:id3}(2), there exists some
positive constant $c$ such that, for $l\neq l^{\prime }$, $\widetilde{%
	\mathcal{C}}_{k,K}^{l}\neq \emptyset $, and $\widetilde{\mathcal{C}}%
_{k,K}^{l^{\prime }}\neq \emptyset $,
\begin{equation*}
\inf_{i\in \widetilde{\mathcal{C}}_{k,K}^{l},j\in \widetilde{\mathcal{C}}%
	_{k,K}^{l^{\prime }}}||\beta _{in}-\beta _{jn}||\geq c>0.
\end{equation*}%
Recall the definitions of $Q_{K}(\cdot )$ and $\hat{Q}_{K}(\cdot )$ in %
\eqref{eq:Q} and \eqref{eq:QK}, respectively. Then, by Lemma \ref{lem:sbsa}%
(2), we have, for any $k\in \mathcal{K}_{1}$, $|Q_{K-1}(k)-\hat{Q}%
_{K-1}(k)|\leq C^{\prime }c_{1}$ $a.s.$ for some constant $C^{\prime }>0$.
For $k\in \mathcal{K}_{2}$, $Q_{K-1}(k)=o(1)$ and $|\hat{Q}_{K-1}(k)|\leq
C^{^{\prime \prime }}c_{1}$. Therefore, $|Q_{K-1}(k)-\hat{Q}_{K-1}(k)|\leq
Cc_{1}$ $a.s.$ for $k=1,\cdots ,K-1.$ Recall that
\begin{equation*}
k^{\ast }=\argmax_{1\leq k\leq K-1}Q_{K-1}(k)
\end{equation*}%
We claim $\hat{k}=k^{\ast }$ $a.s.$ Suppose not. Then by Assumption \ref%
{ass:uniqueDC}(2),
\begin{equation*}
0\leq \hat{Q}_{K-1}(\hat{k})-\hat{Q}_{K-1}(k^{\ast })=Q_{K-1}(\hat{k}%
)-Q_{K-1}(k^{\ast })+2C^{\prime }c_{1}\leq 2Cc_{1}-c.
\end{equation*}%
As $c_{1}$ is sufficiently small, we reach a contradiction.

\textbf{Step 2.\ We show that } $\widehat{\mathcal{C}}_{k^{\ast },K-1}(1)=%
\mathcal{C}_{k^{\ast },K-1}(1)$ and $\widehat{\mathcal{C}}_{k^{\ast
	},K-1}(2)=\mathcal{C}_{k^{\ast },K-1}(2)$ $a.s.$ Because $Z_{K-1}$ and $%
Z_{K}^{b}$ are unique, Lemma \ref{lem:sbsa}(3) implies, up to some
relabeling,
\begin{equation}
\mathcal{C}_{k^{\ast },K-1}(1)=\widehat{\mathcal{C}}_{k^{\ast },K-1}(1)\quad
\text{and}\quad \mathcal{C}_{k^{\ast },K-1}(2)=\widehat{\mathcal{C}}%
_{k^{\ast },K-1}(2).  \label{eq:hetero}
\end{equation}%
Therefore, $\hat{Z}_{K}^{b}=Z_{K}^{b}$ for $k=1,\cdots ,K_{0}$.

\textbf{Step 3.\ We show that }$Z_{K_0} \succeq Z_{K+1}^b.$ For any $
k=1,\cdots ,K_{0}$ and any $K=2,\cdots ,K_{0}$, Theorem \ref{thm:oracleDC}
(1) shows that there exists $k^{\prime } \in \{1,\cdots ,K-1\}$ such that $
\mathcal{C}_{k,K_{0}}\subset \mathcal{C}_{k^{\prime },K-1}$. If $k^{\prime }
\neq k^{\ast}$,  then $\mathcal{C}_{k,K_{0}}\subset \mathcal{C}_{k^{\prime
	},K-1}=\mathcal{C}_{k^{\prime \prime },K}^{b}$  for some $k^{\prime \prime
}=1,\cdots ,K$. If $k^{\prime }= k^{\ast}$, we  know that $\mathcal{C}%
_{k,K_{0}}\subset G_{l,K}$ for some $l=1,\cdots ,L_{K}$.  Therefore,
\begin{equation*}
\mathcal{C}_{k,K_{0}}\subset \mathcal{C}_{k^{\ast },K-1}\cap G_{l,K}=%
\widetilde{\mathcal{C}}_{k^{\ast },K-1}^{l}.
\end{equation*}
Last, by Lemma \ref{lem:sbsa}, we know that
\begin{equation*}
\widetilde{\mathcal{C}}_{k^{\ast },K-1}^{l}\subset \quad \text{either}\quad
\mathcal{C}_{k^{\ast },K-1}(1)\quad \text{or}\quad \mathcal{C}_{k^{\ast
	},K-1}(2).
\end{equation*}
Therefore, there exists $k^{\prime \prime }=1,\cdots ,K$ such that
\begin{equation*}
\mathcal{C}_{k,K_{0}}\subset \widetilde{\mathcal{C}}_{k^{\ast
	},K-1}^{l}\subset \mathcal{C}_{k^{\prime \prime },K}^{b}.
\end{equation*}
This completes the proof of Theorem \ref{thm:oracleDC}(2).

\textbf{For Theorem \ref{thm:oracleDC}(3), }the result holds by the
construction of $\hat{Z}_{K+1}^{b}$ for $K=1,\cdots ,K_{0}$ and the fact
that $\hat{Z}_{K}=Z_{K}$ for $K=1,\cdots ,K_{0}$.

\bigskip

\subsection{Proof of Theorem \protect\ref{thm:underDC}}

We first state $\mathbb{W}_{K}$: if $K=2$,
\begin{equation*}
\mathbb{W}_{K}=%
\begin{Bmatrix}
& W\in \Re ^{K\times K}:W\text{ is symmetric, } \\
& W_{K-1K-1}(W_{K-1\cdot }+W_{K\cdot })^{2}=W_{K-1\cdot
}^{2}(W_{K-1K-1}+2W_{K-1K}+W_{KK}), \\
& W_{K-1K}(W_{K-1\cdot }+W_{K\cdot })^{2}=W_{K-1\cdot }W_{K\cdot
}(W_{K-1K-1}+2W_{K-1K}+W_{KK}), \\
& W_{KK}(W_{K-1\cdot }+W_{K\cdot })^{2}=W_{K\cdot
}^{2}(W_{K-1K-1}+2W_{K-1K}+W_{KK}), \\
&
\end{Bmatrix}%
\end{equation*}%
and if $K\geq 3$
\begin{equation*}
\mathbb{W}_{K}=%
\begin{Bmatrix}
& W\in \Re ^{K\times K}:W\text{ is symmetric, } \\
& W_{kl}(W_{K-1\cdot }+W_{K\cdot })=W_{l\cdot }(W_{kK-1}+W_{kK}),~k=1,\cdots
,K-2,~l=K-1,K, \\
& W_{K-1K-1}(W_{K-1\cdot }+W_{K\cdot })^{2}=W_{K-1\cdot
}^{2}(W_{K-1K-1}+2W_{K-1K}+W_{KK}), \\
& W_{K-1K}(W_{K-1\cdot }+W_{K\cdot })^{2}=W_{K-1\cdot }W_{K\cdot
}(W_{K-1K-1}+2W_{K-1K}+W_{KK}), \\
& W_{KK}(W_{K-1\cdot }+W_{K\cdot })^{2}=W_{K\cdot
}^{2}(W_{K-1K-1}+2W_{K-1K}+W_{KK}), \\
&
\end{Bmatrix}%
\end{equation*}%
where $W_{k\cdot }=\sum_{l=1}^{K}W_{kl}$ for $W=\left[ W_{kl}\right] \in \Re
^{K\times K}.$

By Theorem \ref{thm:oracleDC}, we have $\hat{Z}_{K}^{b}=Z_{K}^{b}$ $a.s.$
for $K\leq K_{0}$. By Theorem \ref{thm:oracleDC}(3), without loss of
generality, we assume that $\hat{Z}_{K}^{b}=Z_{K}^{b}$ is obtained by
splitting the last group in $\hat{Z}_{K-1}=Z_{K-1}$ into the $(K-1)$-th and $%
K$-th groups in $\hat{Z}_{K}$, i.e.,
\begin{equation*}
\#\mathcal{C}_{k,K-1}=\#\mathcal{C}_{k,K}^{b},\text{ for }k=1,\cdots
,K-2\quad \text{and }\quad \#\mathcal{C}_{K-1,K-1}=\#\mathcal{C}%
_{K-1,K}^{b}\cup \#\mathcal{C}_{K,K}^{b}.
\end{equation*}%
Define $O_{kl,K}^{b}$ and $O_{kl,K}$ as \eqref{eq:O} with $\hat{Z}_{K}$
replaced by $Z_{K}^{b}$ and $Z_{K}$, respectively, and $n_{kl,K}^{b}$ and $%
n_{kl,K}$ as \eqref{eq:nk} with $\hat{Z}_{K}$ replaced by $Z_{K}^{b}$ and $%
Z_{K}$, respectively. Further define
\begin{equation*}
\widehat{M}_{kl,K}=\frac{O_{kl,K}}{(\sum_{l^{\prime }=1}^{K}O_{kl^{\prime
		},K})(\sum_{l^{\prime }=1}^{K}O_{ll^{\prime },K})}\quad \text{and}\quad
\widehat{M}_{kl,K}^{b}=\frac{O_{kl,K}^{b}}{(\sum_{l^{\prime
		}=1}^{K}O_{kl^{\prime },K}^{b})(\sum_{l^{\prime }=1}^{K}O_{ll^{\prime
		},K}^{b})},\quad k\neq l,
\end{equation*}%
\begin{equation*}
\widehat{M}_{kk,K}=\frac{O_{kk,K}}{\sum_{i,j\in C_{k,K},i\neq j}\hat{d}_{i}%
	\hat{d}_{j}},\quad \text{and}\quad \widehat{M}_{kk,K}^{b}=\frac{O_{kk,K}^{b}%
}{\sum_{i,j\in C_{k,K}^{b},i\neq j}\hat{d}_{i}\hat{d}_{j}}.
\end{equation*}%
Then, almost surely, for $i\in \widehat{\mathcal{C}}_{k,K}$ and $i\in
\widehat{\mathcal{C}}_{l,K}$
\begin{equation*}
\hat{P}_{ij}(\hat{Z}_{K})=\widehat{M}_{kl,K}\hat{d}_{i}\hat{d}_{j},
\end{equation*}%
and for $i\in \widehat{\mathcal{C}}_{k,K}^{b}$ and $i\in \widehat{\mathcal{C}%
}_{l,K}^{b}$
\begin{equation*}
\hat{P}_{ij}(\hat{Z}_{K}^{b})=\widehat{M}_{kl,K}^{b}\hat{d}_{i}\hat{d}_{j}.
\end{equation*}

Then, for any $k,l \leq K-2$, if $i \in \mathcal{C}^b_{k,K} = \mathcal{C}
_{k,K-1}$ and $j \in \mathcal{C}^b_{l,K} = \mathcal{C}_{l,K-1}$, we have
\begin{equation*}
O_{kl,K}^b = O_{kl,K-1}, \sum_{i^{\prime }\in \mathcal{C} _{k,K}^b}\hat{d}%
_{i^{\prime }} = \sum_{i^{\prime }\in \mathcal{C}_{k,K-1}}\hat{d}
_{i^{\prime }}, \quad \text{and thus,} \quad \hat{P}_{ij}(\hat{Z}^b_K) =
\hat{P}_{ij}(\hat{Z}_{K-1}).
\end{equation*}


By \eqref{eq:plr},
\begin{align*}
& L_{n}(\hat{Z}_{K}^{b},\hat{Z}_{K-1}) \\
=& 2\sum_{k=1}^{K-2}\biggl\{\sum_{l=K-1}^{K}0.5n_{kl,K}^b\biggl( \frac{%
	\widehat{M}_{kl,K}^b}{\widehat{M}_{kK-1,K-1}}-1\biggr)^{2} \biggr\} \\
& +\biggl\{0.5\biggl[n_{K-1K-1,K}^b\biggl(\frac{\widehat{M} _{K-1K-1,K}^b}{%
	\widehat{M}_{K-1K-1,K-1}}-1\biggr)^{2} \\
& +2n_{K-1K,K}^b\biggl(\frac{\widehat{M}_{K-1K,K}^b}{\widehat{M}_{K-1K-1,K-1}%
}-1\biggr)^{2}+n_{KK,K}^b\biggl(\frac{\widehat{M}_{KK,K}^b}{\widehat{M}%
	_{K-1K-1,K-1}}-1\biggr)^{2}\biggr]\biggr\} \\
=:& 2\sum_{k=1}^{K-2}\hat{I}_{kn}+\widehat{II}_{n}.
\end{align*}
For $i\in \mathcal{C}_{k,K}^{b}$ and $j\in \mathcal{C}_{l,K}^{b}$, $%
k,l=1,\cdots ,K$, the  population counterparts of $\hat{P}_{ij}(\hat{Z}_{K})$
and $\hat{P}_{ij}(\hat{Z}_{K}^{b})$ are
\begin{align}
P_{ij}(Z_{K})=\frac{E[O_{kl,K}]d_{i}d_{j}}{\sum_{i^\prime \in \mathcal{C}%
		_{k,K},j^{\prime }\in \mathcal{C}_{l,K},i^{\prime }\neq j^{\prime
	}}d_{i^{\prime }}d_{j^{\prime }}}:=M_{kl,K}^bd_{i}d_{j}  \label{eq:M}
\end{align}
and
\begin{equation}
P_{ij}(Z_{K}^{b})=\frac{E[O_{kl,K}^b]d_{i}d_{j}}{\sum_{i^\prime \in \mathcal{%
			C}_{k,K}^{b},j^{\prime }\in \mathcal{C}_{l,K}^{b},i^{\prime }\neq j^{\prime
	}}d_{i^{\prime }}d_{j^{\prime }}}:=M_{kl,K}^bd_{i}d_{j},  \label{eq:Mb}
\end{equation}
respectively. Let
\begin{equation}
\tilde{\mathcal{B}}_{K,n}=2\sum_{k=1}^{K-2}I_{kn}+II_{n},  \label{eq:Bias}
\end{equation}
where
\begin{align}
I_{kn}& =\sum_{l=K-1}^{K}0.5n_{kl,K}^b\biggl(\frac{M_{kl,K}^b }{M_{kK-1,K-1}}%
-1\biggr)^{2}\text{ and}  \label{eq:Bias1} \\
II_{n}& =0.5n_{K-1K-1,K}^b\biggl(\frac{M_{K-1K-1,K}^b}{ M_{K-1K-1,K-1}}-1%
\biggr)^{2}  \notag \\
& +n_{K-1K,K}^b\biggl(\frac{M_{K-1K,K}^b}{ M_{K-1K-1,K-1}}-1\biggr)%
^{2}+0.5n_{KK,K}^b\biggl(\frac{ M_{KK,K}^b}{M_{K-1K-1,K-1}}-1\biggr)^{2}.
\label{eq:Bias2}
\end{align}
Note that $O_{kl,K}^b$ is independent across $1\leq k,l\leq K$. Let
\begin{equation*}
V_{kl,K}^b=\frac{\sum_{s\in I(\mathcal{C}_{k,K}^{b}),t\in I(\mathcal{C}%
		_{l,K}^{b})}[n_{\theta }^{(1)}(s,t)H_{st,K_0}-n_{\theta
	}^{(2)}(s,t)H_{st,K_0}B_{st}(Z_{K_{0}})]}{n^{2}},
\end{equation*}
where $n_{\theta }^{(m)}(k)=\sum_{i\in \mathcal{C}_{k,K_{0}}}\theta _{i}^{m}$
for $ m=1,\cdots ,4$,
\begin{equation*}
n_{\theta }^{(1)}(s,t)=n_{\theta }^{(1)}(s)n_{\theta }^{(1)}(t)-n_{\theta
}^{(2)}(s)1\{s=t\},
\end{equation*}
and
\begin{equation*}
n_{\theta }^{(2)}(s,t)=n_{\theta }^{(2)}(s)n_{\theta }^{(2)}(t)-n_{\theta
}^{(4)}(s)1\{s=t\}.
\end{equation*}
Then,
\begin{equation}  \label{eq:normal1}
n^{-1}\rho _{n}^{-1/2}\{O_{kl,K}^b-E[O_{kl,K}^b]\}-N_{K}(k,l)=o_{p}(1),
\quad k\neq l,
\end{equation}
where $N_{K}(k,l)$ is normally distributed with expectation zero and
variance $V_{kl,K}^b$,
\begin{equation*}
n^{-1}\rho _{n}^{-1/2}\{O_{kk,K}^b-E[O_{kk,K}^b]\}-N_{K}(k,k)=o_{p}(1),
\quad k=K-1,K,
\end{equation*}
where $N_{K}(k,k)$ is normally distributed with zero expectation and
variance $2V_{kk,K}^b$, and
\begin{equation*}
\{\{N_{K}(k,l)\}_{k=1,\cdots
	,K-2,l=K-1,K},N_{K}(K-1,K),N_{K}(K-1,K-1),N_{K}(K,K)\}
\end{equation*}
are mutually independent.

Next, we consider the linear expansions for $\hat{I}_{kn}-I_{kn}$ and $
\widehat{II}_{n}-II_{n}$ separately in Steps 1 and 2 below.

\textbf{Step 1. We consider the linear expansion of }$\hat{I}_{kn}-I_{kn}.$
\newline
In this step, we focus on the case in which $k=1,\cdots ,K-2$ and $l=K-1,K$.
Note that
\begin{align*}
\frac{\widehat{M}_{kl,K}^b}{\widehat{M}_{kK-1,K-1}}=& \frac{
	O_{kl,K}^b/[\sum_{l^{\prime }=1}^{K}O_{ll^{\prime },K}^b]}{
	O_{kK-1,K-1}/[\sum_{l^{\prime }=1}^{K-1}O_{K-1l^{\prime },K-1}]} \\
=& \frac{O_{kl,K}^b/[\sum_{l^{\prime }=1}^{K}O_{ll^{\prime },K}^b]}{%
	[\sum_{l=K-1}^{K}O_{kl,K}^b]/[\sum_{l=K-1}^{K} \sum_{l^{\prime
		}=1}^{K}O_{ll^{\prime },K}^b]}.
\end{align*}
Similarly,
\begin{equation}
\frac{M_{kl,K}^b}{M_{kK-1,K-1}}=\frac{E[O_{kl,K}^b]/\{ \sum_{l^{\prime
		}=1}^{K}E[O_{ll^{\prime },K}^b]\}}{ \{\sum_{l=K-1}^{K}E[O_{kl,K}^b]\}/\{%
	\sum_{l=K-1}^{K}\sum_{l^{\prime }=1}^{K}E[O_{ll^{\prime },K}^b]\}}.
\label{eq:Mkl}
\end{equation}
Then, by the delta method and some tedious calculation, we have
\begin{equation*}
n\rho _{n}^{1/2}[\widehat{M}_{kl,K}^b-M_{kl,K}^b]=\frac{ N_{K}(k,l)}{\Gamma
	_{l\cdot,K}^b}-\frac{\Gamma _{kl,K}^b[\sum_{l^{\prime
		}=1}^{K}N_{K}(l,l^{\prime })]}{(\Gamma _{l\cdot,K }^b)^{2}}+o_{p}(1),
\end{equation*}
where $N_{K}(K-1,K)=N_{K}(K,K-1),$
\begin{equation}
\Gamma _{kl,K}^b=n^{-2}\rho _{n}^{-1}E[O_{kl}]=\Gamma_{kl,K}^{0b}+o(1),
\label{eq:gamma1}
\end{equation}
and
\begin{equation}
\Gamma _{l\cdot,K}^b=n^{-2}\rho _{n}^{-1}\sum_{l^{\prime
	}=1}^{K}E[O_{ll^{\prime },K}^b]=\Gamma _{l\cdot,K}^{0b}+o(1).
\label{eq:gamma2}
\end{equation}
Similarly,
\begin{align*}
& n\rho _{n}^{1/2}[\widehat{M}_{kK-1,K-1}-M_{kK-1,K-1}] \\
=& \frac{N_{K}(k,K-1)+N_{K}(k,K)}{\Gamma _{K-1\cdot,K}^b+\Gamma_{K\cdot,K}^b}
\\
& -\frac{[\Gamma _{kK-1,K}^b+\Gamma_{kK,K}^b][\sum_{l^{\prime
		}=1}^{K}N_{K}(l^{\prime },K-1)+N_{K}(l^{\prime },K)]}{[\Gamma
	_{K-1\cdot,K}^b+\Gamma_{K\cdot,K}^b]^{2}}+o_{p}(1).
\end{align*}
By Taylor expansion, we have
\begin{align*}
& n\rho _{n}^{1/2}\biggl(\frac{\widehat{M}_{kl,K}^b}{\widehat{M}_{kK-1,K-1}}-%
\frac{M_{kl,K}^b}{M_{kK-1,K-1}}\biggr) \\
=& \frac{1}{M_{kK-1,K-1}}\biggl[\frac{N_{K}(k,l)}{\Gamma _{l\cdot,K}^b}-%
\frac{\Gamma _{kl,K}^b(\sum_{l^{\prime }=1}^{K}N_{K}(l,l^{\prime }))}{%
	(\Gamma _{l\cdot,K}^{b})^2}\biggr] \\
& -\frac{M_{kl,K}^b}{M_{kK-1,K-1}^{2}}\biggl[\frac{ N_{K}(k,K-1)+N_{K}(k,K)}{%
	\Gamma _{K-1\cdot,K}^b+\Gamma _{K\cdot,K}^b} \\
& -\frac{(\Gamma _{kK-1,K}^b+\Gamma _{kK,K}^b)(\sum_{l^{\prime
		}=1}^{K}N_{K}(l^{\prime },K-1)+N_{K}(l^{\prime },K))}{(\Gamma
	_{K-1\cdot,K}^b+\Gamma_{K\cdot,K}^b)^{2}}\biggr]+o_{p}(1).
\end{align*}
This, in conjunction with the fact that $a^{2}-b^{2}=\left( a-b\right)
^{2}+2\left( a-b\right) b,$ implies that
\begin{align}
& n^{-1}\rho _{n}^{1/2}(\hat{I}_{kn}-I_{kn})  \label{eq:I} \\
=& \sum_{l=K-1}^{K}0.5n^{-1}\rho _{n}^{1/2}n_{kl,K}^b\biggl(\frac{ \widehat{M%
	}_{kl,K}^b}{\widehat{M}_{kK-1,K-1}}-\frac{M_{kl,K}^b }{M_{kK-1,K-1}}\biggr)%
^{2}  \notag \\
& +\sum_{l=K-1}^{K}n^{-1}\rho _{n}^{1/2}n_{kl,K}^b\biggl(\frac{\hat{M }%
	_{kl,K}^b}{\widehat{M}_{kK-1,K-1}}-\frac{M_{kl,K}^b}{ M_{kK-1,K-1}}\biggr)%
\biggl(\frac{M_{kl,K}^b}{M_{kK-1,K-1} }-1\biggr)  \notag \\
=& \sum_{l=K-1}^{K}\pi _{k,K}^b\pi _{l,K}^b\biggl(\frac{ M_{kl,K}^b}{%
	M_{kK-1,K-1}}-1\biggr)  \notag \\
& \times n\rho _{n}^{1/2}\biggl(\frac{ \widehat{M}_{kl,K}^b}{\widehat{M}%
	_{kK-1,K-1}}-\frac{M_{kl,K}^b }{M_{kK-1,K-1}}\biggr)+o_{p}(1)  \notag \\
=& \sum_{l^{\prime }=1}^{K-2}\sum_{l=K-1}^{K}\phi _{l^{\prime
	},l}(k)N_{K}(l^{\prime },l)+\phi _{K-1,K-1}(k)N_{K}(K-1,K-1)+\phi
_{K-1,K}(k)N_{K}(K-1,K)  \notag \\
& +\phi _{K,K}(k)N_{K}(K,K)+o_{p}(1),  \notag
\end{align}
where the second equality follows from the facts that $n_{kl,K}^b =
n_{k,K}^bn_{l,K}^b$, $n_{k,K}^b = \sum_{i=1}^n 1\{[Z_K^b]_{ik}=1\}$, and
\begin{equation*}
\frac{n_{k,K}^b}{n} \rightarrow \pi _{k,K}^b:=\sum_{m\in I(\mathcal{C}%
	_{k,K}^{b})}\pi _{m\infty }
\end{equation*}
with $\pi _{m\infty }$ defined in Assumption \ref{ass:nk2} and that $n\rho
_{n}^{1/2}\rightarrow \infty $ as $n\rightarrow \infty $ under Assumption  %
\ref{ass:rate}. For the last line of the above display,
\begin{align*}
& \phi _{l^{\prime },l}(k) \\
=& \pi _{k,K}^b\pi _{l,K}^b\biggl(\frac{M_{kl,K}^b}{ M_{kK-1,K-1}^{2}}-\frac{%
	1}{M_{kK-1,K-1}}\biggr)\biggl[\frac{ 1\{l^{\prime }=k\}}{\Gamma _{l\cdot,K}^b%
}-\frac{\Gamma_{kl,K}^b}{(\Gamma _{l\cdot,K}^{b})^2}\biggr] \\
& -\sum_{l=K-1}^{K}\pi _{k,K}^b\pi _{l,K}^b\biggl(\frac{ (M_{kl,K}^{b})^2}{%
	M_{kK-1,K-1}^{3}}-\frac{M_{kl,K}^b}{ M_{kK-1,K-1}^{2}}\biggr) \\
& \times \biggl[\frac{1\{l^{\prime }=k\}}{\Gamma _{K-1\cdot,K}^b+\Gamma
	_{K\cdot,K}^b}-\frac{\Gamma_{kK-1,K}^b+\Gamma _{kK,K}^b}{[\Gamma
	_{K-1\cdot,K}^b+\Gamma _{K\cdot,K}^b]^{2}}\biggr],\quad l^{\prime }=1,\cdots
,K-2,\quad l=K-1,K,
\end{align*}
\begin{align*}
& \phi _{K-1,K-1}(k) \\
=& -\pi _{k,K}^b\pi _{K-1,K}^b\biggl(\frac{M_{kK-1,K}^b }{M_{kK-1,K-1}^{2}}-%
\frac{1}{M_{kK-1,K-1}}\biggr)\frac{\Gamma_{kK-1,K}^b}{(\Gamma
	_{K-1\cdot,K}^{b})^2} \\
& +\sum_{l=K-1}^{K}\pi _{k,K}^b\pi _{l,K}^b\biggl(\frac{ (M_{kl,K}^{b})^2}{%
	M_{kK-1,K-1}^{3}}-\frac{M_{kl,K}^b}{ M_{kK-1,K-1}^{2}}\biggr)\frac{\Gamma
	_{kK-1,K}^b+\Gamma _{kK,K}^b}{[\Gamma _{K-1\cdot,K}^b+\Gamma _{K\cdot,K
	}^b]^{2}},
\end{align*}
\begin{align*}
& \phi _{K-1,K}(k) \\
=& -\sum_{l=K-1}^{K}\pi _{k,K}^b\pi _{l,K}^b\biggl(\frac{ M_{kl,K}^b}{%
	M_{kK-1,K-1}^{2}}-\frac{1}{M_{kK-1,K-1}} \biggr)\frac{\Gamma _{kl,K}^b}{%
	(\Gamma _{l\cdot,K}^{b})^2} \\
& +\sum_{l=K-1}^{K}\pi _{k,K}^b\pi _{l,K}^b\biggl(\frac{ (M_{kl,K}^{b})^2}{%
	M_{kK-1,K-1}^{3}}-\frac{M_{kl,K}^b}{ M_{kK-1,K-1}^{2}}\biggr)\frac{2[\Gamma
	_{kK-1,K}^b+\Gamma_{kK,K}^b]}{[\Gamma
	_{K-1\cdot,K}^b+\Gamma_{K\cdot,K}^b]^{2}},
\end{align*}
and
\begin{align*}
& \phi _{K,K}(k) \\
=& -\pi _{k,K}^b\pi _{K,K}^b\biggl(\frac{M_{kK,K}^b}{ M_{kK-1,K-1}^{2}}-%
\frac{1}{M_{kK-1,K-1}}\biggr)\frac{\Gamma_{kK,K}^b}{(\Gamma
	_{K\cdot,K}^{b})^2} \\
& +\sum_{l=K-1}^{K}\pi _{k,K}^b\pi _{l,K}^b\biggl(\frac{ (M_{kl,K}^{b})^2}{%
	M_{kK-1,K-1}^{3}}-\frac{M_{kl,K}^b}{ M_{kK-1,K-1}^{2}}\biggr)\frac{\Gamma
	_{kK-1,K}^b+\Gamma_{kK,K}^b}{[\Gamma _{K-1\cdot,K}^b+\Gamma_{K\cdot,K}^b]^{2}%
}.
\end{align*}

\textbf{Step 2. We consider the linear expansion of $\widehat{II}
	_{n}-II_{n}.$}

Note that
\begin{align*}
& \widehat{M}_{K-1K-1,K}^b-M_{K-1K-1,K}^b \\
=& \frac{O_{K-1K-1,K}^b-E[O_{K-1K-1,K}^b]}{\sum_{i^{\prime },j^{\prime }\in
		\mathcal{C}_{K-1,K}^{b},i^{\prime }\neq j^{\prime }}\hat{d} _{i^{\prime }}%
	\hat{d}_{j^{\prime }}} \\
&-\frac{E[O_{K-1K-1,K}^b][ \sum_{i^{\prime },j^{\prime }\in \mathcal{C}%
		_{K-1,K}^{b},i^{\prime }\neq j^{\prime }}( \hat{d}_{i^{\prime }}\hat{d}%
	_{j^{\prime }}-d_{i^{\prime }}d_{j^{\prime }})]}{ (\sum_{i^{\prime
		},j^{\prime }\in \mathcal{C}_{K-1,K}^{b},i^{\prime }\neq j^{\prime }} \hat{d}%
	_{i^{\prime }}\hat{d}_{j^{\prime }})(\sum_{i^{\prime },j^{\prime }\in
		\mathcal{C}_{K-1,K}^{b},i^{\prime }\neq j^{\prime }}d_{i^{\prime
	}}d_{j^{\prime }})}.
\end{align*}
By the proof of \citet[Lemma 3.1]{SWZ17}, we have, for some positive
constant $C>0$,
\begin{equation}
\sup_{i}|\hat{d}_{i}/d_{i}-1|\leq C(\log ^{1/2}(n)(n\rho _{n})^{-1/2}) \leq
CC_1^{-1/2}\quad a.s.  \label{eq:dhat}
\end{equation}
Therefore,
\begin{align*}
n^{-4}\rho _{n}^{-2}\sum_{i^{\prime },j^{\prime }\in \mathcal{C}%
	_{K-1,K}^{b},i^{\prime }\neq j^{\prime }}\hat{d}_{i^{\prime }}\hat{d}%
_{j^{\prime }}= & n^{-4}\rho _{n}^{-2}\left[(\sum_{i^{\prime }\in \mathcal{C}%
	_{K-1,K}^{b}}\hat{d}_{i^{\prime }})^2 - \sum_{i^{\prime }\in \mathcal{C}%
	_{K-1,K}^{b}}\hat{d}_{i^{\prime }}^2\right] \\
= & n^{-4}\rho _{n}^{-2}\left[\left(\sum_{k=1}^K (EO_{kK-1,K}^b +
O_{kK-1,K}^b- EO_{kK-1,K}^b)\right)^2 - \sum_{i^{\prime }\in \mathcal{C}%
	_{K-1,K}^{b}}\hat{d}_{i^{\prime }}^2\right] \\
= & [\Gamma _{K-1\cdot,K}^{b}+O_p((n\rho_n^{1/2})^{-1})]^2 - n^{-4}\rho
_{n}^{-2}\sum_{i^{\prime }\in \mathcal{C}_{K-1,K}^{b}}\hat{d}_{i^{\prime }}^2
\\
= & (\Gamma _{K-1\cdot,K}^{b})^2 + o_p(1),
\end{align*}
where the third equality holds because $O_{kK-1,K}^b- EO_{kK-1,K}^b =
O_p(n\rho_n^{1/2})$ and the last equality holds because
\begin{align*}
n^{-4}\rho _{n}^{-2}\sum_{i^{\prime }\in \mathcal{C}_{K-1,K}^{b}}\hat{d}%
_{i^{\prime }}^2 \leq n^{-4}\rho _{n}^{-2}\sum_{i^{\prime }\in \mathcal{C}%
	_{K-1,K}^{b}}d_i^2(1+CC_1^{-1/2}) = O_{a.s.}(n^{-1}).
\end{align*}

Also note that, by \eqref{eq:dhat},
\begin{align*}
& n^{-3}\rho _{n}^{-3/2}\sum_{i^{\prime },j^{\prime }\in \mathcal{C}%
	_{K-1,K}^{b},i^{\prime }\neq j^{\prime }}(\hat{d}_{i^{\prime }}\hat{d}%
_{j^{\prime }}-d_{i^{\prime }}d_{j^{\prime }}) \\
=& n^{-3}\rho _{n}^{-3/2}\biggl[(\sum_{i^{\prime }\in \mathcal{C}%
	_{K-1,K}^{b}}\hat{d}_{i^{\prime }})^{2}-(\sum_{i^{\prime }\in \mathcal{C}%
	_{K-1,K}^{b}}d_{i^{\prime }})^{2}\biggr]-n^{-3}\rho _{n}^{-3/2}\biggl[%
\sum_{i^{\prime }\in \mathcal{C}_{K-1,K}^{b}}(\hat{d}_{i^{\prime
}}^{2}-d_{i^{\prime }}^{2})\biggr] \\
=& n^{-3}\rho _{n}^{-3/2}\biggl[(\sum_{i^{\prime }\in \mathcal{C}%
	_{K-1,K}^{b}}\hat{d}_{i^{\prime }}-d_{i^{\prime }})(\sum_{i^{\prime }\in
	\mathcal{C}_{K-1,K}^{b}}d_{i^{\prime }}+\hat{d}_{i^{\prime }})\biggr]%
+o_{a.s.}(1) \\
=& n^{-3}\rho _{n}^{-3/2}\biggl[(\sum_{i^{\prime }\in \mathcal{C}%
	_{K-1,K}^{b}}\hat{d}_{i^{\prime }}-d_{i^{\prime }})2(\sum_{i^{\prime }\in
	\mathcal{C}_{K-1,K}^{b}}d_{i^{\prime }})\biggr]+n^{-3}\rho
_{n}^{-3/2}(\sum_{i^{\prime }\in \mathcal{C}_{K-1,K}^{b}}\hat{d}_{i^{\prime
}}-d_{i^{\prime }})^{2} \\
& +o_{a.s.}(1) \\
=& 2\Gamma _{K-1\cdot ,K}\biggl(\sum_{l^{\prime }=1}^{K}N_{K}(K-1,l^{\prime
})\biggr)+o_{p}(1),
\end{align*}%
where the second equality holds because
\begin{align*}
& n^{-3}\rho _{n}^{-3/2}\left\vert \sum_{i^{\prime }\in \mathcal{C}%
	_{K-1,K}^{b}}(\hat{d}_{i^{\prime }}^{2}-d_{i^{\prime }}^{2})\right\vert  \\
=& n^{-3}\rho _{n}^{-3/2}\left\vert \sum_{i^{\prime }\in \mathcal{C}%
	_{K-1,K}^{b}}(\hat{d}_{i^{\prime }}-d_{i^{\prime }})\right\vert \left\vert
\sum_{i^{\prime }\in \mathcal{C}_{K-1,K}^{b}}(\hat{d}_{i^{\prime
}}+d_{i^{\prime }})\right\vert  \\
\leq & n^{-3}\rho _{n}^{-3/2}\left( 1+CC_{1}^{-1/2}\right) \left[
\sum_{i^{\prime }\in \mathcal{C}_{K-1,K}^{b}}d_{i^{\prime }}\right]
^{2}C(\log ^{1/2}(n)(n\rho _{n})^{-1/2})=o_{a.s.}(1),
\end{align*}%
and the last equality holds because
\begin{equation*}
\sum_{i^{\prime }\in \mathcal{C}_{K-1,K}^{b}}(\hat{d}_{i^{\prime
}}-d_{i^{\prime }})=O_{p}(n\rho _{n}^{1/2}).
\end{equation*}

Then, by the delta method,
\begin{align}
& n^{3}\rho _{n}^{3/2}[\widehat{M}_{K-1K-1,K}^b-M_{K-1K-1,K}^b]
\label{eq:K-1,K-1II} \\
=& \frac{N_{K}(K-1,K-1)}{(\Gamma _{K-1\cdot,K}^{b})^2}-\frac{2\Gamma
	_{K-1K-1,K}^b[\sum_{l^{\prime }=1}^{K}N_{K}(K-1,l^{\prime })]}{ (\Gamma
	_{K-1\cdot,K}^b)^{3}} + o_p(1).  \notag
\end{align}
Similarly,
\begin{equation*}
n^{3}\rho _{n}^{3/2}(\widehat{M}_{KK,K}^b-M_{KK,K}^b) =\frac{ N_{K}(K,K)}{%
	(\Gamma _{K\cdot,K }^{b})^2}-\frac{2\Gamma _{KK,K}^b[\sum_{l^{\prime
		}=1}^{K}N_{K}(K,l^{\prime })]}{(\Gamma_{K\cdot,K }^{b})^3} +o_p(1).
\end{equation*}%
Furthermore, we have
\begin{align*}
& \widehat{M}_{K-1K,K}^b-M_{K-1K,K}^b \\
=& \frac{O_{K-1K,K}^b-E[O_{K-1K,K}^b]}{(\sum_{i^{\prime }\in \mathcal{C}%
		_{K-1,K}^{b}}\hat{d}_{i^{\prime }})(\sum_{j^{\prime }\in \mathcal{C}%
		_{K,K}^{b}}\hat{d} _{j^{\prime }})} \\
& -\frac{E[O_{K-1K,K}^b][(\sum_{i^{\prime }\in \mathcal{C}_{K-1,K}^{b}}\hat{d%
	} _{i^{\prime }})(\sum_{j^{\prime }\in \mathcal{C}_{K,K}^{b}}\hat{d}%
	_{j^{\prime }})-(\sum_{i^{\prime }\in \mathcal{C}_{K-1,K}^{b}}d_{i^{\prime
	}})(\sum_{j^{\prime }\in \mathcal{C}_{K,K}^{b}}d_{j^{\prime }})]}{%
	(\sum_{i^{\prime }\in \mathcal{C}_{K-1,K}^{b}}\hat{ d}_{i^{\prime
	}})(\sum_{j^{\prime }\in \mathcal{C}_{K,K}^{b}}\hat{d}_{j^{\prime
	}})(\sum_{i^{\prime }\in \mathcal{C}_{K-1,K}^{b}}d_{i^{\prime
	}})(\sum_{j^{\prime }\in \mathcal{C}_{K,K}^{b}}d_{j^{\prime }})}.
\end{align*}
Therefore,
\begin{align}
& n^{3}\rho _{n}^{3/2}[\widehat{M}_{K-1K,K}^b-M_{K-1K,K}^b]  \notag \\
=& \frac{N_{K}(K-1,K)}{\Gamma _{K-1\cdot,K}^b\Gamma _{K\cdot,K }^b}
\label{eq:k-1kII} \\
& -\frac{\Gamma _{K-1K,K}^b[\Gamma _{K-1\cdot,K }^b\sum_{l^{\prime
		}=1}^{K}N_{K}(l^{\prime },K)+\Gamma _{K\cdot,K }^b\sum_{l^{\prime
		}=1}^{K}N_{K}(l^{\prime },K-1)]}{(\Gamma _{K-1\cdot,K }^{b})^2(\Gamma
	_{K\cdot,K }^{b})^2} +o_p(1).  \notag
\end{align}%
Finally, noting that
\begin{align*}
& \widehat{M}_{K-1K-1,K-1} \\
=& \frac{O_{K-1K-1,K-1}}{\sum_{i^{\prime },j^{\prime }\in \mathcal{C}%
		_{K-1,K-1},i^{\prime }\neq j^{\prime }}\hat{d}_{i^{\prime }}\hat{d}%
	_{j^{\prime }}} \\
=& \frac{O_{K-1K-1,K}^b+2O_{K-1K,K}^b+O_{KK,K}^b}{ \sum_{i^{\prime
		},j^{\prime }\in \mathcal{C}_{K-1,K}^{b},i^{\prime }\neq j^{\prime }} \hat{d}%
	_{i^{\prime }}\hat{d}_{j^{\prime }}+\sum_{i^{\prime },j_{K,K}^{\prime
			b},i^{\prime }\neq j^{\prime }}\hat{d}_{i^{\prime }}\hat{d}_{j^{\prime
	}}+2\sum_{i^{\prime }\in \mathcal{C}_{K-1,K}^{b},j^{\prime }\in \mathcal{C}%
		_{K,K}^{b}}\hat{d} _{i^{\prime }}\hat{d}_{j^{\prime }}},
\end{align*}
we have
\begin{align}
& n^{3}\rho _{n}^{3/2}(\widehat{M}_{K-1K-1,K-1}-M_{K-1K-1,K-1}) \\
=& \frac{N_{K}(K-1,K-1)+2N_{K}(K-1,K)+N_{K}(K,K)}{[\Gamma _{K-1\cdot,K
	}^b+\Gamma _{K\cdot,K}^b]^{2}}  \notag \\
& -\frac{\Gamma _{K-1K-1,K}^b+2\Gamma _{K-1K,K}^b+\Gamma _{KK,K}^b}{[\Gamma
	_{K-1\cdot,K}^b+\Gamma _{K\cdot,K }^b]^{3}}  \notag \\
& \times \left\{\sum_{l^{\prime }=1}^{K}2[N_{K}(K-1,l^{\prime
})+N_{K}(K,l^{\prime })]\right\} +o_p(1).  \notag
\end{align}
For $s,t=K-1,K$, let $\hat{m}_{st,K}^b=n^{2}\rho _{n}\widehat{M} _{st,K}^b$
and
\begin{equation*}
m_{st,K}^b=n^{2}\rho _{n}M_{st,K}^b=\frac{\Gamma _{st,K}^{0b}}{\Gamma
	_{s\cdot,K}^{0b}\Gamma _{t\cdot,K }^{0b}}[1+o(1)].
\end{equation*}
Define $m_{K-1K-1,K-1}$ and $\hat{m}_{K-1K-1,K-1}$ similarly. By  the
previous calculations, we have
\begin{equation*}
\hat{m}_{st,K}^b=m_{st,K}^b[1+o_{a.s.}(1)].
\end{equation*}
Hence,
\begin{align}
& n\rho _{n}^{1/2}\biggl(\frac{\widehat{M}_{K-1K-1,K}^b}{\widehat{M}%
	_{K-1K-1,K-1}}-\frac{M_{K-1K-1,K}^b}{M_{K-1K-1,K-1}} \biggr)  \notag \\
=& \frac{n^{3}\rho _{n}^{3/2}[\widehat{M}_{K-1K-1,K}^b-M_{K-1K-1,K}^b]}{%
	m_{K-1K-1,K-1}}  \notag \\
& -\frac{m_{K-1K-1,K}^bn^{3}\rho _{n}^{3/2}[\widehat{M}
	_{K-1K-1,K-1}-M_{K-1K-1,K-1}]}{m_{K-1K-1,K-1}^{2}} +o_{p}(1),
\end{align}

\begin{align}
& n\rho _{n}^{1/2}\biggl(\frac{\widehat{M}_{KK,K}^b}{\widehat{M}
	_{K-1K-1,K-1}}-\frac{M_{KK,K}^b}{M_{K-1K-1,K-1}}\biggr)  \notag \\
=& \frac{n^{3}\rho _{n}^{3/2}[\widehat{M}_{KK,K}^b-M_{KK,K}^b]}{
	m_{K-1K-1,K-1}}  \notag \\
& -\frac{m_{KK,K}^bn^{3}\rho _{n}^{3/2}[\widehat{M}
	_{K-1K-1,K-1}-M_{K-1K-1,K-1}]}{m_{K-1K-1,K-1}^{2}} +o_{p}(1),
\end{align}
and
\begin{align}
& n\rho _{n}^{1/2}\biggl(\frac{\widehat{M}_{K-1K,K}^b}{\widehat{M}
	_{K-1K-1,K-1}}-\frac{M_{K-1K,K}^b}{M_{K-1K-1,K-1}}\biggr)  \notag \\
=& \frac{n^{3}\rho _{n}^{3/2}[\widehat{M} _{K-1K,K}^b-M_{K-1K,K}^b]}{%
	m_{K-1K-1,K-1}}  \notag \\
& -\frac{m_{K-1K,K}^bn^{3}\rho _{n}^{3/2}[\widehat{M}
	_{K-1K-1,K-1}-M_{K-1K-1,K-1}]}{m_{K-1K-1,K-1}^{2}} +o_{p}(1).
\label{eq:K-1,K-1K-1II}
\end{align}
Then, by \eqref{eq:K-1,K-1II}--\eqref{eq:K-1,K-1K-1II},
\begin{align}
& n^{-1}\rho _{n}^{1/2}(\widehat{II}_{n}-II_{n})  \label{eq:II} \\
=& n\rho _{n}^{1/2}\biggl[(\pi _{K-1,K}^b)^{2}\biggl(\frac{\widehat{M}
	_{K-1K-1,K}^b}{\widehat{M}_{K-1K-1,K-1}}-\frac{ M_{K-1K-1,K}^b}{%
	M_{K-1K-1,K-1}}\biggr)  \notag \\
& +2\pi _{K-1,K}^b\pi _{K,K}^b\biggl(\frac{\widehat{M} _{K-1K,K}^b}{\widehat{%
		M}_{K-1K-1,K-1}}-\frac{M_{K-1K,K}^b}{ M_{K-1K-1,K-1}}\biggr)  \notag \\
& +(\pi _{K,K}^b)^{2}\biggl(\frac{\widehat{M}_{KK,K}^b}{\widehat{M}
	_{K-1K-1,K-1}}-\frac{M_{KK,K}^b}{M_{K-1K-1,K-1}}\biggr) \biggr]+o_{p}(1)
\notag \\
=& n^{3}\rho _{n}^{3/2}\biggl[\frac{(\pi _{K-1,K}^b)^{2}[\widehat{M}
	_{K-1K-1,K}^b-M_{K-1K-1,K}^b]}{m_{K-1K-1,K-1}}  \notag \\
& +\frac{2\pi _{K-1,K}^b\pi _{K,K}^b[\widehat{M} _{K-1K,K}^b-M_{K-1K,K}^b]}{%
	m_{K-1K-1,K-1}}  \notag \\
& +\frac{(\pi_{K,K}^{b})^2[\widehat{M}_{KK,K}^b-M_{KK,K}^b]}{ m_{K-1K-1,K-1}}%
\biggr]  \notag \\
& +\frac{(\pi _{K-1,K}^b)^{2}m_{K-1K-1,K}^b+2\pi _{K-1,K}^b\pi
	_{K,K}^bm_{K-1K,K}^b+(\pi _{K,K}^b)^{2}m_{KK,K}^b}{m_{K-1K-1,K-1}^{2}}
\notag \\
& \times n^{3}\rho _{n}^{3/2}[\widehat{M} _{K-1K-1,K-1}-M_{K-1K-1,K-1}]+
o_{p}(1)  \notag \\
=& \sum_{l^{\prime }=1}^{K-2}\sum_{l=K-1}^{K}\phi _{l^{\prime
	},l}(K-1)N_{K}(l^{\prime },l)+\phi _{K-1,K-1}(K-1)N_{K}(K-1,K-1)  \notag \\
& +\phi _{K-1,K}(K-1)N_{K}(K-1,K)+\phi _{K,K}(K-1)N_{K}(K,K) + o_{p}(1),
\notag
\end{align}
where, by denoting
\begin{equation*}
\phi =\frac{(\pi _{K-1,K}^b)^{2}m_{K-1K-1,K}^b+2\pi _{K-1,K}^b\pi
	_{K,K}^bm_{K-1K,K}^b+(\pi _{K,K}^b)^{2}m_{KK,K}^b }{m_{K-1K-1,K-1}^{2}},
\end{equation*}
we have
\begin{align*}
& \phi _{l^{\prime },K-1}(K-1) \\
=& -\frac{2(\pi _{K-1,K}^b)^{2}\Gamma _{K-1K-1,K}^b}{(\Gamma
	_{K-1\cdot,K}^b)^{3}m_{K-1K-1,K-1} }-\frac{2\pi _{K-1,K}^b\pi _{K,K}^b\Gamma
	_{K-1K,K}^b}{ \Gamma _{K\cdot,K}^b(\Gamma _{K-1\cdot,K}^b)^{2}m_{K-1K-1,K-1}}
\\
& -\frac{2\phi [ \Gamma _{K-1K-1,K}^b+2\Gamma _{K-1K,K}^b+\Gamma _{KK,K}^b]}{%
	[\Gamma _{K-1\cdot,K}^b+\Gamma _{K\cdot,K}^b]^{3}m_{K-1K-1,K-1}^{2}},\quad
l^{\prime }=1,\cdots ,K-2,
\end{align*}
\begin{align*}
& \phi _{l^{\prime },K}(K-1) \\
=& -\frac{2(\pi _{K,K}^b)^{2}\Gamma _{KK,K}^b}{(\Gamma _{K\cdot,K
	}^b)^{3}m_{K-1K-1,K-1}}- \frac{2\pi _{K-1,K}^b\pi _{K,K}^b\Gamma _{K-1K,K}^b%
}{ (\Gamma _{K\cdot,K }^{b})^2\Gamma _{K-1\cdot,K }^bm_{K-1K-1,K-1}} \\
& -\frac{2\phi [ \Gamma _{K-1K-1,K}^b+2\Gamma _{K-1K,K}^b+\Gamma _{KK,K}^b]}{%
	[\Gamma _{K-1\cdot,K}^b+\Gamma _{K\cdot,K}^b]^{3}m_{K-1K-1,K-1}^{2}} ,\quad
l^{\prime }=1,\cdots ,K-2,
\end{align*}
\begin{align*}
& \phi _{K-1,K-1}(K-1) \\
=& \frac{(\pi _{K-1,K}^b)^{2}}{(\Gamma _{K-1\cdot,K}^b)^{2}m_{K-1K-1,K-1}}-%
\frac{2(\pi _{K-1,K}^b)^{2}\Gamma_{K-1K-1,K}^b}{(\Gamma
	_{K-1\cdot,K}^b)^{3}m_{K-1K-1,K-1} } -\frac{2\pi _{K-1,K}^b\pi
	_{K,K}^b\Gamma _{K-1K,K}^b}{ \Gamma _{K\cdot,K}^b(\Gamma
	_{K-1\cdot,K}^b)^{2}m_{K-1K-1,K-1}} \\
&+\frac{\phi }{[\Gamma _{K-1\cdot,K}^b+\Gamma
	_{K\cdot,K}^b]^{2}m_{K-1K-1,K-1}^{2}}-\frac{2\phi [ \Gamma
	_{K-1K-1,K}^b+2\Gamma _{K-1K,K}^b+\Gamma _{KK,K}^b]}{ [\Gamma
	_{K-1\cdot,K}^b+\Gamma _{K\cdot,K}^b]^{3}m_{K-1K-1,K-1}^{2}},
\end{align*}
\begin{align*}
& \phi _{K,K}(K-1) \\
=& \frac{(\pi _{K,K}^b)^{2}}{(\Gamma _{K\cdot,K}^b)^{2}m_{K-1K-1,K-1}}-\frac{%
	2(\pi _{K,K}^b)^{2}\Gamma _{KK,K}^b}{(\Gamma _{K\cdot,K
	}^b)^{3}m_{K-1K-1,K-1}}-\frac{2\pi _{K-1,K}^b\pi _{K,K}^b\Gamma _{K-1K,K}^b}{
	(\Gamma _{K\cdot,K }^{b})^2\Gamma _{K-1\cdot,K}^bm_{K-1K-1,K-1}} \\
& +\frac{\phi }{[\Gamma _{K-1\cdot,K}^b+\Gamma
	_{K\cdot,K}^b]^{2}m_{K-1K-1,K-1}^{2}} -\frac{2\phi [ \Gamma
	_{K-1K-1,K}^b+2\Gamma _{K-1K,K}^b+\Gamma _{KK,K}^b]}{ [\Gamma
	_{K-1\cdot,K}^b+\Gamma _{K\cdot,K}^b]^{3}m_{K-1K-1,K-1}^{2}},
\end{align*}
and
\begin{align*}
& \phi _{K-1,K}(K-1) \\
=& -\frac{2(\pi _{K-1,K}^b)^{2}\Gamma _{K-1K-1,K}^b}{(\Gamma
	_{K-1\cdot,K}^b)^{3}m_{K-1K-1,K-1}}-\frac{2(\pi _{K,K}^b)^{2}\Gamma _{KK,K}^b%
}{(\Gamma _{K\cdot,K }^b)^{3}m_{K-1K-1,K-1}} +\frac{2\pi _{K-1,K}^b\pi
	_{K,K}^b}{\Gamma _{K-1\cdot,K }^b\Gamma _{K\cdot,K}^bm_{K-1K-1,K-1}} \\
& -\frac{2\pi _{K-1,K}^b\pi _{K,K}^b\Gamma _{K-1K,K}^b[\Gamma _{K-1\cdot,K
	}^b+\Gamma _{K\cdot,K}^b]}{(\Gamma _{K\cdot,K }^b)^{2}(\Gamma
	_{K-1\cdot,K}^{b})^2m_{K-1K-1,K-1}} +\frac{2\phi }{[\Gamma
	_{K-1\cdot,K}^b+\Gamma _{K\cdot,K}^b]^{2}m_{K-1K-1,K-1}^{2}} \\
& -\frac{4\phi [ \Gamma _{K-1K-1,K}^b+2\Gamma _{K-1K,K}^b+\Gamma _{KK,K}^b]}{
	[\Gamma _{K-1\cdot,K}^b+\Gamma _{K\cdot,K}^b]^{3}m_{K-1K-1,K-1}^{2}}.
\end{align*}
Combining \eqref{eq:I} and \eqref{eq:II}, we have
\begin{align*}
& n^{-1}\rho _{n}^{1/2}[L_{n}(\hat{Z}_{K},\hat{Z}_{K-1})-\tilde{\mathcal{B}}
_{K,n}] \\
=& \sum_{l^{\prime }=1}^{K-2}\sum_{l=K-1}^{K}\phi _{l^{\prime
	},l}N_{K}(l^{\prime },l)+\phi _{K-1,K-1}N_{K}(K-1,K-1) \\
& +\phi _{K-1,K}N_{K}(K-1,K)+\phi _{K,K}N_{K}(K,K)+o_{p}(1),
\end{align*}
where
\begin{equation*}
\phi _{l^{\prime },l}=\sum_{k=1}^{K-2}2\phi _{l^{\prime },l}(k)+\phi
_{l^{\prime },l}(K-1),\quad l^{\prime }=1,\cdots ,l,\quad l=K-1,\text{ }K.
\end{equation*}
Letting
\begin{align}
\tilde{\varpi}_{K,n}^{2}=& \sum_{l^{\prime }=1,\cdots ,K-2;\text{ }l=K-1,K;
	\text{ }l^{\prime }\leq l}\phi _{l^{\prime },l}^{2}V_{l^{\prime}l,K}^b+\phi
_{K-1,K-1}^{2}2V_{K-1K-1,K}^b  \notag \\
& +\phi _{K,K}^{2}2V_{KK,K}^b+\phi _{K-1,K}^{2}V_{K-1K,K}^b,  \label{eq:Var}
\end{align}
we have
\begin{equation*}
\tilde{\varpi}_{K,n}^{-1}\left\{n^{-1}\rho _{n}^{1/2}[L_{n}(\hat{Z}_{K},\hat{%
	Z} _{K-1})-\tilde{\mathcal{B}}_{K,n}]\right\}\rightsquigarrow N(0,1).
\end{equation*}

\textbf{Step 3. We now prove the second result in the theorem.} \newline
By (\ref{eq:Bias1}), (\ref{eq:Mkl}), \eqref{eq:gamma1} and \eqref{eq:gamma2}
, for $k=1,\cdots ,K-2$, we have
\begin{equation*}
n^{-2}I_{kn}\rightarrow \sum_{l=K-1}^{K}0.5\pi _{k,K}^b\pi _{l,K}^b\biggl(%
\frac{\Gamma _{kl,K}^{0b}[\Gamma _{K-1\cdot,K }^{0b}+\Gamma _{K\cdot,K
	}^{0b}]}{\Gamma _{l\cdot,K }^{0b}[\Gamma _{kK-1,K}^{b}+\Gamma _{kK,K}^{0b}] }%
-1\biggr)^{2}.
\end{equation*}
Similarly, by (\ref{eq:Bias2}), (\ref{eq:Mkl}), \eqref{eq:gamma1} and  %
\eqref{eq:gamma2}, we have
\begin{align*}
& n^{-2}II_{n} \\
\rightarrow & 0.5(\pi _{K-1,K}^b)^{2}\biggl(\frac{\Gamma
	_{K-1K-1,K}^{0b}[\Gamma _{K-1\cdot,K}^{0b}+\Gamma _{K\cdot,K }^{0b}]^{2}}{%
	[\Gamma _{K-1\cdot,K}^{0b}]^{2}[\Gamma _{K-1K-1,K}^{0b}+2\Gamma
	_{K-1K,K}^{0b}+\Gamma _{KK,K}^{0b}]}-1\biggr)^{2} \\
& +\pi _{K-1,K}^b\pi _{K,K}^b \\
& \times \biggl(\frac{\Gamma _{K-1K,K}^{0b}[\Gamma _{K-1\cdot,K}^{0b}+\Gamma
	_{K\cdot,K}^{0b}]^{2}}{\Gamma _{K-1\cdot,K}^{0b}\Gamma
	_{K\cdot,K}^{0b}[\Gamma _{K-1K-1,K}^{0b}+2\Gamma _{K-1K,K}^{0b}+\Gamma
	_{KK,K}^{0b}]}-1\biggr)^{2} \\
& +0.5(\pi _{K,K}^b)^{2}\biggl(\frac{\Gamma _{KK,K}^{0b}[\Gamma
	_{K-1\cdot,K}^{0b}+\Gamma _{K\cdot,K}^{0b}]^{2}}{[\Gamma _{K\cdot,K
	}^{0b}]^{2}[\Gamma _{K-1K-1,K}^{0b}+2\Gamma _{K-1K,K}^{0b}+\Gamma
	_{KK,K}^{0b}]}-1\biggr)^{2}.
\end{align*}
Clearly, there exits $c_{K2}<\infty $ such that
\begin{equation*}
n^{-2}\tilde{\mathcal{B}}_{K,n}=\sum_{k=1}^{K-2}n^{-2}I_{kn}+n^{-2}II_{n}
\leq c_{K2}.
\end{equation*}
In addition, Assumption \ref{ass:BKDC} implies that at least one of the  squares is nonzero. Therefore, there exists
a constant $c_{k1}>0$ such that
\begin{equation*}
n^{-2}\tilde{\mathcal{B}}_{K,n}=\sum_{k=1}^{K-2}n^{-2}I_{kn}+n^{-2}II_{n}
\geq c_{K1}.
\end{equation*}

\subsection{Proof of Theorem \protect\ref{thm:overDC}}

We consider the upper bound for $L_n(\hat{Z}^b_{K_0+1},\hat{Z}_{K_0}) $. We
say $z$  is a $n \times (K_0+1)$ membership matrix for $n$ nodes and $K_0+1$
groups  if there is only one element in each row of $z$ that takes value 1,
and the  rest of the entries are zero. Say $Z_{ik}=1$, then we say that the $%
i$-th  node is identified in group $k$. Let
\begin{align*}
\mathcal{V}_{K_0+1} = &
\begin{Bmatrix}
& z \text{ is a $n \times (K_0+1)$ membership matrix s.t.} \\
& \text{every group identified by z is a subset of } \\
& \text{one of the true communities and} \\
& \inf_{1\leq k \leq K}n_k(z)/n \geq \varepsilon%
\end{Bmatrix}
.
\end{align*}

Without loss of generality, we assume that $\hat{Z}_{K_{0}+1}^{b}$ is
obtained by splitting the last group in $\hat{Z}_{K_{0}}$ into the $K_{0}$
-th and $(K_{0}+1)$-th groups in $\hat{Z}_{K_{0}+1}^{b}$. By Theorem \ref%
{thm:oracleDC} and Assumption \ref{ass:neps}, we have $\hat{Z}%
_{K_{0}+1}^{b}\in \mathcal{V}_{K_{0}+1}$ $a.s.$ Let $z_{K_{0}+1}$ be an
arbitrary realization of $\hat{Z}_{K_{0}+1}^{b}$ such that $z_{K_{0}+1}\in
\mathcal{V}_{K_{0}+1}$ and $h(\cdot |z_{K_{0}+1})$ be a surjective mapping: $%
[K_{0}+1]\mapsto \lbrack K_{0}]$ that maps the community index identified by
$z_{K_{0}+1}$ into the true community index in $[K_{0}]$ for any $%
z_{K_{0}+1}\in \mathcal{V}_{K_{0}+1}$. Then, we have
\begin{equation*}
h(k|z_{K_{0}+1})=k,\quad k=1,\cdots ,K_{0}-1
\end{equation*}%
and
\begin{equation*}
h(K_{0}|z_{K_{0}+1})=h(K_{0}+1|z_{K_{0}+1})=K_{0}.
\end{equation*}%
In the following, we explicitly write down the terms $M_{kl}$, $\widehat{M}%
_{kl}$, and $O_{kl}$ as functions of $z_{K_{0}+1}$, i.e.,
\begin{equation}
M_{kl}(z_{K_{0}+1})=\frac{E[O_{kl}(z_{K_{0}+1})]}{\sum_{i^{\prime }\in
		\mathcal{C}_{k}(z_{K_{0}+1}),j^{\prime }\in \mathcal{C}_{l}(z_{K_{0}+1}),i^{%
			\prime }\neq j^{\prime }}d_{i^{\prime }}d_{j^{\prime }}},  \label{eq:M'}
\end{equation}%
\begin{equation*}
\widehat{M}_{kl}(z_{K_{0}+1})=\frac{O_{kl}(z_{K_{0}+1})}{(\sum_{l^{\prime
		}=1}^{K}O_{kl^{\prime }}(z_{K_{0}+1}))(\sum_{l^{\prime }=1}^{K}O_{ll^{\prime
	}}(z_{K_{0}+1}))},
\end{equation*}%
and
\begin{equation*}
O_{kl}(z_{K_{0}+1})=\sum_{i=1}^{n}\sum_{j\neq
	i}1\{[z_{K_{0}+1}]_{ik}=1,[z_{K_{0}+1}]_{jl}=1\}A_{ij},
\end{equation*}%
where $\mathcal{C}_{l}(z_{K_{0}+1})$ denotes the $l$-th cluster identified
by $z_{K_{0}+1}$. Further recall $n_{kl}$ and $n_{k}$ defined in %
\eqref{eq:nk'} in Section \ref{sec:over}. We emphasize the dependence on $%
z_{K_{0}+1}$ because, by Theorem \ref{thm:oracleDC}, $Z_{K}$ and $Z_{K}^{b}$
for $K=1,\cdots ,K_{0}$ are uniquely defined, while $Z_{K_{0}+1}^{b}$ is
not. By \eqref{eq:M'}, for any $z_{K_{0}+1}\in \mathcal{V}_{K_{0}+1}$, $i\in
\mathcal{C}_{k}(z_{K_{0}+1})$ and $j\in \mathcal{C}_{l}(z_{K_{0}+1})$, $%
k=1,\cdots ,K_{0}-1$, $l=K_{0},K_{0}+1$. Then,
\begin{equation*}
P_{ij}(z_{K_{0}+1})=B_{h(k|z_{K_{0}+1})h(l|z_{K_{0}+1})}\theta _{i}\theta
_{j}=B_{kK_{0},K_{0}}\theta _{i}\theta _{j}=P_{ij}(Z_{K_{0}})
\end{equation*}%
and
\begin{equation}
\frac{M_{kl}(z_{K_{0}+1})}{M_{kK_{0},K_{0}}}=\frac{P_{ij}(z_{K_{0}+1})}{%
	P_{ij}(Z_{K_{0}})}=1,\quad k=1,\cdots ,K_{0}-1,\quad l=K_{0},K_{0}+1.
\label{eq:Over1}
\end{equation}%
Similarly,
\begin{equation}
\frac{M_{K_{0}K_{0}}(z_{K_{0}+1})}{M_{K_{0}K_{0},K_{0}}}=\frac{%
	M_{K_{0}K_{0}+1}(z_{K_{0}+1})}{M_{K_{0}K_{0},K_{0}}}=\frac{%
	M_{K_{0}+1K_{0}+1}(z_{K_{0}+1})}{M_{K_{0}K_{0},K_{0}}}=1.  \label{eq:Over2}
\end{equation}

By Theorem \ref{thm:oracleDC}, $\hat{Z}_{K_{0}}=Z_{K_{0}}$ and $\hat{Z}%
_{K_{0}+1}^{b}\in \mathcal{V}_{K_{0}+1}$ $a.s.$ Therefore, \eqref{eq:Over1}
and \eqref{eq:Over2} still hold when $z_{K_{0}+1}$ and $Z_{K_{0}}$ are
replaced by $\hat{Z}_{K_{0}+1}^{b}$ and $\hat{Z}_{K_{0}}$. Then,
\begin{align}
& L_{n}(\hat{Z}_{K_{0}+1}^{b},\hat{Z}_{K_{0}})  \notag \\
=& 2\sum_{k=1}^{K_{0}-1}\sum_{l=K_{0}}^{K_{0}+1}0.5n_{kl}(\hat{Z}%
_{K_{0}+1}^{b})\biggl(\frac{\widehat{M}_{kl}(\hat{Z}_{K_{0}+1}^{b})}{%
	\widehat{M}_{kK_{0},K_{0}}}-1\biggr)^{2}  \notag \\
& +0.5\biggl[n_{K_{0}K_{0}}(\hat{Z}_{K_{0}+1}^{b})\biggl(\frac{\widehat{M}%
	_{K_{0}K_{0}}(\hat{Z}_{K_{0}+1}^{b})}{\widehat{M}_{K_{0}K_{0},K_{0}}}-1%
\biggr)^{2}  \notag \\
& +2n_{K_{0}K_{0}+1}(\hat{Z}_{K_{0}+1}^{b})\biggl(\frac{\widehat{M}%
	_{K_{0}K_{0}+1}(\hat{Z}_{K_{0}+1}^{b})}{\widehat{M}_{K_{0}K_{0},K_{0}}}-1%
\biggr)^{2}  \notag \\
& +n_{K_{0}+1K_{0}+1}(\hat{Z}_{K_{0}+1}^{b})\biggl(\frac{\widehat{M}%
	_{K_{0}+1,K_{0}+1}(\hat{Z}_{K_{0}+1}^{b})}{\widehat{M}_{K_{0}K_{0},K_{0}}}-1%
\biggr)^{2}\biggr].  \label{eq:LR1}
\end{align}%
For the first term in (\ref{eq:LR1}),
\begin{equation*}
0.5n_{kl}(\hat{Z}_{K_{0}+1}^{b})\biggl(\frac{\widehat{M}_{kl}(\hat{Z}%
	_{K_{0}+1}^{b})}{\widehat{M}_{kK_{0},K_{0}}}-1\biggr)^{2}\lesssim
n^{2}\sup_{z_{K_{0}+1}\in \mathcal{V}_{K_{0}+1}}\biggl(\frac{\widehat{M}%
	_{kl}(z_{K_{0}+1})}{\widehat{M}_{kK_{0},K_{0}}}-\frac{M_{kl}(z_{K_{0}+1})}{%
	M_{kK_{0},K_{0}}}\biggr)^{2}.
\end{equation*}%
The rate of the RHS of the above display depends on that of
\begin{equation*}
\sup_{z_{K_{0}+1}\in \mathcal{V}%
	_{K_{0}+1}}|O_{kl}(z_{K_{0}+1})-E[O_{kl}(z_{K_{0}+1})]|.
\end{equation*}%
By Bernstein inequality,
\begin{align*}
& P(\sup_{z_{K_{0}+1}\in \mathcal{V}%
	_{K_{0}+1}}|O_{kl}(z_{K_{0}+1})-E[O_{kl}(z_{K_{0}+1})]|\geq Cn^{3/2}\rho
_{n}^{1/2}) \\
\leq & 2^{n}\exp \biggl(-\frac{C^{2}n^{3}\rho _{n}/2}{\overline{\theta }%
	^{2}n^{2}\rho _{n}+Cn^{3/2}\rho _{n}^{1/2}/3}\biggr)\leq \exp (-C^{\prime }n)
\end{align*}%
for some constant $C^{\prime }>0$. Therefore,
\begin{equation*}
\sup_{z_{K_{0}+1}\in \mathcal{V}%
	_{K_{0}+1}}|O_{kl}(z_{K_{0}+1})-E[O_{kl}(z_{K_{0}+1})]|=O_{a.s.}(n^{3/2}\rho
_{n}^{1/2}).
\end{equation*}%
It also implies the uniform consistency that
\begin{equation*}
\sup_{z_{K_{0}+1}\in \mathcal{V}_{K_{0}+1}}|n^{-2}\rho
_{n}^{-1}O_{kl}(z_{K_{0}+1})-\Gamma _{kl}(z_{K_{0}+1})|=O_{a.s.}((n\rho
_{n})^{-1/2})+o(1)=o_{a.s.}(1),
\end{equation*}%
where
\begin{equation*}
\Gamma _{kl}(z_{K_{0}+1})=\frac{n_{l}(z_{K_{0}+1})}{n}\frac{%
	n_{k}(z_{K_{0}+1})}{n}H_{h(k|z_{K_{0}+1})h(l|z_{K_{0}+1})}.
\end{equation*}

Following the same and tedious Taylor expansion detailed in Steps 1 and 2 of
the proof of Theorem \ref{thm:underDC}, we have
\begin{equation*}
\sup_{z_{K_{0}+1}\in \mathcal{V}_{K_{0}+1}}\left\vert \widehat{M}%
_{kl}(z_{K_{0}+1})-M_{kl}(z_{K_{0}+1})\right\vert =O_{a.s.}((n^{5/2}\rho
_{n}^{3/2})^{-1}),
\end{equation*}%
\begin{equation*}
|\widehat{M}_{kK_{0},K_{0}}-M_{kK_{0},K_{0}}|=O_{p}((n^{3}\rho
_{n}^{3/2})^{-1}),
\end{equation*}%
and
\begin{equation*}
n^{2}\rho _{n}M_{kK_{0},K_{0}}\geq c,
\end{equation*}%
for some constant $c>0$. Therefore,
\begin{equation*}
\sup_{z_{K_{0}+1}\in \mathcal{V}_{K_{0}+1}}\biggl|\frac{\widehat{M}%
	_{kl}(z_{K_{0}+1})}{\widehat{M}_{kK_{0},K_{0}}}-\frac{M_{kl}(z_{K_{0}+1})}{%
	M_{kK_{0},K_{0}}}\biggr|=O_{p}((n\rho _{n})^{-1/2})
\end{equation*}%
and
\begin{equation*}
0.5n_{kl}(\hat{Z}_{K_{0}+1}^{b})\biggl(\frac{\widehat{M}_{kl}(\hat{Z}%
	_{K_{0}+1}^{b})}{\widehat{M}_{kK_{0},K_{0}}}-1\biggr)^{2}=O_{p}(n\rho
_{n}^{-1}).
\end{equation*}%
The rest of the terms in (\ref{eq:LR1}) can be bounded similarly. Thus, we
conclude that
\begin{equation}
L_{n}(\hat{Z}_{K_{0}+1}^{b},\hat{Z}_{K_{0}})=O_{p}(n\rho _{n}^{-1}).
\label{eq:LR2}
\end{equation}

Next, we study the asymptotic property of $\hat{K}_{1}.$ If $K_{0}=1$, $P(
\hat{K}_{1}\geq 1)=1$ holds trivially. If $K_{0}\geq 2$,
\begin{equation*}
R(1)\asymp \frac{n^{2}}{\eta _{n}}\asymp 1.
\end{equation*}
When $2\leq K<K_{0}$, by Theorem \ref{thm:underDC},
\begin{equation*}
R(K)\asymp \frac{\tilde{\mathcal{B}}_{K-1}+O_{p}(n\rho _{n}^{-1/2})}{\tilde{
		\mathcal{B}}_{K}+O_{p}(n\rho _{n}^{-1/2})}\asymp 1.
\end{equation*}
When $K=K_{0}$, by Theorem \ref{thm:underDC} and (\ref{eq:LR2}),
\begin{equation*}
R(K_{0})\lesssim \frac{n\rho _{n}^{-1}}{c_{K1}n^{2}+O_{p}(n\rho _{n}^{-1/2})}
\rightarrow 0.
\end{equation*}
Since $n^{2}/(n\rho _{n}^{-1})=n\rho _{n} \geq C_1\log(n)\rightarrow \infty $
under  Assumption \ref{ass:rate}, 
\begin{equation*}
P(\hat{K}_{1}\geq K_{0})\leq P\left( R(K_{0})<\max_{K<K_{0}}R(K)\right)
\rightarrow 1.
\end{equation*}

Now, we study the asymptotic property of $\tilde{K}_{2}.$ If $K_{0}=1$,
\begin{equation*}
R(1)\lesssim \frac{1}{n\rho _{n}}\rightarrow 0.
\end{equation*}
Therefore, $P(\tilde{K}_{2}=1)=P(R(1)\leq h_{n})\rightarrow 1$ because $
n\rho _{n}h_{n}\rightarrow \infty $ as $n\rightarrow \infty .$ If $K_{0}\geq
2$, by Theorem \ref{thm:underDC} and (\ref{eq:LR2}),
\begin{equation*}
\begin{cases}
R(K)\asymp & \frac{n^{2}}{n\rho _{n}}\rightarrow \infty ,\quad \text{\ if }
K=1, \\
R(K)\asymp & 1,\quad \text{if }2\leq K<K_{0}, \\
R(K)\lesssim & \frac{n\rho _{n}^{-1}}{n^{2}}\asymp \frac{1}{n\rho _{n}}
\rightarrow 0,\quad \text{if }K=K_{0}.%
\end{cases}%
\end{equation*}
This, in conjunction with the conditions that $n\rho _{n}h_{n}\rightarrow
\infty $ and $h_{n}\rightarrow 0$ as $n\rightarrow \infty $ implies that
\begin{equation*}
P(\tilde{K}_{2}=K_{0})=P\left( \min_{1\leq K<K_{0}}R(K)>h_{n},R(K_{0})\leq
h_{n}\right) \rightarrow 1.
\end{equation*}
It follows that $P(\hat{K}_{2}=K_{0})\geq P(\hat{K}_{1}\geq K_{0},$ $\tilde{%
	K	}_{2}=K_{0})\rightarrow 1.$

\section{Technical lemmas}

\begin{lemma}
	\label{lem:id3} Suppose Assumptions \ref{ass:id3} and \ref{ass:nk2} hold.
	Let $u_{i}^{T}$ be the $i$-th row of $U_{1n}$.
	\begin{enumerate}
		\item[(1)]  There exists a $%
		K_{0}\times K_{0}$ matrix $S_{n}^{\tau }$ such that $(S_{n}^{\tau
		})^{T}S_{n}^{\tau }=I_{K_{0}}$ and $U_{1n}=\Theta _{\tau
		}^{1/2}Z_{K_{0}}(Z_{K_{0}}^{T}\Theta _{\tau }Z_{K_{0}})^{-1/2}S_{n}^{\tau }$.
		\item[(2)] Let $[S_{n}^{\tau }](K)$ and $[S_{n}^{\tau
		}]_{k}(K)$ denote the first $K$ columns of $S_{n}^{\tau }$ and its $%
		k $-th row, respectively. There exist some $K\times K$ orthonormal matrix $%
		O_{s}$, a $K_{0}\times K_{0}$ matrix $S_{\infty }$ and a positive constant $%
		c $ such that for any $K\leq K_{0}$, $[S_{n}^{\tau }]_{k}(K)O_{s}\rightarrow [ S_{\infty }](K)$, $[S_{\infty }](K)$ has
		rank $K$, and for any $k=1,\cdots ,K_{0}$ and $K=1,\cdots ,K_{0}$,
		\begin{equation*}
		\liminf_{n}||[S_{n}^{\tau }]_{k}(K)||\geq \underline{c}.
		\end{equation*}
	\end{enumerate}
\end{lemma}%
\begin{proof}[of Lemma \ref{lem:id3}]
	The first result is proved in \cite{SWZ17}.	For part
	(2), by the proof of Theorem \ref{thm:id3}(2), we have
	$$S_{n}^{\tau }[K]O_{s}\rightarrow
	S_{\infty }[K]$$ where $S_{\infty }$ is the eigenvector matrix of $\Pi
	_{\infty }^{\prime 1/2}H_{0,K_0}^{\ast }\Pi _{\infty }^{\prime 1/2}$
	and is of full rank, and $O_{s}$ is a $K\times
	K$ orthogonal matrix. In addition, by Assumptions \ref{ass:id3}(2) and \ref%
	{ass:nk2}, all elements in $\Pi _{\infty }^{\prime 1/2}H_{0,K_0}^{\ast
	}\Pi _{\infty }^{\prime 1/2}$ are positive. By
	\citet[Lemma
	8.2.1]{HJ90}, all elements in the first column of $S_{\infty }$ are strictly
	positive. This implies that, for any $k=1,\cdots ,K_{0}$,
	\begin{equation*}
	\liminf_{n}||[S_{n}^{\tau }]_{k}(K)||=\liminf_{n}||[S_{n}^{\tau
	}]_{k}(K)O_{s}||=||[S_{\infty }]_{k}(K)||\geq ||[S_{\infty
	}]_{k1}||>0.
	\end{equation*}
	This concludes the proof.
\end{proof}

\bigskip The following lemma is largely based on \citet[Theorem 3.2]{SW17}
and \citet[Theorem 2.3]{SWZ17}.

\begin{lemma}
	\label{lem:sbsa} Let $\mathcal{C}$ be a set of nodes and $\{\hat{\beta}_{in}\}_{i\in \mathcal{C}} $ be a sequence of $d_\beta \times 1$ vectors such that $\sup_{i\in \mathcal{C}}||\hat{%
		\beta}_{in}-\beta _{in}||\leq c_1~a.s.$ and $\sup_{i\in \mathcal{C}}||\beta
	_{in}|| \leq M $ for some sufficiently small constant $c_1>0$ and some constant $M>0$, respectively. In addition, suppose $\{\beta _{in}\}_{i\in \mathcal{C}}$ has $L$
	distinct vectors for some $L \geq K$ and we group index $i$ into $L$ mutually exclusive groups $%
	\{\mathcal{C}_{l}\}_{l=1}^{L}$ such that if $i,j\in \mathcal{C}_{l}$, $\beta _{in}=\beta _{jn}$
	and for any $i\in \mathcal{C}_{l}$, $j\in \mathcal{C}_{l^{\prime }}$, $l\neq l^{\prime }$, $%
	\inf_{i,j,n}||\beta _{in}-\beta _{jn}||>c_2>0$. Let $\pi_l = \frac{\# \mathcal{C}_l}{n}$, $l=1,\cdots,L$. Then, $\min_{l=1,\cdots,L}\pi_l \geq \underline{\pi} >0$. We apply k-means algorithm on $\{\beta _{in}\}_{i=1}^{n}$ and $\{\hat{\beta}%
	_{in}\}_{i=1}^{n}$ and obtain $K$ sets of mutually exclusive groups $%
	(\mathcal{C}(1),\cdots,\mathcal{C}(K))$ and $(\widehat{\mathcal{C}}(1),\cdots,\widehat{\mathcal{C}}(K))$, respectively. Suppose $\mathcal{C}(k)$, $k=1,\cdots,K$ are uniquely defined, then
	\begin{enumerate}
		\item[(1)] 	for
		any $l=1,\cdots ,L$,
		\begin{equation*}
		\mathcal{C}_{l}\subset \text{one of } \{\mathcal{C}(k), k=1,\cdots,K\};
		\end{equation*}%
		\item[(2)] 	
		\begin{equation*}
		\left|\frac{\widehat{\Phi }(\mathcal{C})-\sum_{k=1}^K \widehat{\Phi }(\widehat{\mathcal{C}}(k))}{\#\mathcal{C}}-\frac{\Phi (\mathcal{C})-\sum_{k=1}^k\Phi (\mathcal{C}(k))}{\#\mathcal{C}}\right|\leq Cc_1,~ a.s.,
		\end{equation*}%
		where $C>0$ is some constant independent of $n$ and for a generic index set $\mathcal{C}$,
		\begin{equation*}
		\widehat{\Phi }(\mathcal{C})=\sum_{i\in \mathcal{C}}||\hat{\beta}_{in}-\frac{\sum_{i\in \mathcal{C}}\hat{%
				\beta}_{in}}{\#\mathcal{C}}||^{2}
		\end{equation*}%
		and
		\begin{equation*}
		\Phi (\mathcal{C})=\sum_{i\in \mathcal{C}}||\beta _{in}-\frac{\sum_{i\in \mathcal{C}}\beta _{in}}{\#\mathcal{C}}%
		||^{2}; \quad \text{and}
		\end{equation*}%
		\item[(3)] after relabeling, $\widehat{\mathcal{C}}(k)=\mathcal{C}(k)$, $k=1,\cdots,K$.
	\end{enumerate}
\end{lemma}

\begin{proof}[of Lemma \protect\ref{lem:sbsa}]
	Following the proof of \citet[Theorem 3.2]{SW17}, we focus on the case $L=3$. The proof for $L \geq 4$ is similar but require more notation. When $K=1$, the results are trivial. When $K =3$, Lemma \ref{lem:sbsa}(1) is trivial as $\mathcal{C}(k) = \mathcal{C}_k$, $k=1,2,3$ after relabeling. Lemma \ref{lem:sbsa}(3) directly follows \citet[Theorem 2.3]{SWZ17}, given that $c_1$ is sufficiently small so that
	\begin{align*}
	(2c_1 \underline{\pi}^{1/2} + 16 K^{3/4}M^{1/2}c_1)^2 \leq \underline{\pi}c_2^2.
	\end{align*}
	Given Lemma \ref{lem:sbsa}(3), Lemma \ref{lem:sbsa}(2) holds with $C = 16M$ because
	\begin{align*}
	\left|\left\Vert\hat{\beta}_{in} - \frac{\sum_{i \in \mathcal{C}}\hat{\beta}_{in}}{\#\mathcal{C}}\right\Vert^2 - \left\Vert\beta_{in} - \frac{\sum_{i \in \mathcal{C}}\beta_{in}}{\#\mathcal{C}}\right\Vert^2\right| \leq 8Mc_1.
	\end{align*}
	
	Next, we proof Lemma \ref{lem:sbsa} for $K=2$. Denote $\bar{\beta}_l$, $l=1,2,3$ as the true values $\beta_{in}$ can take when $i \in \mathcal{C}_1$, $\mathcal{C}_2$, and $\mathcal{C}_3$, respectively. \\
	
	\textbf{Step 1. Proof of Lemma \ref{lem:sbsa}(1).} Suppose
	\begin{align}
	\label{eq:suppose}
	\frac{\pi_2 \pi_3}{\pi_2+\pi_3}||\bar{\beta}_2-\bar{\beta}_3||^2<\frac{\pi_1 \pi_3}{\pi_1+\pi_3}||\bar{\beta}_1-\bar{\beta}_3||^2<\frac{\pi_1 \pi_2}{\pi_1+\pi_2}||\bar{\beta}_1-\bar{\beta}_2||^2
	\end{align}
	In this case, we aim to show that $\mathcal{C}(1) = \mathcal{C}_1$ and $\mathcal{C}(2) = \mathcal{C}_2 \cup \mathcal{C}_3$.
	Suppose that, by the k-means algorithm, $n\pi_l^*$ nodes of $i \in \mathcal{C}_l$, $\pi_l^* \in [0,\pi_l]$, $l=1,2,3$ are classified into $\mathcal{C}(1)$ and the rest are in $\mathcal{C}(2)$. We aim to show that \eqref{eq:suppose} implies $\pi_1^* = \pi_1$ and $\pi_2^* = \pi_3^* = 0$. The k-means objective function for the classification $(\mathcal{C}(1),\mathcal{C}(2))$ is
	\begin{align*}
	F(\alpha_1,\alpha_2;\pi^*_1,\pi^*_2,\pi^*_3) \equiv \sum_{l=1}^3\pi_l^*||\bar{\beta}_l - \alpha_1||^2 + \sum_{l=1}^3(\pi_l-\pi_l^*)||\bar{\beta}_l - \alpha_2||^2,
	\end{align*}
	where $\alpha_1 = \frac{\sum_{l=1}^3 \pi_l^* \bar{\beta}_l}{\sum_{l=1}^3 \pi_l^*}$ and $\alpha_2 = \frac{\sum_{l=1}^3 (\pi-\pi_l^*) \bar{\beta}_l}{\sum_{l=1}^3 (\pi-\pi_l^*)}$. Suppose $\pi_1^* \in (0,\pi_1)$, then we have
	\begin{align*}
	||\bar{\beta}_1 - \alpha_1|| = ||\bar{\beta}_1 - \alpha_2||,
	\end{align*}
	which implies that, for any $\tilde{\pi} \in (0,\pi)$,
	\begin{align*}
	F(\alpha_1,\alpha_2;\pi^*_1,\pi^*_2,\pi^*_3) = F(\alpha_1,\alpha_2;\tilde{\pi},\pi^*_2,\pi^*_3) \geq F(\tilde{\alpha}_1,\tilde{\alpha}_2;\tilde{\pi},\pi^*_2,\pi^*_3),
	\end{align*}
	where $\tilde{\alpha}_1 = \frac{\tilde{\pi}_1 \bar{\beta}_1 + \pi^*_2 \bar{\beta}_2+\pi^*_3 \bar{\beta}_3}{ \tilde{\pi}_1+\pi_2^*+\pi_3^*}$ and $\tilde{\alpha}_2 = \frac{(\pi_1-\tilde{\pi}_1) \bar{\beta}_1 + (\pi_2-\pi^*_2) \bar{\beta}_2+ (\pi_3-\pi^*_3) \bar{\beta}_3}{1- \tilde{\pi}_1-\pi_2^*-\pi_3^*}$ are the minimizer of $F(\cdot,\cdot;\tilde{\pi},\pi^*_2,\pi^*_3)$. In addition, because $F(\alpha_1,\alpha_2;\pi^*_1,\pi^*_2,\pi^*_3)$ achieves the minimum of the k-means objective function among all classifications, we have
	\begin{align*}
	F(\alpha_1,\alpha_2;\pi^*_1,\pi^*_2,\pi^*_3) \leq  F(\tilde{\alpha}_1,\tilde{\alpha}_2;\tilde{\pi},\pi^*_2,\pi^*_3),
	\end{align*}
	which implies that the equality holds, for any $\tilde{\pi}_1 \in (0,\pi_1)$. Then, by the uniqueness of the minimizer for the quadratic objective function $F(\cdot,\cdot;\tilde{\pi},\pi^*_2,\pi^*_3)$, we have, for any $\tilde{\pi} \in (0,\pi_1)$,
	\begin{align*}
	(\alpha_1,\alpha_2) = (\tilde{\alpha}_1,\tilde{\alpha}_2).
	\end{align*}
	This implies that $\bar{\beta}_1 = \frac{\pi_2^* \bar{\beta}_2+\pi_3^* \bar{\beta}_3}{\pi_2^* + \pi_3^*} = \frac{(\pi_2 - \pi_2^*) \bar{\beta}_2+(\pi_3-\pi_3^*) \bar{\beta}_3}{\pi_2 - \pi_2^* + (\pi_3-\pi_3^*)} = \frac{\pi_2 \bar{\beta}_2+\pi_3 \bar{\beta}_3}{\pi_2 + \pi_3}$. Plugging this equality into \eqref{eq:suppose}, we have
	\begin{align*}
	\frac{\pi_2\pi_3}{\pi_2+\pi_3}||\bar{\beta}_2-\bar{\beta}_3||^2<\frac{\pi_1\pi_2}{\pi_1 + \pi_2}||\bar{\beta}_1-\bar{\beta}_2||^2 = \left(\frac{\pi_1}{\pi_1+\pi_2}\right) \left(\frac{\pi_3}{\pi_2+\pi_3}\right)\left(\frac{\pi_2\pi_3}{\pi_2+\pi_3}||\bar{\beta}_2-\bar{\beta}_3||^2\right),
	\end{align*}
	which is a contradiction. This implies that $\pi_1^* = 0$ or $\pi_1$. Similarly, if $\pi_2^* \in (0,\pi_2)$, we can show that $\bar{\beta}_2 = \frac{\pi_1 \bar{\beta}_1 + \pi_3 \bar{\beta}_3}{\pi_1+\pi_3}$. Then, by \eqref{eq:suppose},
	\begin{align*}
	\frac{\pi_1\pi_3}{\pi_1+\pi_3}||\bar{\beta}_1-\bar{\beta}_3||^2<\frac{\pi_1\pi_2}{\pi_1 + \pi_2}||\bar{\beta}_1-\bar{\beta}_2||^2 = \left(\frac{\pi_3}{\pi_1+\pi_2}\right) \left(\frac{\pi_2}{\pi_2+\pi_3}\right)\left(\frac{\pi_1\pi_3}{\pi_1+\pi_3}||\bar{\beta}_1-\bar{\beta}_3||^2\right),
	\end{align*}
	which is again a contradiction. Therefore, $\pi_2^* = 0$ or $\pi_2$. This means, $\mathcal{C}_k \subset \mathcal{C}(1) ~\text{or}~\mathcal{C}(2)$, for $k=1,2$. Last, we assume the k-means algorithm classify $\pi_3^*$ fraction of $\mathcal{C}_3$ with $\mathcal{C}_1$ and the rest with $\mathcal{C}_2$. Then, the k-means objective function becomes
	\begin{align*}
	\min_{\alpha_1,\alpha_2}F(\alpha_1,\alpha_2;\pi_1,\pi_2,\pi_3^*) = \frac{\pi_1\pi_3^*}{\pi_1 + \pi_3^*}||\bar{\beta}_1 - \bar{\beta}_3||^2 + \frac{\pi_2(\pi_3-\pi_3^*)}{\pi_2 + \pi_3-\pi_3^*}||\bar{\beta}_2 - \bar{\beta}_3||^2.
	\end{align*}
	When $\pi_3^* = 0$, the above display becomes $\frac{\pi_2\pi_3}{\pi_2+\pi_3}||\bar{\beta}_2 - \bar{\beta}_3||^2$. In addition,
	\begin{align*}
	& \left(\frac{\pi_1\pi_3^*}{\pi_1 + \pi_3^*}||\bar{\beta}_1 - \bar{\beta}_3||^2 + \frac{\pi_2(\pi_3-\pi_3^*)}{\pi_2 + \pi_3-\pi_3^*}||\bar{\beta}_2 - \bar{\beta}_3||^2 \right) - \left(\frac{\pi_2\pi_3}{\pi_2+\pi_3}||\bar{\beta}_2 - \bar{\beta}_3||^2 \right) \\
	= & \pi_3^*\left(\frac{\pi_1}{\pi_1 + \pi_3^*}||\bar{\beta}_1 - \bar{\beta}_3||^2 - \frac{\pi_2^2}{(\pi_2+\pi_3)(\pi_2+\pi_3 - \pi_3^*)}||\bar{\beta}_2 - \bar{\beta}_3||^2 \right) \\
	\geq & \pi_3^*\left(\frac{\pi_1}{\pi_1 + \pi_3}||\bar{\beta}_1 - \bar{\beta}_3||^2 - \frac{\pi_2}{(\pi_2+\pi_3)}||\bar{\beta}_2 - \bar{\beta}_3||^2 \right) \geq 0,
	\end{align*}
	where the first inequality holds because the term in the parenthesis after the first equal sign is a decreasing function in $\pi_3^* \in [0,\pi_3]$ and the last inequality holds because of \eqref{eq:suppose}. This implies that $\pi_3^* = 0$, i.e., $\mathcal{C}(1) = \mathcal{C}_1$ and $\mathcal{C}(2) = \mathcal{C}_2 \cup \mathcal{C}_3$, which implies Lemma \ref{lem:sbsa}(1).
	
	If the three terms in \eqref{eq:suppose} take distinctive values, the above argument is valid after relabeling. If at least two terms take same values, then the k-means algorithm applying to $\{\beta_{in}\}_{i=1}^n$ do not have a unique solution. This situation has been ruled out by our assumption. \\
	
	\textbf{Step 2. Proof of Lemma \ref{lem:sbsa}(3).}
	Let $\mathcal{Q}_{n}(\mathcal{A})=\sum_{l=1}^{L}\min_{1\leq k\leq K}\Vert \bar{\beta}
	_{l}-\alpha _{k}\Vert ^{2}\pi _{k}$, $%
	\mathcal{A}\in \mathcal{M}=\{(\alpha _{1},\ldots ,\alpha _{K}):\sup_{1\leq
		k\leq K}\Vert \alpha _{k}\Vert \leq 2M\}$ for some constant $M$ independent
	of $n$, $g_i^0 = k$ if $i \in \mathcal{C}(k)$, and $R_{n}=\sup_{i}\Vert \hat{\beta}_{in}-\beta_{in}\Vert $. By the assumptions in Lemma \ref{lem:sbsa},
	\begin{equation}
	R_{n} \leq c_{1} \quad a.s.  \label{eq:Rn}
	\end{equation}%
	In addition,
	\begin{align*}
	\Vert \hat{\beta}_{in}-\alpha _{k}\Vert ^{2}& \geq \Vert \beta
	_{in}-\alpha _{k}\Vert ^{2} - 2|(\beta _{in}-\hat{\beta}%
	_{in})^{T}(\beta _{in}-\alpha _{l})| - \Vert \beta _{in}-%
	\hat{\beta}_{in}\Vert ^{2} \\
	& \geq \Vert \beta _{in}-\alpha _{k}\Vert ^{2} - 2\Vert \beta
	_{in}-\hat{\beta}_{in}\Vert\Vert \beta _{in}-\alpha
	_{k}\Vert -R_{n}^{2} \\
	& \geq \Vert \beta _{in}-\alpha _{k}\Vert ^{2} - 6MR_{n} - R_{n}^{2} \\
	& \geq \Vert \beta _{in}-\alpha _{k}\Vert ^{2} - 7MR_n,
	\end{align*}%
	where the third inequality follows the Cauchy-Schwarz inequality. Taking $\min_{1\leq k\leq K}$ on both sides and averaging over $i$,
	we have
	\begin{align*}
	\widehat{\mathcal{Q}}_{n}(\mathcal{A}) \equiv & n^{-1}\sum_{i=1}^n \min_{1\leq k\leq K}||\hat{\beta}_{in} - \alpha_l||^2 \\
	\geq & n^{-1}\sum_{i=1}^n \min_{1\leq k\leq K}||\beta_{in} - \alpha_l||^2 -7MR_n \geq \mathcal{Q}_{n}(\mathcal{A})-7M c_{1},
	\end{align*}
	where the inequality is due to \eqref{eq:Rn}.
	Similarly, we have $\widehat{\mathcal{Q}}_{n}(\mathcal{A})\leq \mathcal{Q}_{n}(\mathcal{A})+7M c_{1}$, and thus,
	\begin{equation}
	\label{eq:Rnbreve}
	\breve{R}_{n}\equiv \sup_{\mathcal{A}\in \mathcal{M}}|\widehat{\mathcal{Q}}_{n}(%
	\mathcal{A})-\mathcal{Q}_{n}(\mathcal{A})| \leq 7M c_{1}\quad a.s.
	\end{equation}
	
	We maintain \eqref{eq:suppose}. In this case, the minimizer of $\mathcal{Q}_{n}(\cdot)$, as shown in the previous step, is $\mathcal{A}^* = (\alpha_1^*,\alpha_2^*)$, where $\alpha_1^* = \bar{\beta}_1$ and $\alpha_2^* = \frac{\pi_2 \bar{\beta}_2 + \pi_3 \bar{\beta}_3}{\pi_2 + \pi_3}$. Then, $Q_n(\mathcal{A}^*) = \frac{\pi_2 \pi_3}{\pi_2 + \pi_3}||\bar{\beta}_2 - \bar{\beta}_3||^2.$ For a generic $\mathcal{A} = (\alpha_1,\alpha_2)$ and $\mathcal{H}(\mathcal{A},\mathcal{A}^*) \geq \eta$, where $\mathcal{H}(\cdot,\cdot)$ denotes the Hausdorff distance of two sets, we aim to lower bound $\mathcal{Q}_{n}(\mathcal{A})- \mathcal{Q}_{n}(\mathcal{A}^*)$. In view of the definition of $Q_n(\cdot)$, we consider the following three cases: between $\alpha_1$ and $\alpha_2$,
	\begin{enumerate}
		\item[(1)] $\bar{\beta}_1$ is closer to $\alpha_1$ while $(\bar{\beta}_2,\bar{\beta}_3)$ are closer to $\alpha_2$;
		\item[(2)] $\bar{\beta}_2$ is closer to one of $\alpha_1$ while $(\bar{\beta}_1,\bar{\beta}_3)$ are closer to $\alpha_2$;
		\item[(3)] $\bar{\beta}_3$ is closer to one of $\alpha_1$ while $(\bar{\beta}_1,\bar{\beta}_2)$ are closer to $\alpha_2$;
		\item[(4)] $(\bar{\beta}_1,\bar{\beta}_2,\beta_3)$ are all closer to one of $\alpha_1$ and $\alpha_2$.
	\end{enumerate}
	For case (1),
	\begin{align*}
	\mathcal{Q}_{n}(\mathcal{A})- \mathcal{Q}_{n}(\mathcal{A}^*) =  & \pi_1||\bar{\beta}_1-\alpha_1||^2 + \sum_{l=2,3} \pi_l\left[||\bar{\beta}_l-\alpha_2||^2 - ||\bar{\beta}_l-\alpha_2^*||^2\right] \\
	= & \pi_1 ||\alpha_1^* - \alpha_1||^2  + \sum_{l=2,3} \pi_l\left[2(\bar{\beta}_l-\alpha_2^*)^T(\alpha_2^* - \alpha_2) + ||\alpha_2 - \alpha_2^*||^2\right] \\
	= & \pi_1 ||\alpha_1^* - \alpha_1||^2 + (\pi_2+\pi_3)||\alpha_2-\alpha_2^*||^2 \\
	\geq & \underline{\pi} \max(||\alpha_1^* - \alpha_1||,||\alpha_2-\alpha_2^*||)^2 \geq \underline{\pi}\eta^2,
	\end{align*}
	where the third equality holds because $\alpha_2^* = \frac{\pi_2 \bar{\beta}_2 + \pi_3 \bar{\beta}_3}{\pi_2 + \pi_3}$, the first inequality holds because for arbitrary constants $a,b>0$, $a+b \geq \max(a,b)$, and the last inequality holds because,
	\begin{align*}
	\mathcal{H}(\mathcal{A},\mathcal{A}^*) =\max(\mathcal{H}_1(\mathcal{A},\mathcal{A}^*),\mathcal{H}_2(\mathcal{A},\mathcal{A}^*)),
	\end{align*}
	where
	\begin{align*}
	\mathcal{H}_1(\mathcal{A},\mathcal{A}^*) = & \max( \min(||\alpha_1^*-\alpha_1||,||\alpha_1^*-\alpha_2||),\min(||\alpha_2^*-\alpha_1||,||\alpha_2^*-\alpha_2||)) \\
	\leq & \max (||\alpha_1^*-\alpha_1||,||\alpha_2^*-\alpha_2||)
	\end{align*}
	and
	\begin{align*}
	\mathcal{H}_2(\mathcal{A},\mathcal{A}^*) = & \max( \min(||\alpha_1^*-\alpha_1||,||\alpha_1-\alpha_2^*||),\min(||\alpha_2-\alpha_1^*||,||\alpha_2^*-\alpha_2||)) \\
	\leq & \max (||\alpha_1^*-\alpha_1||,||\alpha_2^*-\alpha_2||).
	\end{align*}
	
	For case (2), we have
	\begin{align*}
	Q_n(\mathcal{A}) - Q_n(\mathcal{A}^*) \geq & \inf_{\alpha_2}\left(\pi_1||\bar{\beta}_1 - \alpha_2||^2 + \pi_3||\bar{\beta}_3 - \alpha_2||^2\right) - \frac{\pi_2\pi_3}{\pi_2+\pi_3}||\bar{\beta}_2 - \bar{\beta}_3||^2 \\
	\geq & \frac{\pi_1\pi_3}{\pi_1+\pi_3}||\bar{\beta}_1 - \bar{\beta}_3||^2 - \frac{\pi_2\pi_3}{\pi_2+\pi_3}||\bar{\beta}_2 - \bar{\beta}_3||^2 \geq \underline{M}>0.
	\end{align*}
	where
	\begin{align*}
	\underline{M} = \min\left(\frac{\pi_1\pi_3}{\pi_1+\pi_3}||\bar{\beta}_1 - \bar{\beta}_3||^2,\frac{\pi_1\pi_2}{\pi_1+\pi_2}||\bar{\beta}_1 - \bar{\beta}_2||^2\right) - \frac{\pi_2\pi_3}{\pi_2+\pi_3}||\bar{\beta}_2 - \bar{\beta}_3||^2
	\end{align*}
	and the last inequality holds by \eqref{eq:suppose}.
	
	Similarly, for case (3), we have
	\begin{align*}
	Q_n(\mathcal{A}) - Q_n(\mathcal{A}^*) \geq & \inf_{\alpha_2}\left(\pi_1||\bar{\beta}_1 - \alpha_2||^2 + \pi_2||\bar{\beta}_2 - \alpha_2||^2\right) - \frac{\pi_2\pi_3}{\pi_2+\pi_3}||\bar{\beta}_2 - \bar{\beta}_3||^2 \\
	\geq & \frac{\pi_1\pi_2}{\pi_1+\pi_2}||\bar{\beta}_1 - \bar{\beta}_2||^2 - \frac{\pi_2\pi_3}{\pi_2+\pi_3}||\bar{\beta}_2 - \bar{\beta}_3||^2 \geq \underline{M}>0.
	\end{align*}
	
	Last, for the same reason, for case (4),
	\begin{align}
	\label{eq:Qn}
	Q_n(\mathcal{A}) - Q_n(\mathcal{A}^*) \geq \underline{M}>0.
	\end{align}
	Therefore, we have
	\begin{align*}
	\inf_{\mathcal{H}(\mathcal{A},\mathcal{A}^*) \geq \eta }Q_n(\mathcal{A}) - Q_n(\mathcal{A}^*) \geq \min(\underline{\pi}\eta^2,\underline{M}).
	\end{align*}
	Further define $\hat{\mathcal{A}}_n = (\hat{\alpha}_1,\hat{\alpha}_2) = \argmin_{\mathcal{A}}\hat{Q}_n(\mathcal{A})$. Note $\hat{\alpha}_1$ and $\hat{\alpha}_2$ are weighted average of $\{\hat{\beta}_{in}\}_{i=1}^n$ and $\sup_i ||\hat{\beta}_{in}|| \leq M+c_1 \leq 2M.$ Therefore, by \eqref{eq:Rnbreve},
	\begin{align}
	\label{eq:Rnbreve1}
	|\hat{Q}_n(\hat{\mathcal{A}}_n) - Q_n(\hat{\mathcal{A}}_n)| \leq 7Mc_1, \quad a.s.
	\end{align}
	and
	\begin{align}
	\label{eq:Rnbreve2}
	|\hat{Q}_n(\mathcal{A}^*) - Q_n(\mathcal{A}^*)| \leq 7Mc_1, \quad a.s.
	\end{align}
	
	Then,
	\begin{align*}
	& P(\mathcal{H}(\hat{\mathcal{A}}_n, \mathcal{A}^*) \geq (15M/\underline{\pi})^{1/2}c_1^{1/2} \quad i.o.) \\
	= & P(\mathcal{H}(\hat{\mathcal{A}}_n, \mathcal{A}^*) \geq (15M/\underline{\pi})^{1/2}c_1^{1/2},~Q_n(\hat{\mathcal{A}}_n)-Q_n(\mathcal{A}^*) \geq \min(15Mc_1 ,\underline{M}) \quad i.o.) \\
	\leq & P(14Mc_1 + \hat{Q}_n(\hat{\mathcal{A}}_n)-\hat{Q}_n(\mathcal{A}^*) \geq \min(15Mc_1,\underline{M}) \quad i.o.) \\
	\leq & P(14Mc_1  \geq \min(15Mc_1,\underline{M}) \quad i.o.) \\
	= & 0,
	\end{align*}
	where the first equality holds due to \eqref{eq:Qn}, the first inequality holds because of \eqref{eq:Rnbreve1} and \eqref{eq:Rnbreve2}, the second inequality holds because $ \hat{Q}_n(\hat{\mathcal{A}}_n)-\hat{Q}_n(\mathcal{A}^*) \geq 0$, and the last equality holds because $c_1$ is sufficiently small so that $15Mc_1 \leq \underline{M}$. This implies
	\begin{align*}
	\mathcal{H}(\hat{\mathcal{A}}_n, \mathcal{A}^*) \leq (15M/\underline{\pi})^{1/2}c_1^{1/2}, \quad a.s.
	\end{align*}
	
	Further note that $||\alpha_1^* - \alpha_2^*||>0$, otherwise $\bar{\beta}_1 = \frac{\pi_2 \bar{\beta}_2 + \pi_3 \bar{\beta}_3}{\pi_2 + \pi_3}$ which is a contradiction as shown in Step 1. Let $c_1$ be sufficiently small so that $(15M/\underline{\pi})^{1/2}c_1^{1/2}< ||\alpha_1^* - \alpha_2^*||$. Then there is a one-to-one mapping $\mathcal{F}_n$: $\{1,2\} \mapsto \{1,2\}$ such that
	\begin{align*}
	\sup_{k=1,2}||\hat{\alpha}_k - \alpha_{\mathcal{F}_n(k)}^*|| \leq (15M/\underline{\pi})^{1/2}c_1^{1/2}.
	\end{align*}
	W.l.o.g., we assume $\mathcal{F}_n(k) = k$ such that
	\begin{align*}
	\sup_{k=1,2}||\hat{\alpha}_k - \alpha_{k}^*|| \leq (15M/\underline{\pi})^{1/2}c_1^{1/2}.
	\end{align*}
	Denote $\hat{g}_i = k$ if $i \in \widehat{\mathcal{C}}(k)$, $k=1,2$ and  $g_i^0 = k$ if $i \in \mathcal{C}(k)$, $k=1,2$. If $\hat{g}_i \neq g_i^0$, then $||\hat{\beta}_{in} - \hat{\alpha}_{\hat{g}_i}|| \leq ||\hat{\beta}_{in} - \hat{\alpha}_{g^0_i}||$. Therefore,
	\begin{align*}
	& ||\beta_{in} - \alpha_{g^0_i}|| + c_1 +(15M/\underline{\pi})^{1/2}c_1^{1/2} \\
	\geq & ||\hat{\beta}_{in} - \hat{\alpha}_{g^0_i}|| \\
	\geq & ||\hat{\beta}_{in} - \hat{\alpha}_{\hat{g}_i}|| \geq ||\beta_{in} - \alpha_{\hat{g}_i}^*|| - c_1 -(15M/\underline{\pi})^{1/2}c_1^{1/2}.
	\end{align*}
	Therefore,
	\begin{align*}
	1\{\hat{g}_i \neq g_i^0\} \leq & 1\{2c_1 +2(15M/\underline{\pi})^{1/2}c_1^{1/2}\geq  ||\beta_{in} - \alpha_{\hat{g}_i}^*|| - ||\beta_{in} - \alpha_{g^0_i}^*||\} \quad a.s.
	\end{align*}
	By Lemma \ref{lem:sbsa}(1), we only need to consider the lower bound for the RHS of the above display in three cases: (1) $g_i^0 = 1$ and $\beta_{in} = \bar{\beta}_1$, (2) $g_i^0 = 2$ and $\beta_{in} = \bar{\beta}_2$, and (3) $g_i^0 = 2$ and $\beta_{in} = \bar{\beta}_3$. For case (1),
	\begin{align*}
	||\beta_{in} - \alpha_{\hat{g}_i}^*|| - ||\beta_{in} - \alpha_{g^0_i}^*|| = ||\alpha_1^* - \alpha_2^*|| = \left\Vert \bar{\beta}_1 - \frac{\pi_2 \bar{\beta}_2 + \pi_3 \bar{\beta}_3}{\pi_2 + \pi_3}\right \Vert >0,
	\end{align*}
	where the last inequality holds because by the argument in Step 1, $\bar{\beta}_1 \neq \frac{\pi_2 \bar{\beta}_2 + \pi_3 \bar{\beta}_3}{\pi_2 + \pi_3}$.
	
	For case (2), $\alpha_{\hat{g}_i}^* = \alpha_{1}^* = \bar{\beta}_1$ and
	\begin{align*}
	||\beta_{in} - \alpha_{\hat{g}_i}^*|| - ||\beta_{in} - \alpha_{g^0_i}^*|| = & ||\bar{\beta}_2 - \bar{\beta}_1|| - \frac{\pi_3}{\pi_2 + \pi_3}||\bar{\beta}_2 - \bar{\beta}_3|| \\
	\geq & ||\bar{\beta}_2 - \bar{\beta}_3||\sqrt{\frac{\pi_3}{\pi_2+\pi_3}}\left(\sqrt{\frac{\pi_1+\pi_2}{\pi_1}} - \sqrt{\frac{\pi_3}{\pi_2+\pi_3}} \right)>0,
	\end{align*}
	where the first inequality holds due to \eqref{eq:suppose}. Similarly, for case (3), we have
	\begin{align*}
	||\beta_{in} - \alpha_{\hat{g}_i}^*|| - ||\beta_{in} - \alpha_{g^0_i}^*|| = & ||\bar{\beta}_3 - \bar{\beta}_1|| - \frac{\pi_2}{\pi_2 + \pi_3}||\bar{\beta}_2 - \bar{\beta}_3|| \\
	\geq & ||\bar{\beta}_2 - \bar{\beta}_3||\sqrt{\frac{\pi_2}{\pi_2+\pi_3}}\left(\sqrt{\frac{\pi_1+\pi_3}{\pi_1}} - \sqrt{\frac{\pi_2}{\pi_2+\pi_3}} \right)>0.
	\end{align*}
	Let constant $\underline{C}$ be
	\begin{align*}
	\min\left(\left\Vert \bar{\beta}_1 - \frac{\pi_2 \bar{\beta}_2 + \pi_3 \bar{\beta}_3}{\pi_2 + \pi_3}\right \Vert, ||\bar{\beta}_2 - \bar{\beta}_1|| - \frac{\pi_3}{\pi_2 + \pi_3}||\bar{\beta}_2 - \bar{\beta}_3||, ||\bar{\beta}_3 - \bar{\beta}_1|| - \frac{\pi_2}{\pi_2 + \pi_3}||\bar{\beta}_2 - \bar{\beta}_3||\right) \geq \underline{C}
	\end{align*}
	such that $\underline{C}>0$. Then,
	\begin{align*}
	1\{\hat{g}_i \neq g_i^0\} \leq & 1\{2c_1 +2(15M/\underline{\pi})^{1/2}c_1^{1/2}\geq  ||\beta_{in} - \alpha_{\hat{g}_i}^*|| - ||\beta_{in} - \alpha_{g^0_i}^*||\} \\
	\leq & 1\{2c_1 +2(15M/\underline{\pi})^{1/2}c_1^{1/2}\geq \underline{C}\}.
	\end{align*}
	Noting that the RHS of the above display is independent of $i$ and choosing $c_1$ sufficiently small such that
	\begin{align*}
	2c_1 +2(15M/\underline{\pi})^{1/2}c_1^{1/2}< \underline{C},
	\end{align*}
	we have
	\begin{align*}
	P(\sup_i 1\{\hat{g}_i \neq g_i^0\} >0,~i.o.) \leq P(2c_1 +2(15M/\underline{\pi})^{1/2}c_1^{1/2}\geq \underline{C},~ i.o.) = 0
	\end{align*}
	This concludes that $\widehat{\mathcal{C}}(k) = \mathcal{C}(k)$ for $k=1,2$, which is the desired result for Lemma \ref{lem:sbsa}(3).
	
	\textbf{Step 3. Proof of Lemma \ref{lem:sbsa}(2).} Given Lemma \ref{lem:sbsa}(3), the desired results can be derived by the same argument for $K=3$.
\end{proof}

\bibliographystyle{elsarticle-harv}
\bibliography{CD-1}

\end{document}